%% file: main.tex
\pdfoutput=1
\documentclass[acmsmall,screen]{acmart}
\usepackage{booktabs}   
\usepackage{subcaption} 
\usepackage{mathpartir}
\usepackage{makecell}
\usepackage{amsmath,mathtools,listings,amsthm,amssymb,xspace,color,supertabular}
\usepackage{empheq, soul}
\definecolor{lightgrey}{rgb}{.9,.9,.9}
\sethlcolor{lightgrey}
\usepackage{authoredcomments}
\usepackage{wrapfig}
\usepackage{tabularx}
\usepackage{ottlayout}
\ottstyledefaults{showruleschema=no,showcomment=no,rulelayout=nobreaks,premisenamelayout=justify,numberpremises=no,premisenamelayout=topright}
\usepackage{tikz}
\usetikzlibrary{arrows}
\usetikzlibrary{cd} 
\usetikzlibrary{decorations.pathmorphing}

\acmPrice{}
\acmDOI{10.1145/3158137}
\acmYear{2018}
\copyrightyear{2018}
\setcopyright{rightsretained}
\acmJournal{PACMPL}
\acmYear{2018} \acmVolume{2} \acmNumber{POPL} \acmArticle{49} \acmMonth{1} \acmPrice{}\acmDOI{10.1145/3158137}
\startPage{1}

\UNXXX{}

\input{sourir}

\renewcommand{\ottpremise}[1]{\premiseSTY{#1}}
\renewcommand{\ottusedrule}[1]{\usedruleSTY{#1}}
\renewcommand{\ottdrule}[4][]{\druleSTY[#1]{#2}{#3}{#4}}
\renewenvironment{ottdefnblock}[3][]{\defnblockSTY[#1]{#2}{#3}}{\enddefnblockSTY}
\newcommand{\code}[1]{\textsf{#1}}
\newcommand{\sourir}{\textsf{sourir}\xspace}
\newcommand{\Sourir}{\textsf{Sourir}\xspace}
\newcommand{\assume}{\code{assume}\xspace}

\newcommand{\tto}[1][\phantom{\tau}]{
  \overset{#1\phantom{*}}{\longrightarrow^{\smash{\mkern-2mu*}}}
}
\newcommand{\nto}[1][\phantom{\tau}]{
  \overset{#1}{\longrightarrow}
}

\tikzset{
    rar/.style={
        shorten >=.0em,#1-to,
        to path={-- (\tikztotarget) \tikztonodes}
    },
    rar/.default=
}
\tikzset{
    to*/.style={
        shorten >=.25em,#1-to,
        to path={-- node[inner sep=0pt,at end,sloped] {${}^*$} (\tikztotarget) \tikztonodes}
    },
    to*/.default=
}

\newcommand{\textsfb}[1]{\textsf{\textbf{#1}}}
\newcommand{\eg}{\emph{e.g.}\xspace}
\newcommand{\ie}{\emph{i.e.}\xspace}

\begin{document}

\title[Correctness of Speculative Optimizations]{Correctness of Speculative Optimizations \\ with Dynamic Deoptimization} 
\author[Fl\"uckiger]{Olivier Fl\"uckiger}
\affiliation{\institution{Northeastern University} \country{USA} }
\author[Scherer]{Gabriel Scherer}
\affiliation{   \institution{Northeastern University}  \country{USA}}
\affiliation{   \institution{INRIA}  \country{France}}
\author[Yee]{Ming-Ho Yee}
\author[Goel]{Aviral Goel}
\author[Ahmed]{Amal Ahmed}
\affiliation{   \institution{Northeastern University}  \country{USA} }
\author[Vitek]{Jan Vitek}
\affiliation{   \institution{Northeastern University} \country{USA} }
\affiliation{   \institution{CVUT} \country{Czech Republic} }

\authorsaddresses{}

\begin{abstract}
High-performance dynamic language implementations make heavy use of
speculative optimizations to achieve speeds close to statically compiled
languages. These optimizations are typically performed by a just-in-time
compiler that generates code under a set of assumptions about the state of
the program and its environment. In certain cases, a program may execute
code compiled under assumptions that are no longer valid. The implementation
must then deoptimize the program on-the-fly; this entails finding
semantically equivalent code that does not rely on invalid
assumptions, translating program state to that expected by the target code,
and transferring control.  This paper looks at the interaction between
optimization and deoptimization, and shows that reasoning about speculation
is surprisingly easy when assumptions are made explicit in the program
representation. This insight is demonstrated on a compiler intermediate
representation, named \sourir, modeled after the high-level representation
for a dynamic language. Traditional compiler optimizations such as constant
folding, unreachable code elimination, and function inlining are shown to be correct
in the presence of assumptions. Furthermore, the paper establishes the
correctness of compiler transformations specific to deoptimization: namely
unrestricted deoptimization, predicate hoisting, and assume composition.
\end{abstract}


\maketitle

\section{Introduction}
Dynamic languages pose unique challenges to compiler writers.  With features
such as dynamic binding, runtime code generation, and generalized reflection,
languages such as Java, C\#, Python, JavaScript, R, or Lisp force
implementers to postpone code generation until the last possible
instant. The intuition being that just-in-time (JIT) compilation can
leverage information about the program state and its environment,
\eg, the value of program inputs or which libraries were loaded, to
generate efficient code and potentially update code on-the-fly.

Many dynamic language compilers support some form of \emph{speculative}
optimization to avoid generating code for unlikely control-flow paths. In a
dynamic language prevalent polymorphism causes even the simplest code to
have non-trivial control flow. Consider the JavaScript snippet
in \autoref{fig:js-ex} (example from~\citet{Spur10}).  Without optimization one
iteration of the loop executes 210 instructions; all arithmetic operations
are dispatched and their results boxed.  If the compiler is allowed
to make
\begin{wrapfigure}{r}{44mm}
\vspace{.2cm}
\centering
\small\sf
\parbox{0cm}{\begin{tabbing}
{\bf for} \=(i=0; i < a.length-1; i++) \{\\[-1mm]
\>  {\bf var} t=a[i];\\[-1mm]
\>  a[i]=a[i+1];\\[-1mm]
\>  a[i+1]=t;\\[-1mm]
\}
\end{tabbing}}
\caption{JavaScript rotate function.}
\label{fig:js-ex}  \vspace{-.1cm}
\end{wrapfigure} 
\kern -1mm
the assumption it is operating on integers, the body of the loop shrinks
down to 13 instructions.
As another example, most Java implementations assume that non-final methods are not
overridden. Speculating on this fact allows compilers to avoid emitting
dispatch code~\citep{ish00}. Newly loaded classes are monitored, and any time a method is overridden,
the virtual machine invalidates code that contains devirtualized calls to that method.
The validity of speculations is expressed as a predicate on the program state.
If some program action, like loading a new class, falsifies that predicate, the generated code must be
discarded.
To undo an assumption, an implementation must ensure that functions compiled
under that assumption are retired. This entails replacing affected code with
a version that does not depend on the invalid predicate and, if a function
currently being executed is found to contain invalid code, that function
needs to be replaced on-the-fly. In such a case, it is necessary to transfer
control to a different version of the function, and in the process, it may
be necessary to materialize portions of the state that were optimized
away and perform other recovery actions. In particular, if the invalid
function was inlined into another function, it is necessary to synthesize a
new stack frame for the caller.  This is referred to as
\emph{deoptimization}, or \emph{on-stack-replacement}, and is found in most
industrial-strength compilers.

Speculative optimization gives rise to a large and multi-dimensional design
space that lies mostly unexplored.  First, compiler writers must decide how
to obtain information about program state. This can be done ahead-of-time by
profiling, just-in-time by sampling or instrumenting code.  Next, they
must select what facts to record. This can range from information about the
program, its class hierarchy, which packages were loaded, to information
about the value of a particular mutable location in the heap. Finally, they
must decide how to efficiently monitor the validity of
speculations.  While some points in this space have been explored
empirically, existing systems have done it in an \emph{ad hoc} manner that
is often both language- and implementation-specific, and thus difficult to
apply broadly.

This paper has a focused goal. We aim to demystify the interaction
between compiler transformations and deoptimization.  When are two versions
compiled under different assumptions equivalent?  How should traditional
optimizations be adapted when operating on code containing deoptimization
points?  In what ways does deoptimization inhibit optimizations?  In this
work we give compiler writers the formal tools they need to reason about
speculative optimizations.  To do this in a way that is independent of the
specific language being targeted and of implementation details relative to a
particular compiler infrastructure, we have designed a high-level compiler
intermediate representation (IR), named \sourir, that is adequate for many
dynamic languages without being tied to any one in particular.

\Sourir is inspired by our work on RIR, an IR for the R language.  A \sourir
program is made up of functions, and each function can have multiple
versions.  We equip the IR with a single instruction, named $ \textsfb{assume} $,
specific to speculative optimization.  This instruction has the role of
describing what assumptions are being used to perform speculative
optimization and what information must be preserved for deoptimization. It
tests if those assumptions hold, and in case they do not, transfers control
to another, less optimized version of the code.  Reifying assumptions in
the IR makes the interaction with compiler transformations explicit and
simplifies reasoning.  The \assume instruction is more than a branch: when
deoptimizing it replaces the current stack frame with a stack frame that has
the variables and values expected by the target version, and, in case the
function was inlined, it synthesizes missing stack frames. Furthermore,
unlike a branch, its deoptimization target is not followed by the compiler during analysis and
optimization.  The code executed in case of deoptimization is invisible to
the optimizer. This simplifies optimizations and reduces compile time as
\begin{wrapfigure}{r}{7.9cm}
\( \\  {\small
                                                          \begin{array}{l}
                                                              \scalebox{0.97}{$\mathsf{ rot }$}   (  \,  \kern 0.04em ) \\
                                                                   \begin{array}{llll}
                                                                     \kern 1.5pt   \hspace{2.5mm}  \mathsf{V\kern -0.25pt \scalebox{0.87}{$\mathsf{ native }$} }   \\
                                                                             \begin{array}{l!{\,\color{gray}\vrule}lll}
                                                                               &      & \dots \\  \,  &      &   \textsfb{call} ~  \mathsf{ type }  \nobreak\hspace{0pt}=\nobreak\hspace{0pt}  \mathsf{ typeof }  (  \mathsf{ a }  )   \\  \,  &      &  \textsfb{assume} ~   \mathsf{ type }    =    \mathsf{ NumArray }   ~ \ottkw{else} ~    \scalebox{0.97}{$\mathsf{ rot }$}  .\kern -0.5pt  \mathsf{V\kern -0.25pt \scalebox{0.87}{$\mathsf{ base }$} }  .\kern -0.5pt  \mathsf{L\kern -0.25pt \scalebox{0.87}{$\mathsf{ t }$} }   ~  [{ \,  \kern 0.033em}]    \\  \,  &   \mathsf{L\kern -0.25pt \scalebox{0.87}{$\mathsf{ t }$} }   &   \textsfb{branch} ~   \mathsf{i}    <    \mathsf{ limit }   ~  \mathsf{L\kern -0.25pt \scalebox{0.87}{$\mathsf{ o }$} }  ~  \mathsf{L\kern -0.25pt \scalebox{0.87}{$\mathsf{ rt }$} }    \\  \,  &   \mathsf{L\kern -0.25pt \scalebox{0.87}{$\mathsf{ o }$} }   &   \textsfb{var} ~  \mathsf{ t }  \nobreak\hspace{0pt}=\nobreak\hspace{0pt}   \mathsf{ a }  [{  \mathsf{i}  }]    \\  \,  &      &  \textsfb{assume} ~   \mathsf{ t }    \neq    \mathsf{ HL }   ~ \ottkw{else} ~    \scalebox{0.97}{$\mathsf{ rot }$}  .\kern -0.5pt  \mathsf{V\kern -0.25pt \scalebox{0.87}{$\mathsf{ base }$} }  .\kern -0.5pt  \mathsf{L\kern -0.25pt \scalebox{0.87}{$\mathsf{ s }$} }   ~  [{  \mathsf{i}   \ottsym{=}   \mathsf{i}   \ottsym{,}   \mathsf{ j }   \ottsym{=}    \mathsf{i}    \ottsym{+}   1   \kern 0.033em}]    \\  \,  &      &    \mathsf{ a }  [{  \mathsf{i}   \kern 0.05em}]\nobreak\hspace{0pt} \leftarrow \nobreak\hspace{0pt}   \mathsf{ a }  [{   \mathsf{i}    \ottsym{+}   1  }]    \\  \,  &      &    \mathsf{ a }  [{   \mathsf{i}    \ottsym{+}   1   \kern 0.05em}]\nobreak\hspace{0pt} \leftarrow \nobreak\hspace{0pt}  \mathsf{ t }    \\  \,  &      &    \mathsf{i}  \nobreak\hspace{0pt} \leftarrow \nobreak\hspace{0pt}   \mathsf{i}    \ottsym{+}   1    \\  \,  &      &  \textsfb{goto} \,  \mathsf{L\kern -0.25pt \scalebox{0.87}{$\mathsf{ t }$} }   \\  \,  &   \mathsf{L\kern -0.25pt \scalebox{0.87}{$\mathsf{ rt }$} }   & \dots \\  
                                                                             \end{array} \vspace{0.25em} \\  \,  \hspace{2.5mm}  \mathsf{V\kern -0.25pt \scalebox{0.87}{$\mathsf{ base }$} }   \\
                                                                             \begin{array}{l!{\,\color{gray}\vrule}lll}
                                                                               &      & \dots \\  \,  &   \mathsf{L\kern -0.25pt \scalebox{0.87}{$\mathsf{ t }$} }   &   \textsfb{branch} ~   \mathsf{i}    <    \mathsf{ limit }   ~  \mathsf{L\kern -0.25pt \scalebox{0.87}{$\mathsf{ o }$} }  ~  \mathsf{L\kern -0.25pt \scalebox{0.87}{$\mathsf{ rt }$} }    \\  \,  &   \mathsf{L\kern -0.25pt \scalebox{0.87}{$\mathsf{ o }$} }   &   \textsfb{call} ~  \mathsf{ j }  \nobreak\hspace{0pt}=\nobreak\hspace{0pt}  \mathsf{ add }  (  \mathsf{i}   \ottsym{,}  1 )   \\  \,  &   \mathsf{L\kern -0.25pt \scalebox{0.87}{$\mathsf{ s }$} }   &   \textsfb{call} ~  \mathsf{ t1 }  \nobreak\hspace{0pt}=\nobreak\hspace{0pt}  \mathsf{ get }  (  \mathsf{ a }   \ottsym{,}   \mathsf{i}  )   \\  \,  &      &   \textsfb{call} ~  \mathsf{ t2 }  \nobreak\hspace{0pt}=\nobreak\hspace{0pt}  \mathsf{ get }  (  \mathsf{ a }   \ottsym{,}   \mathsf{ j }  )   \\  \,  &      &   \textsfb{call} ~  \mathsf{ t3 }  \nobreak\hspace{0pt}=\nobreak\hspace{0pt}  \mathsf{ store }  (  \mathsf{ a }   \ottsym{,}   \mathsf{i}   \ottsym{,}   \mathsf{ t2 }  )   \\  \,  &      &   \textsfb{call} ~  \mathsf{ t4 }  \nobreak\hspace{0pt}=\nobreak\hspace{0pt}  \mathsf{ store }  (  \mathsf{ a }   \ottsym{,}   \mathsf{ j }   \ottsym{,}   \mathsf{ t1 }  )   \\  \,  &      &    \mathsf{i}  \nobreak\hspace{0pt} \leftarrow \nobreak\hspace{0pt}  \mathsf{ j }    \\  \,  &      &  \textsfb{goto} \,  \mathsf{L\kern -0.25pt \scalebox{0.87}{$\mathsf{ t }$} }   \\  \,  &   \mathsf{L\kern -0.25pt \scalebox{0.87}{$\mathsf{ rt }$} }   & \dots \\  
                                                                             \end{array} \vspace{0.25em} \\  
                                                                   \end{array} \\  
                                                          \end{array} }  \\ \)
  \caption{Compiled function from \autoref{fig:js-ex}.}\label{fig:first-ex}\vspace{-.1cm}
\vspace{-.1cm}
\end{wrapfigure} %
analysis remains local to the version
being optimized and the deoptimization
metadata is considered to be a stand-in for the target version.

As an example consider the function from \autoref{fig:js-ex}.  A possible
translation to \sourir is shown in \autoref{fig:first-ex} (less relevant
code elided).  $  \mathsf{V\kern -0.25pt \scalebox{0.87}{$\mathsf{ base }$} }  $ contains the original version.  Helper
functions $  \scalebox{0.97}{$\mathsf{ get }$}  $ and $  \scalebox{0.97}{$\mathsf{ store }$}  $ implement JavaScript
(JS) array semantics, and the function $  \scalebox{0.97}{$\mathsf{ add }$}  $ implement JS
addition.  Version~$  \mathsf{V\kern -0.25pt \scalebox{0.87}{$\mathsf{ native }$} }  $ contains only primitive \sourir
instructions.  This version is optimized under the assumption that the
variable $  \mathsf{ a }  $ is an array of primitive numbers, which is
represented by the first \assume instruction.  Further, JS arrays can be
sparse and contain holes, in which case access might need to be delegated to
a getter function.  For this example $  \mathsf{ HL }  $ denotes such a
hole.  The second \assume instruction reifies the compiler's speculation
that the array has no holes, by asserting the predicate $  \mathsf{ t }    \neq    \mathsf{ HL }  $.  It
also contains the associated deoptimization metadata.  In case the predicate
does not hold, we deoptimize to a related position in the base version by
recreating the variables in the target scope.  As can be seen in the second
\assume, local variables are mapped as $[  \mathsf{i}   \ottsym{=}   \mathsf{i}   \ottsym{,}   \mathsf{ j }   \ottsym{=}    \mathsf{i}    \ottsym{+}   1  ]$; the
current value of $  \mathsf{i}  $ is carried over into the target frame's
$  \mathsf{i}  $, whereas variable $  \mathsf{ j }  $ has to be recomputed.

We prove the correctness of a selection of traditional compiler
optimizations in the presence of speculation; these are constant
propagation, unreachable code elimination, and function inlining.  The main
challenge for correctness is that the transformations operate on one version
in isolation and therefore only see a subset of all possible control flows.
We show how to split the work to prove correctness between the pass that
establishes a version-to-version correspondence and the actual
optimizations.  Furthermore, we introduce and prove the correctness of three optimizations
specific to speculation, namely unrestricted deoptimization,
predicate hoisting, and assume composition.

Our work makes several simplifying assumptions.  We use the same IR for
optimized and unoptimized code.  We ignore the issue of generation of
versions: we study optimizations operating on a program at a certain point
in time, on a set of versions created before that time.  We do not model the
low-level details of code generation. Correctness of runtime code generation
and code modification within a JIT compiler has been addressed by
\citet{myr10}.  \Sourir is not designed for implementation, but to give a
reasoning model for existing JIT implementations.  We do not intend to
implement a new JIT engine.  Instead, we evaluated our work by discussing it
with JIT implementers; the V8 team~\citep{v8} confirmed that intuitions and
correctness arguments could be ported from \sourir to their
setting.

\section{Related Work}\label{sec:related}

The SELF virtual machine pioneered dynamic deoptimization~\citep{hol92}. The
SELF compiler implemented many optimizations, one of which was aggressive
inlining, yet the language designers wanted to give end users the illusion
that they were debugging source code.  They achieved this by replacing
optimized code and the corresponding stack frames with non-optimized code
and matching stack frames. When deoptimizing code that had been inlined, the
SELF compiler synthesized stack frames.  The HotSpot compiler followed from
the work on SELF by introducing the idea of speculative
optimizations~\citep{pal01}. HotSpot supported very specific assumptions
related to the structure of the class hierarchy and instrumented the class
loader to trigger invalidation.  When an invalidation occurred affected
functions were rolled forward to a safe point and control was transferred
from native code to the interpreter. The Jikes RVM adopted these ideas
to avoid compiling uncommon code paths~\citep{fin03}.

One drawback of the early work was that deoptimization points were barriers
around which optimizations were not allowed. \citet{oda05} were the first to
investigate exception reordering by hoisting guards.  They remarked that
checking assumptions early might improve code.  In
\citet{som06} the optimizer is allowed to update the deoptimization
metadata.  In particular they support eliding duplicate variables in the
mapping and lazily reconstructing values when transferring control.  This
unlocks further optimizations, which were blocked in previous work. The
paper also introduces the idea of being able to transfer control at any
point.  We support both the update of metadata and unconstrained
deoptimization.

Modern virtual machines have all incorporated some degree of speculation and
support for deoptimization.  These include implementations of Java (HotSpot,
Jikes RVM), JavaScript (WebKit Core, Chromium V8, Truffle/JS, Firefox), Ruby
(Truffle/Ruby), and R (FastR), among others.  Anecdotal evidence suggests
that the representation adopted in this work is representative of the
instructions found in the IR of production VMs: the TurboFan IR from
V8~\citep{v8} represents \assume with three distinct nodes. First a
\emph{checkpoint}, holding the deoptimization target, marks a stable point,
to where execution can be rolled back. In \sourir this corresponds to the
original location of an \assume.  A \emph{framestate} node records the
layout of, and changes to, the local frame, roughly the varmap in \sourir.
Assumption predicates are guarded by conditional deoptimization nodes, such
as \emph{deoptimizeIf}.  Graal~\citep{dub13} also has an explicit
representation for assumptions and associated metadata as \emph{guard} and
\emph{framestate} nodes in their high-level IR. In both cases guards are
associated with the closest dominating checkpoint.

Lowering deoptimization metadata is described in~\citet{schn12,dub14}.  A
detailed empirical evaluation of deoptimization appears in~\citet{ecoop17}.
The implementation of control-flow transfer is not modeled here as it is not
relevant to our results. For one particular implementation, we refer readers
to \citet{eli16} which builds on LLVM.  Alternatively, \citet{wan15} propose
an IR that supports restricted primitives for hot-patching code in a JIT.

There is a rich literature on formalizing compiler optimizations. The
CompCert project~\citep{ler08} for example implements many optimizations,
and contains detailed proof arguments for a data-flow optimization used for
constant folding that is similar to ours. In fact, \sourir is close to
CompCert's RTL language without versions or assumptions.  There are
formalizations for tracing compilers~\citep{guo11,dis14}, but we are unaware
of any other formalization effort for speculative optimizations in general.
\citet{bera16} present a verifier for a bytecode-to-bytecode optimizer.  By
symbolically executing optimized and unoptimized code, they verify that the
deoptimization metadata produced by their optimizer correctly maps the
symbolic values of the former to the latter at all deoptimization points.

\section{Sourir: Speculative Compilation Under Assumptions}\label{sec:sourir-intro}

This section introduces our IR and its design principles.  We first present
the structure of programs and the \assume instruction.  Then,
\autoref{subsec:invariants} and following explain how \sourir maintains
multiple equivalent versions of the same function, each with a different set
of assumptions.  This enables the speculative optimizations presented in
\autoref{sec:optimizations}.  All concepts introduced in this section are
formalized in \autoref{sec:sourir-intro-formal}.

\subsection{Sourir in a Nutshell} \label{sec:nutshell}

\begin{wrapfigure}{r}{5.0cm}
\vspace{-.4cm}
  \center
  \(  {\small \begin{array}{lll}
                                                      &      &   \textsfb{var} ~  \mathsf{n}  \nobreak\hspace{0pt}=\nobreak\hspace{0pt} \textsfb{nil}   \\  \,  &      &  \textsfb{read} \,  \mathsf{n}   \\  \,  &      &   \textsfb{array} ~  \mathsf{ t }  [{  \mathsf{n}   \kern 0.06em}]   \\  \,  &      &   \textsfb{var} ~  \mathsf{ k }  \nobreak\hspace{0pt}=\nobreak\hspace{0pt} 0   \\  \,  &      &  \textsfb{goto} \,  \mathsf{L\kern -0.25pt \scalebox{0.83}{$\mathsf{ 1 }$} }   \\  \,  &   \mathsf{L\kern -0.25pt \scalebox{0.83}{$\mathsf{ 1 }$} }   &   \textsfb{branch} ~   \mathsf{ k }    <    \mathsf{n}   ~  \mathsf{L\kern -0.25pt \scalebox{0.83}{$\mathsf{ 2 }$} }  ~  \mathsf{L\kern -0.25pt \scalebox{0.83}{$\mathsf{ 3 }$} }    \\  \,  &   \mathsf{L\kern -0.25pt \scalebox{0.83}{$\mathsf{ 2 }$} }   &    \mathsf{ t }  [{  \mathsf{ k }   \kern 0.05em}]\nobreak\hspace{0pt} \leftarrow \nobreak\hspace{0pt}  \mathsf{ k }    \\  \,  &      &    \mathsf{ k }  \nobreak\hspace{0pt} \leftarrow \nobreak\hspace{0pt}   \mathsf{ k }    \ottsym{+}   1    \\  \,  &      &  \textsfb{goto} \,  \mathsf{L\kern -0.25pt \scalebox{0.83}{$\mathsf{ 1 }$} }   \\  \,  &   \mathsf{L\kern -0.25pt \scalebox{0.83}{$\mathsf{ 3 }$} }   &  \textsfb{drop} \,  \mathsf{ k }   \\  \,  &      &  \textsfb{stop}  \\  
                                                \end{array} }  \)
  \caption{Example \sourir code.}
  \label{fig:std-instr-example}
\end{wrapfigure} %
\Sourir is an untyped language with lexically scoped mutable variables and
first-class functions.
As an example the function
in \autoref{fig:std-instr-example} queries a number~$  \mathsf{n}  $ from the user and
initializes an array with values from \code{0} to \code{n-1}. By design, \sourir is a
cross between a compiler representation and a high-level language. We have
equipped it with sufficient expressive power so that it is possible to write
interesting programs in a style reminiscent of dynamic
languages.\footnote{An implementation of \sourir and the optimizations
  presented here is available at \url{https://github.com/reactorlabs/sourir}.} The only features that are
critical to our result are \emph{versions} and \emph{assumptions}.  Versions
are the counterpart of dynamically generated code fragments. Assumptions,
represented by the \assume instruction, support dynamic deoptimization of
speculatively compiled code.  The syntax of \sourir instructions is shown in
\autoref{fig:syntax-i}.

\Sourir supports defining a local variable, removing a variable from scope,
variable assignment, creating arrays, array assignment, (unstructured)
control flow, input and output, function calls and returns, assumptions, and
terminating execution. Control-flow instructions take explicit labels, which
are compiler-generated symbols but we sometimes give them meaningful names
for clarity of exposition.  Literals are integers, booleans, and nil.
Together with variables and function references, they form simple
expressions.  Finally, an expression is either a simple expression or an
operation: array access, array length, or primitive operation (arithmetic,
comparison, and logic operation).  Expressions are not nested---this is
common in intermediate representations such as A-normal
form~\citep{sab92}. We do allow bounded nesting in instructions for brevity.

\begin{figure}
\hrule\begin{small}
\vskip 2mm
\begin{mathpar}
\grammartabularSTY{%
  \otti{}\ottinterrule
 }
\grammartabularSTY{%
  \otteprint{}\ottinterrule
  \ottse{}\ottinterrule
  \ottlit{}\ottinterrule
 }

 \begin{matrix*}[l]
   \xi & ::=    &   \ottmv{F} .\kern -0.5pt \ottmv{V} .\kern -0.5pt \ottmv{L}  ~ \mathit{VA} 
                      & \emph{target and varmap} \\
   \tilde{\xi} & ::=    &   \ottmv{F} .\kern -0.5pt \ottmv{V} .\kern -0.5pt \ottmv{L}  ~ \mathit{x} ~ \mathit{VA} 
                      & \emph{extra continuation} \\
   \mathit{VA} & ::=
                      & [ \mathit{x}_{{\mathrm{1}}}  \ottsym{=}  \ottnt{e_{{\mathrm{1}}}}  \ottsym{,} \, .. \, \ottsym{,}  \mathit{x}_{\ottmv{n}}  \ottsym{=}  \ottnt{e_{\ottmv{n}}} ]
                      & \emph{varmap} \\
   \end{matrix*}
\end{mathpar}
\end{small}\caption{The syntax of \sourir.}
\label{fig:syntax-i}\end{figure}

A program~$\ottnt{P}$ is a set of function declarations.  The body of a
function is a list of versions indexed by a version label, where each
version is an instruction sequence.  The first instruction sequence in the
list (the \emph{active version}) is executed when the function is called.
$\ottmv{F}$ ranges over function names, $\ottmv{V}$ over version labels,
and $\ottmv{L}$ over instruction labels.  An absolute reference to an instruction
is thus a triple $ \ottmv{F} .\kern -0.5pt \ottmv{V} .\kern -0.5pt \ottmv{L} $.  Every instruction is
labeled, but for brevity we omit unused labels.

Versions model the speculative optimizations performed by the compiler. The
only instruction that explicitly references versions is \assume.  It has the
form $  \textsfb{assume} ~ e^* ~ \ottkw{else} ~ \xi ~ \tilde{\xi}^*  $ with a list of predicates
$(e^*)$ and deoptimization metadata $\xi$ and $\tilde{\xi}^*$.  When
executed, \assume evaluates its predicates; if they hold execution skips to
the next instruction. Otherwise, deoptimization occurs according to the
metadata.  The format of $\xi$ is $  \ottmv{F} .\kern -0.5pt \ottmv{V} .\kern -0.5pt \ottmv{L}  ~  [{ \mathit{x}_{{\mathrm{1}}}  \ottsym{=}  \ottnt{e_{{\mathrm{1}}}}  \ottsym{,} \, .. \, \ottsym{,}  \mathit{x}_{\ottmv{n}}  \ottsym{=}  \ottnt{e_{\ottmv{n}}}  \kern 0.033em}]  $,
which contains a target~$ \ottmv{F} .\kern -0.5pt \ottmv{V} .\kern -0.5pt \ottmv{L} $ and a varmap~$[ \mathit{x}_{{\mathrm{1}}}  \ottsym{=}  \ottnt{e_{{\mathrm{1}}}}  \ottsym{,} \, .. \, \ottsym{,}  \mathit{x}_{\ottmv{n}}  \ottsym{=}  \ottnt{e_{\ottmv{n}}}
]$.  To deoptimize, a fresh environment for the target is created according
to the varmap.  Each expression~$e_i$ is evaluated in the old
environment and bound to $x_i$ in the new environment.  The environment
specified by $\xi$ replaces the current one. Deoptimization might also
need to create additional continuations, if \assume occurs in an inlined
function. In this case multiple $\tilde{\xi}$ of the form $  \ottmv{F} .\kern -0.5pt \ottmv{V} .\kern -0.5pt \ottmv{L}  ~ \mathit{x} ~  [{ \mathit{x}_{{\mathrm{1}}}  \ottsym{=}  \ottnt{e_{{\mathrm{1}}}}  \ottsym{,} \, .. \, \ottsym{,}  \mathit{x}_{\ottmv{n}}  \ottsym{=}  \ottnt{e_{\ottmv{n}}}  \kern 0.033em}]  $ can be appended. Each one synthesizes a
continuation with an environment constructed according to the varmap, a
return target $ \ottmv{F} .\kern -0.5pt \ottmv{V} .\kern -0.5pt \ottmv{L} $, and the name $\mathit{x}$ to hold the returned
result---this situation and inlining are discussed in~\autoref{sec:inlining}.
The purpose of deoptimization metadata is twofold.  First, it provides the
necessary information for jumping to the target version.  Second, its
presence in the instruction stream allows the optimizer to keep the mapping
between different versions up-to-date.

\paragraph{Example}

\begin{wrapfigure}{r}{7.8cm}
\vspace{-.4cm}
  $  {\small
                                                          \begin{array}{l}
                                                              \scalebox{0.97}{$\mathsf{ size }$}   (   \mathsf{ x }   \kern 0.04em ) \\
                                                                   \begin{array}{llll}
                                                                     \kern 1.5pt   \hspace{2.5mm}  \mathsf{V\kern -0.25pt \scalebox{0.87}{$\mathsf{ o }$} }   \\
                                                                             \begin{array}{l!{\,\color{gray}\vrule}lll}
                                                                               &      &  \textsfb{assume} ~   \mathsf{ x }    \neq   \textsfb{nil}  ~ \ottkw{else} ~    \scalebox{0.97}{$\mathsf{ size }$}  .\kern -0.5pt  \mathsf{V\kern -0.25pt \scalebox{0.87}{$\mathsf{ b }$} }  .\kern -0.5pt  \mathsf{L\kern -0.25pt \scalebox{0.83}{$\mathsf{ 2 }$} }   ~  [{  \mathsf{ el }   \ottsym{=}  32  \ottsym{,}   \mathsf{ x }   \ottsym{=}   \mathsf{ x }   \kern 0.033em}]    \\  \,  &   \phantom{  \mathsf{L\kern -0.25pt \scalebox{0.83}{$\mathsf{ 1 }$} }  }   &   \textsfb{var} ~  \mathsf{ l }  \nobreak\hspace{0pt}=\nobreak\hspace{0pt}   \mathsf{ x }  [{ 0 }]    \\  \,  &      &  \textsfb{return} \,   \mathsf{ l }    \ottsym{*}   32   \\  
                                                                             \end{array} \vspace{0.25em} \\  \,  \hspace{2.5mm}  \mathsf{V\kern -0.25pt \scalebox{0.87}{$\mathsf{ b }$} }   \\
                                                                             \begin{array}{l!{\,\color{gray}\vrule}lll}
                                                                               &   \mathsf{L\kern -0.25pt \scalebox{0.83}{$\mathsf{ 1 }$} }   &   \textsfb{var} ~  \mathsf{ el }  \nobreak\hspace{0pt}=\nobreak\hspace{0pt} 32   \\  \,  &   \mathsf{L\kern -0.25pt \scalebox{0.83}{$\mathsf{ 2 }$} }   &   \textsfb{branch} ~   \mathsf{ x }    =   \textsfb{nil}  ~  \mathsf{L\kern -0.25pt \scalebox{0.83}{$\mathsf{ 4 }$} }  ~  \mathsf{L\kern -0.25pt \scalebox{0.83}{$\mathsf{ 3 }$} }    \\  \,  &   \mathsf{L\kern -0.25pt \scalebox{0.83}{$\mathsf{ 3 }$} }   &   \textsfb{var} ~  \mathsf{ l }  \nobreak\hspace{0pt}=\nobreak\hspace{0pt}   \mathsf{ x }  [{ 0 }]    \\  \,  &      &  \textsfb{return} \,   \mathsf{ l }    \ottsym{*}    \mathsf{ el }    \\  \,  &   \mathsf{L\kern -0.25pt \scalebox{0.83}{$\mathsf{ 4 }$} }   &  \textsfb{return} \, 0  \\  
                                                                             \end{array} \vspace{0.25em} \\  
                                                                   \end{array} \\  
                                                          \end{array} }  $
\caption{Speculation on $  \mathsf{ x }  $.}\label{fig:first-example}
\vspace{-.1cm}
\end{wrapfigure} %
Consider the function~$  \scalebox{0.97}{$\mathsf{ size }$}  $ in \autoref{fig:first-example}
which computes the size of a vector $  \mathsf{ x }  $.  In version~$  \mathsf{V\kern -0.25pt \scalebox{0.87}{$\mathsf{ b }$} }  $, $  \mathsf{ x }  $
is either nil or an array with its length stored at index~$0$.  The
optimized version~$  \mathsf{V\kern -0.25pt \scalebox{0.87}{$\mathsf{ o }$} }  $ expects that the input is never nil.
Classical compiler optimizations can leverage this fact: unreachable code
removal prunes the unused branch.  Constant propagation replaces the use of
$  \mathsf{ el }  $ with its value and updates the varmap so that it restores the
deleted variable upon deoptimization to the base version~$  \mathsf{V\kern -0.25pt \scalebox{0.87}{$\mathsf{ b }$} }  $.

\clearpage

\subsection{Deoptimization Invariants}\label{subsec:invariants}

\begin{wrapfigure}{r}{7.8cm}
\vspace{-.2cm}
$  {\small
                                                          \begin{array}{l}
                                                              \scalebox{0.97}{$\mathsf{ show }$}   (   \mathsf{ x }   \kern 0.04em ) \\
                                                                   \begin{array}{llll}
                                                                     \kern 1.5pt   \hspace{2.5mm}  \mathsf{V\kern -0.25pt \scalebox{0.87}{$\mathsf{ o }$} }   \\
                                                                             \begin{array}{l!{\,\color{gray}\vrule}lll}
                                                                               &   \phantom{  \mathsf{L\kern -0.25pt \scalebox{0.83}{$\mathsf{ 1 }$} }  }   &  \textsfb{assume} ~   \mathsf{ x }    =   42  ~ \ottkw{else} ~    \scalebox{0.97}{$\mathsf{ show }$}  .\kern -0.5pt  \mathsf{V\kern -0.25pt \scalebox{0.87}{$\mathsf{ b }$} }  .\kern -0.5pt  \mathsf{L\kern -0.25pt \scalebox{0.83}{$\mathsf{ 1 }$} }   ~  [{  \mathsf{ x }   \ottsym{=}   \mathsf{ x }   \kern 0.033em}]    \\  \,  &      &  \textsfb{print} \, 42  \\  
                                                                             \end{array} \vspace{0.25em} \\  \,  \hspace{2.5mm}  \mathsf{V\kern -0.25pt \scalebox{0.87}{$\mathsf{ w }$} }   \\
                                                                             \begin{array}{l!{\,\color{gray}\vrule}lll}
                                                                               &   \phantom{  \mathsf{L\kern -0.25pt \scalebox{0.83}{$\mathsf{ 1 }$} }  }   &  \textsfb{assume} ~ \textsfb{true} ~ \ottkw{else} ~    \scalebox{0.97}{$\mathsf{ show }$}  .\kern -0.5pt  \mathsf{V\kern -0.25pt \scalebox{0.87}{$\mathsf{ b }$} }  .\kern -0.5pt  \mathsf{L\kern -0.25pt \scalebox{0.83}{$\mathsf{ 1 }$} }   ~  [{  \mathsf{ x }   \ottsym{=}  42  \kern 0.033em}]    \\  \,  &      &  \textsfb{print} \,  \mathsf{ x }   \\  
                                                                             \end{array} \vspace{0.25em} \\  \,  \hspace{2.5mm}  \mathsf{V\kern -0.25pt \scalebox{0.87}{$\mathsf{ b }$} }   \\
                                                                             \begin{array}{l!{\,\color{gray}\vrule}lll}
                                                                               &   \mathsf{L\kern -0.25pt \scalebox{0.83}{$\mathsf{ 1 }$} }   &  \textsfb{print} \,  \mathsf{ x }   \\  
                                                                             \end{array} \vspace{0.25em} \\  
                                                                   \end{array} \\  
                                                          \end{array} } 
$
\caption{The version $\mathsf{w}$ violates the deoptimization invariant.}
\label{fig:deopt-invariants}
\end{wrapfigure} %
A version is the unit of optimization and deoptimization.  Thus we expect
that each function will have one original version and possibly many
optimized versions.  Versions are constructed such that they preserve two
crucial invariants: (1) \emph{version equivalence} and (2) \emph{assumption
  transparency}.  By the first invariant all versions of a function are
observationally equivalent.  The second invariant ensures that even if the
assumption predicates \emph{do} hold, deoptimizing to the target should be
correct.  Thus one could execute an optimized version and its base in
lockstep; at every \assume the varmap provides a complete mapping from the
new version to the base.  This simulation relation between versions is our
correctness argument.  The transparency invariant allows us to add
assumption predicates without fear of altering program semantics.  Consider
a function $  \scalebox{0.97}{$\mathsf{ show }$}  $ in \autoref{fig:deopt-invariants} which
prints its argument~$  \mathsf{ x }  $.  Version~$  \mathsf{V\kern -0.25pt \scalebox{0.87}{$\mathsf{ o }$} }  $ respects both
invariants: any value for $  \mathsf{ x }  $ will result in the same behavior as
the base version and deoptimizing is always possible.  On the other hand,
$  \mathsf{V\kern -0.25pt \scalebox{0.87}{$\mathsf{ w }$} }  $, which is equivalent because it will never deoptimize,
violates the second invariant: if it were to deoptimize, the value of
$  \mathsf{ x }  $ would be set to $42$, which is almost always incorrect.  We
present a formal treatment of the invariants and the correctness proofs in
\autoref{subsec:formal-invariants} and following.

\subsection{Creating Fresh Versions}\label{sec:new-version}

We expect that versions are chained.  A compiler will create a new version,
say $ \mathsf{V\kern -0.5pt \scalebox{0.83}{$\mathsf{ 1 }$} } $, from an existing version $ \mathsf{V\kern -0.5pt \scalebox{0.83}{$\mathsf{ 0 }$} } $ by copying all
instructions from the original version and chaining their
\begin{wrapfigure}{r}{7.9cm}
\vskip -1mm
$  {\small
                                                          \begin{array}{l}
                                                              \scalebox{0.97}{$\mathsf{ fun }$}   (  \,  \kern 0.04em ) \\
                                                                   \begin{array}{llll}
                                                                     \kern 1.5pt   \hspace{2.5mm}  \mathsf{V\kern -0.5pt \scalebox{0.83}{$\mathsf{ 2 }$} }   \\
                                                                             \begin{array}{l!{\,\color{gray}\vrule}lll}
                                                                               &   \mathsf{L\kern -0.25pt \scalebox{0.83}{$\mathsf{ 0 }$} }   &  \textsfb{assume} ~ \textsfb{true} ~ \ottkw{else} ~    \scalebox{0.97}{$\mathsf{ fun }$}  .\kern -0.5pt  \mathsf{V\kern -0.5pt \scalebox{0.83}{$\mathsf{ 1 }$} }  .\kern -0.5pt  \mathsf{L\kern -0.25pt \scalebox{0.83}{$\mathsf{ 0 }$} }   ~  [{ \,  \kern 0.033em}]    \\  \,  &      &   \textsfb{var} ~  \mathsf{ x }  \nobreak\hspace{0pt}=\nobreak\hspace{0pt} 1   \\  \,  &   \mathsf{L\kern -0.25pt \scalebox{0.83}{$\mathsf{ 1 }$} }   &  \textsfb{assume} ~ \ottnt{e} ~ \ottkw{else} ~    \scalebox{0.97}{$\mathsf{ fun }$}  .\kern -0.5pt  \mathsf{V\kern -0.5pt \scalebox{0.83}{$\mathsf{ 1 }$} }  .\kern -0.5pt  \mathsf{L\kern -0.25pt \scalebox{0.83}{$\mathsf{ 1 }$} }   ~  [{  \mathsf{ x }   \ottsym{=}   \mathsf{ x }   \kern 0.033em}]    \\  \,  &   \mathsf{L\kern -0.25pt \scalebox{0.83}{$\mathsf{ 2 }$} }   &  \textsfb{print} \,   \mathsf{ x }    \ottsym{+}   2   \\  
                                                                             \end{array} \vspace{0.25em} \\  \,  \hspace{2.5mm}  \mathsf{V\kern -0.5pt \scalebox{0.83}{$\mathsf{ 1 }$} }   \\
                                                                             \begin{array}{l!{\,\color{gray}\vrule}lll}
                                                                               &   \mathsf{L\kern -0.25pt \scalebox{0.83}{$\mathsf{ 0 }$} }   &   \textsfb{var} ~  \mathsf{ x }  \nobreak\hspace{0pt}=\nobreak\hspace{0pt} 1   \\  \,  &   \mathsf{L\kern -0.25pt \scalebox{0.83}{$\mathsf{ 1 }$} }   &  \textsfb{assume} ~ \ottnt{e} ~ \ottkw{else} ~    \scalebox{0.97}{$\mathsf{ fun }$}  .\kern -0.5pt  \mathsf{V\kern -0.5pt \scalebox{0.83}{$\mathsf{ 0 }$} }  .\kern -0.5pt  \mathsf{L\kern -0.25pt \scalebox{0.83}{$\mathsf{ 1 }$} }   ~  [{  \mathsf{ g }   \ottsym{=}   \mathsf{ x }   \kern 0.033em}]    \\  \,  &   \mathsf{L\kern -0.25pt \scalebox{0.83}{$\mathsf{ 2 }$} }   &  \textsfb{assume} ~ \textsfb{true} ~ \ottkw{else} ~    \scalebox{0.97}{$\mathsf{ fun }$}  .\kern -0.5pt  \mathsf{V\kern -0.5pt \scalebox{0.83}{$\mathsf{ 0 }$} }  .\kern -0.5pt  \mathsf{L\kern -0.25pt \scalebox{0.83}{$\mathsf{ 2 }$} }   ~  [{  \mathsf{ g }   \ottsym{=}   \mathsf{ x }   \ottsym{,}   \mathsf{ h }   \ottsym{=}    \mathsf{ x }    \ottsym{+}   1   \kern 0.033em}]    \\  \,  &      &  \textsfb{print} \,   \mathsf{ x }    \ottsym{+}   2   \\  
                                                                             \end{array} \vspace{0.25em} \\  \,  \hspace{2.5mm}  \mathsf{V\kern -0.5pt \scalebox{0.83}{$\mathsf{ 0 }$} }   \\
                                                                             \begin{array}{l!{\,\color{gray}\vrule}lll}
                                                                               &   \mathsf{L\kern -0.25pt \scalebox{0.83}{$\mathsf{ 0 }$} }   &   \textsfb{var} ~  \mathsf{ g }  \nobreak\hspace{0pt}=\nobreak\hspace{0pt} 1   \\  \,  &   \mathsf{L\kern -0.25pt \scalebox{0.83}{$\mathsf{ 1 }$} }   &   \textsfb{var} ~  \mathsf{ h }  \nobreak\hspace{0pt}=\nobreak\hspace{0pt}   \mathsf{ g }    \ottsym{+}   1    \\  \,  &   \mathsf{L\kern -0.25pt \scalebox{0.83}{$\mathsf{ 2 }$} }   &  \textsfb{print} \,   \mathsf{ h }    \ottsym{+}   1   \\  
                                                                             \end{array} \vspace{0.25em} \\  
                                                                   \end{array} \\  
                                                          \end{array} }  $
\caption{Chained \assume instructions: Version 1 was created from 0, then optimized. Version 2 is a fresh copy of 1.}
\label{fig:create-version}
\end{wrapfigure} %
deoptimization targets.
The latter is done by updating the target and varmap of \assume instructions
such that all targets refer to $ \mathsf{V\kern -0.5pt \scalebox{0.83}{$\mathsf{ 0 }$} } $ at the same label as the
current instruction. As the new version starts out as a copy, the varmap
is the identity function.
For instance, if the target contains the variables $  \mathsf{ x }  $ and $  \mathsf{ y }  $, then the varmap is $[  \mathsf{ x }   \ottsym{=}   \mathsf{ x }   \ottsym{,}   \mathsf{ z }   \ottsym{=}   \mathsf{ z }  ]$.
Additional \assume instructions can be added;
\assume instructions that bear no predicates (\ie, the predicate list is either empty
or just tautologies) can be removed while preserving equivalence.
As an example in \autoref{fig:create-version}, the new version $ \mathsf{V\kern -0.5pt \scalebox{0.83}{$\mathsf{ 2 }$} } $ is a copy of $ \mathsf{V\kern -0.5pt \scalebox{0.83}{$\mathsf{ 1 }$} } $; the instruction at
$ \mathsf{L\kern -0.25pt \scalebox{0.83}{$\mathsf{ 0 }$} } $ was added, the instruction at $ \mathsf{L\kern -0.25pt \scalebox{0.83}{$\mathsf{ 1 }$} } $ was
updated, and the one at $ \mathsf{L\kern -0.25pt \scalebox{0.83}{$\mathsf{ 2 }$} } $ was removed.

\clearpage
\begin{wrapfigure}{r}{8cm}
$  {\small
                                                          \begin{array}{l}
                                                              \scalebox{0.97}{$\mathsf{ size }$}   (   \mathsf{ x }   \kern 0.04em ) \\
                                                                   \begin{array}{llll}
                                                                     \kern 1.5pt   \hspace{2.5mm}  \mathsf{V\kern -0.25pt \scalebox{0.87}{$\mathsf{ dup }$} }   \\
                                                                             \begin{array}{l!{\,\color{gray}\vrule}lll}
                                                                               &   \mathsf{L\kern -0.25pt \scalebox{0.83}{$\mathsf{ 1 }$} }   &  \textsfb{assume} ~ \textsfb{true} ~ \ottkw{else} ~    \scalebox{0.97}{$\mathsf{ size }$}  .\kern -0.5pt  \mathsf{V\kern -0.25pt \scalebox{0.87}{$\mathsf{ b }$} }  .\kern -0.5pt  \mathsf{L\kern -0.25pt \scalebox{0.83}{$\mathsf{ 1 }$} }   ~  [{  \mathsf{ x }   \ottsym{=}   \mathsf{ x }   \kern 0.033em}]    \\  \,  &      &   \textsfb{var} ~  \mathsf{ el }  \nobreak\hspace{0pt}=\nobreak\hspace{0pt} 32   \\  \,  &   \mathsf{L\kern -0.25pt \scalebox{0.83}{$\mathsf{ 2 }$} }   &  \textsfb{assume} ~ \textsfb{true} ~ \ottkw{else} ~    \scalebox{0.97}{$\mathsf{ size }$}  .\kern -0.5pt  \mathsf{V\kern -0.25pt \scalebox{0.87}{$\mathsf{ b }$} }  .\kern -0.5pt  \mathsf{L\kern -0.25pt \scalebox{0.83}{$\mathsf{ 2 }$} }   ~  [{  \mathsf{ el }   \ottsym{=}   \mathsf{ el }   \ottsym{,}   \mathsf{ x }   \ottsym{=}   \mathsf{ x }   \kern 0.033em}]    \\  \,  &      &   \textsfb{branch} ~   \mathsf{ x }    =   \textsfb{nil}  ~  \mathsf{L\kern -0.25pt \scalebox{0.83}{$\mathsf{ 4 }$} }  ~  \mathsf{L\kern -0.25pt \scalebox{0.83}{$\mathsf{ 3 }$} }    \\  \,  &   \mathsf{L\kern -0.25pt \scalebox{0.83}{$\mathsf{ 3 }$} }   &   \textsfb{var} ~  \mathsf{ l }  \nobreak\hspace{0pt}=\nobreak\hspace{0pt}   \mathsf{ x }  [{ 0 }]    \\  \,  &      &  \textsfb{return} \,   \mathsf{ l }    \ottsym{*}    \mathsf{ el }    \\  \,  &   \mathsf{L\kern -0.25pt \scalebox{0.83}{$\mathsf{ 4 }$} }   &  \textsfb{return} \, 0  \\  
                                                                             \end{array} \vspace{0.25em} \\  \,  \hspace{2.5mm}  \mathsf{V\kern -0.25pt \scalebox{0.87}{$\mathsf{ b }$} }   ~\dots~ \\  
                                                                   \end{array} \\  
                                                          \end{array} }  $
\caption{A fresh copy of the base version of size.}\label{fig:copied-example}
\vspace{-.1cm}
\end{wrapfigure}%
Updating \assume instructions is not required for correctness. 
But the idea with a new version is that it captures a set of assumptions that can be undone independently from the previously existing assumptions.
Thus, we want to be able to undo one version at a time. In an implementation,
versions might, for example, correspond to optimization tiers.\footnote{A common strategy for VMs is to have different kind of optimizing compilers with different compilation speed versus code quality trade-offs. The more a code fragment is executed, the more powerful optimizations will be applied to it.}
This approach can lead to a cascade of deoptimizations if an
inherited assumption fails; we discuss this in \autoref{sec:checkpoint-combine}.
In the following sections we use the base version~$  \mathsf{V\kern -0.25pt \scalebox{0.87}{$\mathsf{ b }$} }  $ of \autoref{fig:first-example} as our running example.
As a first step, we generate the new version~$  \mathsf{V\kern -0.25pt \scalebox{0.87}{$\mathsf{ dup }$} }  $ with two fresh \assume instructions shown in \autoref{fig:copied-example}.
Initially the predicates are \code{true} and the \assume instructions never fire.
Version~$  \mathsf{V\kern -0.25pt \scalebox{0.87}{$\mathsf{ b }$} }  $ stays unchanged.

\subsection{Injecting Assumptions}\label{sec:inject-assumption}

We advocate an approach where the compiler
first injects assumption predicates,
and then uses them in optimizations.  In contrast, earlier work would
apply an unsound optimization and then recover by adding a guard~(see, for
example, \citet{dub13}).  While the end result is the same, the different
perspective helps with reasoning about correctness.  Assumptions are boolean
predicates, similar to user-provided assertions.  For example, to speculate on
a branch target, the assumption is the branch condition or its negation. It is
therefore correct for the compiler to expect that the predicate holds
immediately following an \assume.
Injecting predicates is done after establishing the correspondence between two versions with \assume instructions, as presented above.
Inserting a fresh \assume in a function is difficult in general, as one
must determine where to transfer control to or how to reconstruct the target
environment.  On the other hand, it is always correct to add a predicate to an
existing \assume.  Thanks to the assumption transparency
invariant it is safe to deoptimize more often to the target.
For instance, in $  \textsfb{assume} ~   \mathsf{ x }    \neq   \textsfb{nil}   \ottsym{,}    \mathsf{ x }    >   10  ~ \ottkw{else} ~  \dots   $ the
predicate $  \mathsf{ x }    \neq   \textsfb{nil} $ was narrowed down to $  \mathsf{ x }    >   10 $.

\section{Optimization with Assumptions}\label{sec:optimizations}

In the previous section we introduced our approach for establishing a fresh
version of a function that lends itself to speculative optimizations. Next,
we introduce classical compiler optimizations that are exemplary of our
approach.  Then we give additional transformations for the \assume in
\autoref{sec:move-checkpoint} and following, and conclude with a case study
in \autoref{sec:case-study}.  All transformations introduced in this section
are proved correct in \autoref{sec:optimizations-correctness}.

\subsection{Constant Propagation}
\label{sec:constantprop}\label{sec:specl-constantprop}

Consider a simple constant propagation pass that finds constant variables
and then updates all uses.  This pass maintains a map
from variable names to constant expressions or \emph{unknown}.  The map is
computed for every position in the instruction stream using a data-flow
analysis. Following the approach by \citet{kil73}, the analysis has an update function to add and
remove constants to the map.  For example analyzing $  \textsfb{var} ~  \mathsf{ x }  \nobreak\hspace{0pt}=\nobreak\hspace{0pt} 2  $, or
$   \mathsf{ x }  \nobreak\hspace{0pt} \leftarrow \nobreak\hspace{0pt} 2  $ adds the mapping $  \mathsf{ x }   \rightarrow 2$.  The instruction
$  \textsfb{var} ~  \mathsf{ y }  \nobreak\hspace{0pt}=\nobreak\hspace{0pt}   \mathsf{ x }    \ottsym{+}   1   $ adds $  \mathsf{ y }   \rightarrow 3$ to the previous map.  Finally,
$ \textsfb{drop} \,  \mathsf{ x }  $ removes a mapping.  Control-flow merges rely on a join
function for intersecting two maps; mappings which agree are preserved, while
others are set to \emph{unknown}.  In a second step, expressions that can be
evaluated to values are replaced and unused variables are removed.  No
additional care needs to be taken to make this pass correct in the presence
of assumptions.  This is because in \sourir, the expressions needed to
reconstruct environments appear in the varmap of the \assume
and are thus visible to the constant propagation pass.
Additionally, the pass can update them, for
example, in $  \textsfb{assume} ~ \textsfb{true} ~ \ottkw{else} ~    \mathsf{F}  .\kern -0.5pt  \mathsf{V}  .\kern -0.5pt  \mathsf{L}   ~  [{  \mathsf{ x }   \ottsym{=}    \mathsf{ y }    \ottsym{+}    \mathsf{ z }    \kern 0.033em}]    $, the variables
$  \mathsf{ y }  $ and $  \mathsf{ z }  $ are treated the same as in $  \textsfb{call} ~  \mathsf{ h }  \nobreak\hspace{0pt}=\nobreak\hspace{0pt}  {  \scalebox{0.97}{$\mathsf{ foo }$}  }  (   \mathsf{ y }    \ottsym{+}    \mathsf{ z }   )  $.  They can be replaced and will not artificially keep
constant variables alive.

Constant propagation can become speculative.  After the instruction
$  \textsfb{assume} ~   \mathsf{ x }    =   0  ~ \ottkw{else} ~  \dots   $, the variable $  \mathsf{ x }  $ is 0.  Therefore,
$  \mathsf{ x }   \leftarrow 0$ is added to the state map.  This is the only
extension required for speculative constant propagation. As an example, in
the case where we speculate on a nil check \hskip -3pt
  \( \\ {\small \begin{array}{lll}
                                                      &      & \dots \\  \,  &   \mathsf{L\kern -0.25pt \scalebox{0.83}{$\mathsf{ 2 }$} }   &  \textsfb{assume} ~   \mathsf{ x }    \neq   \textsfb{nil}  ~ \ottkw{else} ~    \scalebox{0.97}{$\mathsf{ size }$}  .\kern -0.5pt  \mathsf{V\kern -0.25pt \scalebox{0.87}{$\mathsf{ b }$} }  .\kern -0.5pt  \mathsf{L\kern -0.25pt \scalebox{0.83}{$\mathsf{ 2 }$} }   ~  [{  \mathsf{ el }   \ottsym{=}   \mathsf{ el }   \ottsym{,}   \mathsf{ x }   \ottsym{=}   \mathsf{ x }   \kern 0.033em}]    \\  \,  &      &   \textsfb{branch} ~   \mathsf{ x }    =   \textsfb{nil}  ~  \mathsf{L\kern -0.25pt \scalebox{0.83}{$\mathsf{ 4 }$} }  ~  \mathsf{L\kern -0.25pt \scalebox{0.83}{$\mathsf{ 3 }$} }    \\  \,  &      & \dots \\  
                                                \end{array} }  \\ \)
the map is $  \mathsf{ x }   \rightarrow   \neg    \textsfb{nil} $ after $ \mathsf{L\kern -0.25pt \scalebox{0.83}{$\mathsf{ 2 }$} } $.
Evaluating the branch condition under this context yields $  \neg    \textsfb{nil}  ==
 \textsfb{nil} $, and a further optimization opportunity presents itself.

\subsection{Unreachable Code Elimination}\label{sec:prune}

\begin{wrapfigure}{r}{7.5cm}
\vspace{-.3cm}
$  {\small
                                                          \begin{array}{l}
                                                              \scalebox{0.97}{$\mathsf{ size }$}   (   \mathsf{ x }   \kern 0.04em ) \\
                                                                   \begin{array}{llll}
                                                                     \kern 1.5pt   \hspace{2.5mm}  \mathsf{V\kern -0.25pt \scalebox{0.87}{$\mathsf{ pruned }$} }   \\
                                                                             \begin{array}{l!{\,\color{gray}\vrule}lll}
                                                                               &   \mathsf{L\kern -0.25pt \scalebox{0.83}{$\mathsf{ 1 }$} }   &  \textsfb{assume} ~ \textsfb{true} ~ \ottkw{else} ~    \scalebox{0.97}{$\mathsf{ size }$}  .\kern -0.5pt  \mathsf{V\kern -0.25pt \scalebox{0.87}{$\mathsf{ b }$} }  .\kern -0.5pt  \mathsf{L\kern -0.25pt \scalebox{0.83}{$\mathsf{ 1 }$} }   ~  [{  \mathsf{ x }   \ottsym{=}   \mathsf{ x }   \kern 0.033em}]    \\  \,  &      &   \textsfb{var} ~  \mathsf{ el }  \nobreak\hspace{0pt}=\nobreak\hspace{0pt} 32   \\  \,  &   \mathsf{L\kern -0.25pt \scalebox{0.83}{$\mathsf{ 2 }$} }   &  \textsfb{assume} ~   \mathsf{ x }    \neq   \textsfb{nil}  ~ \ottkw{else} ~    \scalebox{0.97}{$\mathsf{ size }$}  .\kern -0.5pt  \mathsf{V\kern -0.25pt \scalebox{0.87}{$\mathsf{ b }$} }  .\kern -0.5pt  \mathsf{L\kern -0.25pt \scalebox{0.83}{$\mathsf{ 2 }$} }   ~  [{  \mathsf{ el }   \ottsym{=}   \mathsf{ el }   \ottsym{,}   \mathsf{ x }   \ottsym{=}   \mathsf{ x }   \kern 0.033em}]    \\  \,  &      &   \textsfb{var} ~  \mathsf{ l }  \nobreak\hspace{0pt}=\nobreak\hspace{0pt}   \mathsf{ x }  [{ 0 }]    \\  \,  &      &  \textsfb{return} \,   \mathsf{ l }    \ottsym{*}    \mathsf{ el }    \\  
                                                                             \end{array} \vspace{0.25em} \\  \,  \hspace{2.5mm}  \mathsf{V\kern -0.25pt \scalebox{0.87}{$\mathsf{ b }$} }   ~\dots~ \\  
                                                                   \end{array} \\  
                                                          \end{array} }  $
\caption{A speculation that the argument is not nil eliminated one of the former branches.}\label{fig:pruned-example}
\vspace{-.2cm}
\end{wrapfigure} %
As shown above, an assumption coupled with constant folding leads to branches becoming
deterministic.  Unreachable code elimination benefits from that.
We consider a two step algorithm: the first pass replaces $  \textsfb{branch} ~ \ottnt{e} ~  \mathsf{L\kern -0.25pt \scalebox{0.83}{$\mathsf{ 1 }$} }  ~  \mathsf{L\kern -0.25pt \scalebox{0.83}{$\mathsf{ 2 }$} }   $ with $ \textsfb{goto} \,  \mathsf{L\kern -0.25pt \scalebox{0.83}{$\mathsf{ 1 }$} }  $ if $e$ is a tautology and with
$ \textsfb{goto} \,  \mathsf{L\kern -0.25pt \scalebox{0.83}{$\mathsf{ 2 }$} }  $ if it is a contradiction.  The second pass removes
unreachable instructions.  In our running example from
\autoref{fig:copied-example}, we add the predicate $  \mathsf{ x }    \neq   \textsfb{nil} $ to the
empty \assume at $ \mathsf{L\kern -0.25pt \scalebox{0.83}{$\mathsf{ 2 }$} } $.  Constant propagation shows that the
branch always goes to $ \mathsf{L\kern -0.25pt \scalebox{0.83}{$\mathsf{ 3 }$} } $, and unreachable code elimination removes
the dead statement at $ \mathsf{L\kern -0.25pt \scalebox{0.83}{$\mathsf{ 4 }$} } $ and branch.  This creates the version shown in \autoref{fig:pruned-example}.
Additionally, constant
propagation can replace $  \mathsf{ el }  $ by $32$. By also replacing its
mention in the varmap of the \assume at $ \mathsf{L\kern -0.25pt \scalebox{0.83}{$\mathsf{ 2 }$} } $, $  \mathsf{ el }  $
becomes unused and can be removed from the optimized version. This yields
version~$  \mathsf{V\kern -0.25pt \scalebox{0.87}{$\mathsf{ o }$} }  $ in \autoref{fig:first-example} at the top.

\subsection{Function Inlining}\label{sec:inlining}

Function inlining is our most involved optimization, since \assume
instructions inherited from the inlinee need to remain correct.  The
inlining itself is standard. Name mangling is used to separate the caller
and callee environments.  As an example \autoref{fig:ex-inlining} shows the
inlining of $  \scalebox{0.97}{$\mathsf{ size }$}  $ into a function~$  \mathsf{main}  $.
Na\"ively inlining without updating the metadata of the \assume at
$ \mathsf{L\kern -0.25pt \scalebox{0.83}{$\mathsf{ 2 }$} } $ will result in an incorrect deoptimization, as execution would
transfer to $  \scalebox{0.97}{$\mathsf{ size }$}  .\kern -0.5pt  \mathsf{V\kern -0.25pt \scalebox{0.87}{$\mathsf{ b }$} }  .\kern -0.5pt  \mathsf{L\kern -0.25pt \scalebox{0.83}{$\mathsf{ 2 }$} }  $ with no way to return to the
$  \mathsf{main}  $ function.  Also, $  \mathsf{main}  $'s part of the
environment is discarded in the transfer and permanently lost.  The solution
is to synthesize a new stack frame.  As shown in the figure, the \assume at
in the optimized $  \mathsf{main}  $ is thus extended with
$   \mathsf{main}  .\kern -0.5pt  \mathsf{V\kern -0.25pt \scalebox{0.87}{$\mathsf{ b }$} }  .\kern -0.5pt  \mathsf{L\kern -0.25pt \scalebox{0.87}{$\mathsf{ ret }$} }   ~  \mathsf{ s }  ~  [{  \mathsf{ pl }   \ottsym{=}   \mathsf{ pl }   \ottsym{,}   \mathsf{ vec }   \ottsym{=}   \mathsf{ vec }   \kern 0.033em}]  $.\kern -2pt  
\begin{wrapfigure}[26]{r}{8.7cm}
$ {\small
                                                          \begin{array}{l}
                                                              \mathsf{main}   (  \,  \kern 0.04em ) \\
                                                                   \begin{array}{llll}
                                                                     \kern 1.5pt   \hspace{2.5mm}  \mathsf{V\kern -0.25pt \scalebox{0.87}{$\mathsf{ inl }$} }   \\
                                                                             \begin{array}{l!{\,\color{gray}\vrule}lll}
                                                                               &      &   \textsfb{array} ~  \mathsf{ pl }  \nobreak\hspace{0pt}=\nobreak\hspace{0pt}[{ 1  \ottsym{,}  2  \ottsym{,}  3  \ottsym{,}  4  \kern 0.08em}]   \\  \,  &      &   \textsfb{array} ~  \mathsf{ vec }  \nobreak\hspace{0pt}=\nobreak\hspace{0pt}[{  \textsf{length} (   \mathsf{ pl }   \kern 0.04em )   \ottsym{,}   \mathsf{ pl }   \kern 0.08em}]   \\  \,  &      &   \textsfb{var} ~  \mathsf{ s }  \nobreak\hspace{0pt}=\nobreak\hspace{0pt} \textsfb{nil}   \\  \,  &      &   \textsfb{var} ~  \mathsf{ x }  \nobreak\hspace{0pt}=\nobreak\hspace{0pt}  \mathsf{ vec }    \\  \,  &      &
                                                                                 \begin{matrix*}[l]  \textsfb{assume} ~   \mathsf{ x }    \neq   \textsfb{nil}  ~ \ottkw{else} &    \scalebox{0.97}{$\mathsf{ size }$}  .\kern -0.5pt  \mathsf{V\kern -0.25pt \scalebox{0.87}{$\mathsf{ b }$} }  .\kern -0.5pt  \mathsf{L\kern -0.25pt \scalebox{0.83}{$\mathsf{ 2 }$} }   ~  [{  \mathsf{ el }   \ottsym{=}  32  \ottsym{,}   \mathsf{ x }   \ottsym{=}   \mathsf{ x }   \kern 0.033em}]     \\
                                                                                                                   &    \mathsf{main}  .\kern -0.5pt  \mathsf{V\kern -0.25pt \scalebox{0.87}{$\mathsf{ b }$} }  .\kern -0.5pt  \mathsf{L\kern -0.25pt \scalebox{0.87}{$\mathsf{ ret }$} }   ~  \mathsf{ s }  ~  [{  \mathsf{ pl }   \ottsym{=}   \mathsf{ pl }   \ottsym{,}   \mathsf{ vec }   \ottsym{=}   \mathsf{ vec }   \kern 0.033em}]   
                                                                                 \end{matrix*} \\  \,  &      &   \textsfb{var} ~  \mathsf{ l }  \nobreak\hspace{0pt}=\nobreak\hspace{0pt}   \mathsf{ x }  [{ 0 }]    \\  \,  &      &    \mathsf{ s }  \nobreak\hspace{0pt} \leftarrow \nobreak\hspace{0pt}   \mathsf{ l }    \ottsym{*}   32    \\  \,  &      &  \textsfb{drop} \,  \mathsf{ l }   \\  \,  &      &  \textsfb{drop} \,  \mathsf{ x }   \\  \,  &      &  \textsfb{goto} \,  \mathsf{L\kern -0.25pt \scalebox{0.87}{$\mathsf{ ret }$} }   \\  \,  &   \mathsf{L\kern -0.25pt \scalebox{0.87}{$\mathsf{ ret }$} }   &  \textsfb{print} \,  \mathsf{ s }   \\  \,  &      &  \textsfb{stop}  \\  
                                                                             \end{array} \vspace{0.25em} \\  \,  \hspace{2.5mm}  \mathsf{V\kern -0.25pt \scalebox{0.87}{$\mathsf{ b }$} }   \\
                                                                             \begin{array}{l!{\,\color{gray}\vrule}lll}
                                                                               &      &   \textsfb{array} ~  \mathsf{ pl }  \nobreak\hspace{0pt}=\nobreak\hspace{0pt}[{ 1  \ottsym{,}  2  \ottsym{,}  3  \ottsym{,}  4  \kern 0.08em}]   \\  \,  &      &   \textsfb{array} ~  \mathsf{ vec }  \nobreak\hspace{0pt}=\nobreak\hspace{0pt}[{  \textsf{length} (   \mathsf{ pl }   \kern 0.04em )   \ottsym{,}   \mathsf{ pl }   \kern 0.08em}]   \\  \,  &      &   \textsfb{call} ~  \mathsf{ s }  \nobreak\hspace{0pt}=\nobreak\hspace{0pt}  {  \scalebox{0.97}{$\mathsf{ size }$}  }  (  \mathsf{ vec }  )   \\  \,  &   \mathsf{L\kern -0.25pt \scalebox{0.87}{$\mathsf{ ret }$} }   &  \textsfb{print} \,  \mathsf{ s }   \\  \,  &      &  \textsfb{stop}  \\  
                                                                             \end{array} \vspace{0.25em} \\  
                                                                   \end{array} \\  \,   \scalebox{0.97}{$\mathsf{ size }$}   (   \mathsf{ x }   \kern 0.04em ) \\
                                                                   \begin{array}{llll}
                                                                     \kern 1.5pt   \hspace{2.5mm}  \mathsf{V\kern -0.25pt \scalebox{0.87}{$\mathsf{ o }$} }   \\
                                                                             \begin{array}{l!{\,\color{gray}\vrule}lll}
                                                                               &   \mathsf{L\kern -0.25pt \scalebox{0.83}{$\mathsf{ 2 }$} }   &  \textsfb{assume} ~   \mathsf{ x }    \neq   \textsfb{nil}  ~ \ottkw{else} ~    \scalebox{0.97}{$\mathsf{ size }$}  .\kern -0.5pt  \mathsf{V\kern -0.25pt \scalebox{0.87}{$\mathsf{ b }$} }  .\kern -0.5pt  \mathsf{L\kern -0.25pt \scalebox{0.83}{$\mathsf{ 2 }$} }   ~  [{  \mathsf{ el }   \ottsym{=}  32  \ottsym{,}   \mathsf{ x }   \ottsym{=}   \mathsf{ x }   \kern 0.033em}]    \\  \,  &      &   \textsfb{var} ~  \mathsf{ l }  \nobreak\hspace{0pt}=\nobreak\hspace{0pt}   \mathsf{ x }  [{ 0 }]    \\  \,  &   \phantom{  \mathsf{L\kern -0.25pt \scalebox{0.87}{$\mathsf{ ret }$} }  }   &  \textsfb{return} \,   \mathsf{ l }    \ottsym{*}   32   \\  
                                                                             \end{array} \vspace{0.25em} \\  \,  \hspace{2.5mm}  \mathsf{V\kern -0.25pt \scalebox{0.87}{$\mathsf{ b }$} }   ~\dots~ \\  
                                                                   \end{array} \\  
                                                          \end{array} } 
$
\caption{An inlining of  \code{size} into a \code{main}.}\label{fig:ex-inlining}
\vspace{-5mm}
\end{wrapfigure}
\noindent This creates an additional stack
frame that returns to the base version of $  \mathsf{main}  $, and stores the result in
$  \mathsf{ s }  $ with the entire caller portion of the environment
reconstructed.  It is always possible to compute the continuation, since the
original call site must have a label and the scope at this label is known.
Overall, after deoptimization, it appears as if version~$  \mathsf{V\kern -0.25pt \scalebox{0.87}{$\mathsf{ b }$} }  $ of
$  \mathsf{main}  $ had called version~$  \mathsf{V\kern -0.25pt \scalebox{0.87}{$\mathsf{ b }$} }  $ of
$  \scalebox{0.97}{$\mathsf{ size }$}  $.  Note, it would erroneous to  create a continuation that returns
to the optimized version of the caller $  \mathsf{V\kern -0.25pt \scalebox{0.87}{$\mathsf{ inl }$} }  $.  If 
deoptimization from the inlined code occurs, it is precisely because some of its
assumptions are invalid.  Multiple continuations can be appended for further
levels of inlining.  The inlining needs to be applied bottom up: for the
next level of inlining, \eg, to inline $  \mathsf{V\kern -0.25pt \scalebox{0.87}{$\mathsf{ inl }$} }  $ into an outer
caller, renamings must also be applied to the expressions in the extra continuations,
since they refer to local variables in $  \mathsf{V\kern -0.25pt \scalebox{0.87}{$\mathsf{ inl }$} }  $.

\subsection{Unrestricted Deoptimization}
\label{sec:move-checkpoint}

The \assume instructions are expensive: they create dependencies on live
variables and are barriers for moving instructions.  Hoisting a
side-effecting instruction over an \assume is invalid, because if we
deoptimize the effect happens twice.  Removing a local variable is also not
possible if its value is needed to reconstruct the target environment.  Thus
it makes sense to insert as few \assume instructions as possible.  On the
other hand it is desirable to be able to ``deoptimize
everywhere''---checking assumptions in the basic block in which they are
used can avoid unnecessary deoptimization---so there is a tension between
speculation and optimization.  Reaching an \assume marks a stable state in
the execution of the program that we can fall back to, similar to a
transaction.  Implementations, like \citep{dub13}, separate deoptimization
points and the associated guards into two separate instructions, to be able
to deoptimize more freely.  As long as the effects of instructions performed
since the last deoptimization point are not observable, it is valid to throw
away intermediate results and resume control from there. Effectively, in
\sourir this corresponds to moving an \assume instruction forward in the
instruction stream, while keeping its deoptimization target fixed.

\noindent An \assume can be moved over another instruction if that instruction:

\begin{enumerate}
  \item has no side-effects and is not a call instruction,
  \item does not interfere with the varmap or predicates, and
  \item has the \assume as its only predecessor instruction.
\end{enumerate}

\noindent The first condition prevents side-effects from happening twice.  The second
condition can be enabled by copying the affected variables at the original
\assume instruction location (\ie, taking a snapshot of the required part of the
environment).\footnote{In an SSA based IR this step is not necessary for SSA variables, since the captured ones are guaranteed to stay unchanged.}
The last condition prevents capturing traces incoming from
other basic blocks where (1) and (2) do not hold for all intermediate
instructions since the original location.  This is not the weakest
condition, but a reasonable, sufficient one.  Let us consider a
modified version of our running example in \autoref{fig:move-example} on the
left.  Again, we have an \assume before the branch, but would like to place
a guard inside one of the branches.

\begin{figure}[h]
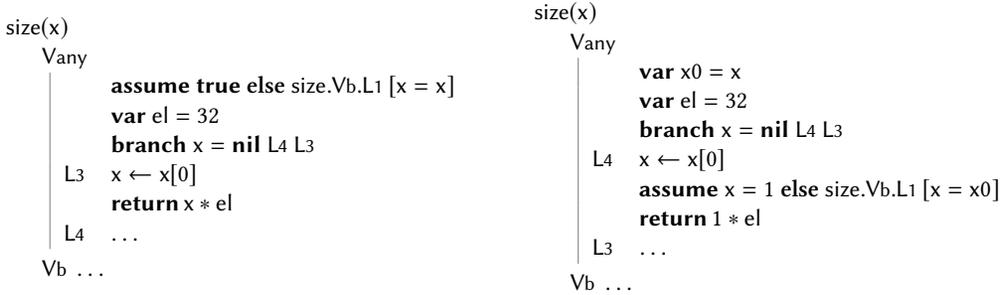

\[
  \begin{array}{ll}
   {\small
                                                          \begin{array}{l}
                                                              \scalebox{0.97}{$\mathsf{ size }$}   (   \mathsf{ x }   \kern 0.04em ) \\
                                                                   \begin{array}{llll}
                                                                     \kern 1.5pt   \hspace{2.5mm}  \mathsf{V\kern -0.25pt \scalebox{0.87}{$\mathsf{ any }$} }   \\
                                                                             \begin{array}{l!{\,\color{gray}\vrule}lll}
                                                                               &      &  \textsfb{assume} ~ \textsfb{true} ~ \ottkw{else} ~    \scalebox{0.97}{$\mathsf{ size }$}  .\kern -0.5pt  \mathsf{V\kern -0.25pt \scalebox{0.87}{$\mathsf{ b }$} }  .\kern -0.5pt  \mathsf{L\kern -0.25pt \scalebox{0.83}{$\mathsf{ 1 }$} }   ~  [{  \mathsf{ x }   \ottsym{=}   \mathsf{ x }   \kern 0.033em}]    \\  \,  &      &   \textsfb{var} ~  \mathsf{ el }  \nobreak\hspace{0pt}=\nobreak\hspace{0pt} 32   \\  \,  &      &   \textsfb{branch} ~   \mathsf{ x }    =   \textsfb{nil}  ~  \mathsf{L\kern -0.25pt \scalebox{0.83}{$\mathsf{ 4 }$} }  ~  \mathsf{L\kern -0.25pt \scalebox{0.83}{$\mathsf{ 3 }$} }    \\  \,  &   \mathsf{L\kern -0.25pt \scalebox{0.83}{$\mathsf{ 3 }$} }   &    \mathsf{ x }  \nobreak\hspace{0pt} \leftarrow \nobreak\hspace{0pt}   \mathsf{ x }  [{ 0 }]    \\  \,  &      &  \textsfb{return} \,   \mathsf{ x }    \ottsym{*}    \mathsf{ el }    \\  \,  &   \mathsf{L\kern -0.25pt \scalebox{0.83}{$\mathsf{ 4 }$} }   & \dots \\  
                                                                             \end{array} \vspace{0.25em} \\  \,  \hspace{2.5mm}  \mathsf{V\kern -0.25pt \scalebox{0.87}{$\mathsf{ b }$} }   ~\dots~ \\  
                                                                   \end{array} \\  
                                                          \end{array} } 
    &
 {\small
                                                          \begin{array}{l}
                                                              \scalebox{0.97}{$\mathsf{ size }$}   (   \mathsf{ x }   \kern 0.04em ) \\
                                                                   \begin{array}{llll}
                                                                     \kern 1.5pt   \hspace{2.5mm}  \mathsf{V\kern -0.25pt \scalebox{0.87}{$\mathsf{ any }$} }   \\
                                                                             \begin{array}{l!{\,\color{gray}\vrule}lll}
                                                                               &      &   \textsfb{var} ~  \mathsf{ x0 }  \nobreak\hspace{0pt}=\nobreak\hspace{0pt}  \mathsf{ x }    \\  \,  &      &   \textsfb{var} ~  \mathsf{ el }  \nobreak\hspace{0pt}=\nobreak\hspace{0pt} 32   \\  \,  &      &   \textsfb{branch} ~   \mathsf{ x }    =   \textsfb{nil}  ~  \mathsf{L\kern -0.25pt \scalebox{0.83}{$\mathsf{ 4 }$} }  ~  \mathsf{L\kern -0.25pt \scalebox{0.83}{$\mathsf{ 3 }$} }    \\  \,  &   \mathsf{L\kern -0.25pt \scalebox{0.83}{$\mathsf{ 4 }$} }   &    \mathsf{ x }  \nobreak\hspace{0pt} \leftarrow \nobreak\hspace{0pt}   \mathsf{ x }  [{ 0 }]    \\  \,  &      &  \textsfb{assume} ~   \mathsf{ x }    =   1  ~ \ottkw{else} ~    \scalebox{0.97}{$\mathsf{ size }$}  .\kern -0.5pt  \mathsf{V\kern -0.25pt \scalebox{0.87}{$\mathsf{ b }$} }  .\kern -0.5pt  \mathsf{L\kern -0.25pt \scalebox{0.83}{$\mathsf{ 1 }$} }   ~  [{  \mathsf{ x }   \ottsym{=}   \mathsf{ x0 }   \kern 0.033em}]    \\  \,  &      &  \textsfb{return} \,  1   \ottsym{*}    \mathsf{ el }    \\  \,  &   \mathsf{L\kern -0.25pt \scalebox{0.83}{$\mathsf{ 3 }$} }   & \dots \\  
                                                                             \end{array} \vspace{0.25em} \\  \,  \hspace{2.5mm}  \mathsf{V\kern -0.25pt \scalebox{0.87}{$\mathsf{ b }$} }   ~\dots~ \\  
                                                                   \end{array} \\  
                                                          \end{array} } 
  \end{array}
\]
  \caption{Moving an \assume forward in the instruction stream.}\label{fig:move-example}
\end{figure}

There is an interfering instruction at $ \mathsf{L\kern -0.25pt \scalebox{0.83}{$\mathsf{ 4 }$} } $ that modifies
$  \mathsf{ x }  $.  By creating a temporary variable to hold the value of
$  \mathsf{ x }  $ at the original \assume location, it is possible to resolve
the interference.  Now the \assume can move inside the branch and a
predicate can be added on the updated $  \mathsf{ x }  $ (see right side of the
figure).  Note that the target is unchanged.  This approach allows for the
(logical) separation between the deoptimization point and the position of
assumption predicates.  In the transformed example a stable
deoptimization point is established at the beginning of the function by
storing the value of $  \mathsf{ x }  $, but then the assumption is checked only
in one branch.  The intermediate states are ephemeral and can be safely
discarded when deoptimizing.  For example the variable $  \mathsf{ el }  $ is not
mentioned in the varmap here, it is not captured by the \assume.  Instead it
is recomputed by the original code at the deoptimization target
$  \scalebox{0.97}{$\mathsf{ size }$}  .\kern -0.5pt  \mathsf{V\kern -0.25pt \scalebox{0.87}{$\mathsf{ b }$} }  .\kern -0.5pt  \mathsf{L\kern -0.25pt \scalebox{0.83}{$\mathsf{ 1 }$} }  $.  To be able to deoptimize from any position it is
sufficient to have an \assume after every side-effecting instruction, call,
and control-flow merge.

\subsection{Predicate Hoisting}
\label{sec:hoist-assumpt}

Moving an \assume backwards in the code would require replaying the
moved-over instructions in the case of deoptimization.  Hoisting
$  \textsfb{assume} ~ \textsfb{true} ~ \ottkw{else} ~    \scalebox{0.97}{$\mathsf{ size }$}  .\kern -0.5pt  \mathsf{V\kern -0.25pt \scalebox{0.87}{$\mathsf{ b }$} }  .\kern -0.5pt  \mathsf{L\kern -0.25pt \scalebox{0.83}{$\mathsf{ 2 }$} }   ~  [{  \mathsf{ el }   \ottsym{=}   \mathsf{ el }   \ottsym{,}   \dots   \kern 0.033em}]    $ above $  \textsfb{var} ~  \mathsf{ el }  \nobreak\hspace{0pt}=\nobreak\hspace{0pt} 32  $ is allowed if the varmap is changed to $[  \mathsf{ el }   \ottsym{=}  32  \ottsym{,}   \dots  ]$ to
compensate for the lost definition.  However this approach is tricky and
does not work for instructions with multiple predecessors as it could lead
to conflicting compensation code.  But a simple alternative to hoisting
\assume is to hoist a \emph{predicate} from one \assume to a previous one.
To understand why, let us decompose the approach into two steps.  Given an
\assume at $ \mathsf{L\kern -0.25pt \scalebox{0.83}{$\mathsf{ 1 }$} } $ that dominates a second one at $ \mathsf{L\kern -0.25pt \scalebox{0.83}{$\mathsf{ 2 }$} } $, we copy a
predicate from the latter to the former.  This is valid since the assumption
transparency invariant allows strengthening predicates.  A data-flow
analysis can determine if the copied predicate from $ \mathsf{L\kern -0.25pt \scalebox{0.83}{$\mathsf{ 1 }$} } $ is available
at $ \mathsf{L\kern -0.25pt \scalebox{0.83}{$\mathsf{ 2 }$} } $, in which case it can be removed from the original
instruction.  In our running example, version $  \mathsf{V\kern -0.25pt \scalebox{0.87}{$\mathsf{ pruned }$} }  $ in
\autoref{fig:pruned-example} has two \assume instructions and one predicate.
It is trivial to hoist $  \mathsf{ x }    \neq   \textsfb{nil} $, since there are no interfering
instructions.  This allows us to remove the \assume with the larger
scope. More interestingly, in the case of a loop-invariant assumption,
predicates can be hoisted out of the loop.

\subsection{Assume Composition}
\label{sec:checkpoint-combine}

As we have argued in \autoref{sec:new-version}, it is beneficial to
undo as few assumptions as possible.
On the other hand, deoptimizing an assumption added in an early version
cascades through all the later versions.
To be able to remove chained \assume instructions, we show
that assumptions are \emph{composable}.  If an \assume in version~$ \mathsf{V\kern -0.5pt \scalebox{0.83}{$\mathsf{ 3 }$} } $
transfers control to a target $ \mathsf{V\kern -0.5pt \scalebox{0.83}{$\mathsf{ 2 }$} } . \mathsf{L\kern -0.25pt \scalebox{0.87}{$\mathsf{ a }$} } $, that is itself an
assumption with $ \mathsf{V\kern -0.5pt \scalebox{0.83}{$\mathsf{ 1 }$} } . \mathsf{L\kern -0.25pt \scalebox{0.87}{$\mathsf{ b }$} } $ as target, then we can combine the
metadata to take both steps at once.  By the assumption transparency
invariant, the pre- and post-deoptimization states are equivalent: even if
the assumptions are not the same, it is correct to conservatively trigger
the second deoptimization.  For example, an instruction $  \textsfb{assume} ~ \ottnt{e} ~ \ottkw{else} ~    \mathsf{F}  .\kern -0.5pt  \mathsf{V\kern -0.5pt \scalebox{0.83}{$\mathsf{ 2 }$} }  .\kern -0.5pt  \mathsf{L\kern -0.25pt \scalebox{0.87}{$\mathsf{ a }$} }   ~  [{  \mathsf{ x }   \ottsym{=}  1  \kern 0.033em}]    $ that jumps to $  \textsfb{assume} ~ \ottnt{e'} ~ \ottkw{else} ~    \mathsf{F}  .\kern -0.5pt  \mathsf{V\kern -0.5pt \scalebox{0.83}{$\mathsf{ 0 }$} }  .\kern -0.5pt  \mathsf{L\kern -0.25pt \scalebox{0.87}{$\mathsf{ b }$} }   ~  [{  \mathsf{ y }   \ottsym{=}   \mathsf{ x }   \kern 0.033em}]    $ can be combined as $  \textsfb{assume} ~ \ottnt{e}  \ottsym{,}  \ottnt{e'} ~ \ottkw{else} ~    \mathsf{F}  .\kern -0.5pt  \mathsf{V\kern -0.5pt \scalebox{0.83}{$\mathsf{ 0 }$} }  .\kern -0.5pt  \mathsf{L\kern -0.25pt \scalebox{0.87}{$\mathsf{ b }$} }   ~  [{  \mathsf{ y }   \ottsym{=}  1  \kern 0.033em}]    $.  This new unified \assume skips the
intermediate version~$ \mathsf{V\kern -0.5pt \scalebox{0.83}{$\mathsf{ 2 }$} } $ and goes to $ \mathsf{V\kern -0.5pt \scalebox{0.83}{$\mathsf{ 0 }$} } $ directly.  This could
be an interesting approach for multi-tier JITs: after the system stabilizes,
intermediate versions are rarely used and may be discarded.

\subsection{Case Study}\label{sec:case-study}

\begin{figure}[!b]
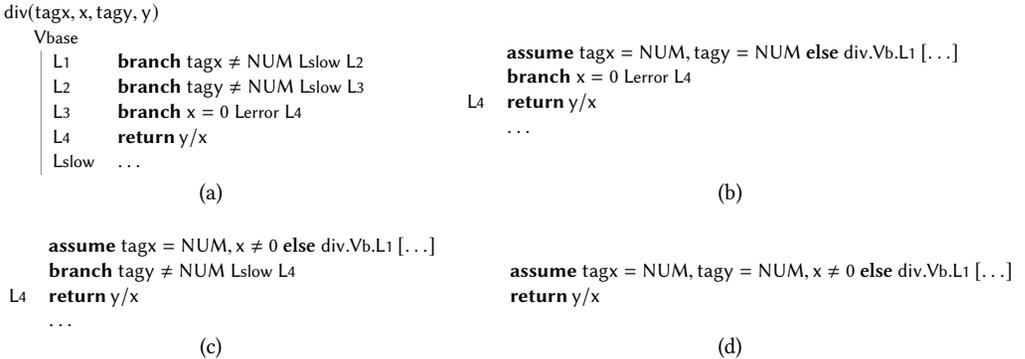

\begin{table}[H]
\centering
\resizebox{1.03\textwidth}{!}{
\begin{tabular}{cc}
  $ \hspace{-5mm}  {\small
                                                          \begin{array}{l}
                                                              \scalebox{0.97}{$\mathsf{ div }$}   (   \mathsf{ tagx }   \ottsym{,}   \mathsf{ x }   \ottsym{,}   \mathsf{ tagy }   \ottsym{,}   \mathsf{ y }   \kern 0.04em ) \\
                                                                   \begin{array}{llll}
                                                                     \kern 1.5pt   \hspace{2.5mm}  \mathsf{V\kern -0.25pt \scalebox{0.87}{$\mathsf{ base }$} }   \\
                                                                             \begin{array}{l!{\,\color{gray}\vrule}lll}
                                                                               &   \mathsf{L\kern -0.25pt \scalebox{0.83}{$\mathsf{ 1 }$} }   &   \textsfb{branch} ~   \mathsf{ tagx }    \neq    \mathsf{ NUM }   ~  \mathsf{L\kern -0.25pt \scalebox{0.87}{$\mathsf{ slow }$} }  ~  \mathsf{L\kern -0.25pt \scalebox{0.83}{$\mathsf{ 2 }$} }    \\  \,  &   \mathsf{L\kern -0.25pt \scalebox{0.83}{$\mathsf{ 2 }$} }   &   \textsfb{branch} ~   \mathsf{ tagy }    \neq    \mathsf{ NUM }   ~  \mathsf{L\kern -0.25pt \scalebox{0.87}{$\mathsf{ slow }$} }  ~  \mathsf{L\kern -0.25pt \scalebox{0.83}{$\mathsf{ 3 }$} }    \\  \,  &   \mathsf{L\kern -0.25pt \scalebox{0.83}{$\mathsf{ 3 }$} }   &   \textsfb{branch} ~   \mathsf{ x }    =   0  ~  \mathsf{L\kern -0.25pt \scalebox{0.87}{$\mathsf{ error }$} }  ~  \mathsf{L\kern -0.25pt \scalebox{0.83}{$\mathsf{ 4 }$} }    \\  \,  &   \mathsf{L\kern -0.25pt \scalebox{0.83}{$\mathsf{ 4 }$} }   &  \textsfb{return} \,   \mathsf{ y }    \ottsym{/}    \mathsf{ x }    \\  \,  &   \mathsf{L\kern -0.25pt \scalebox{0.87}{$\mathsf{ slow }$} }   & \dots \\  
                                                                             \end{array} \vspace{0.25em} \\  
                                                                   \end{array} \\  
                                                          \end{array} }  $
  & \hspace{-15mm}
$  {\small \begin{array}{lll}
                                                      &      &  \textsfb{assume} ~   \mathsf{ tagx }    =    \mathsf{ NUM }    \ottsym{,}    \mathsf{ tagy }    =    \mathsf{ NUM }   ~ \ottkw{else} ~    \scalebox{0.97}{$\mathsf{ div }$}  .\kern -0.5pt  \mathsf{V\kern -0.25pt \scalebox{0.87}{$\mathsf{ b }$} }  .\kern -0.5pt  \mathsf{L\kern -0.25pt \scalebox{0.83}{$\mathsf{ 1 }$} }   ~  [{  \dots   \kern 0.033em}]    \\  \,  &      &   \textsfb{branch} ~   \mathsf{ x }    =   0  ~  \mathsf{L\kern -0.25pt \scalebox{0.87}{$\mathsf{ error }$} }  ~  \mathsf{L\kern -0.25pt \scalebox{0.83}{$\mathsf{ 4 }$} }    \\  \,  &   \mathsf{L\kern -0.25pt \scalebox{0.83}{$\mathsf{ 4 }$} }   &  \textsfb{return} \,   \mathsf{ y }    \ottsym{/}    \mathsf{ x }    \\  \,  &      & \dots \\  
                                                \end{array} }  $
  \\ (a) & \hspace{-6mm} (b)
  \\ &  \\
$  {\small \begin{array}{lll}
                                                      &      &  \textsfb{assume} ~   \mathsf{ tagx }    =    \mathsf{ NUM }    \ottsym{,}    \mathsf{ x }    \neq   0  ~ \ottkw{else} ~    \scalebox{0.97}{$\mathsf{ div }$}  .\kern -0.5pt  \mathsf{V\kern -0.25pt \scalebox{0.87}{$\mathsf{ b }$} }  .\kern -0.5pt  \mathsf{L\kern -0.25pt \scalebox{0.83}{$\mathsf{ 1 }$} }   ~  [{  \dots   \kern 0.033em}]    \\  \,  &      &   \textsfb{branch} ~   \mathsf{ tagy }    \neq    \mathsf{ NUM }   ~  \mathsf{L\kern -0.25pt \scalebox{0.87}{$\mathsf{ slow }$} }  ~  \mathsf{L\kern -0.25pt \scalebox{0.83}{$\mathsf{ 4 }$} }    \\  \,  &   \mathsf{L\kern -0.25pt \scalebox{0.83}{$\mathsf{ 4 }$} }   &  \textsfb{return} \,   \mathsf{ y }    \ottsym{/}    \mathsf{ x }    \\  \,  &      & \dots \\  
                                                \end{array} }  $
  & \hspace{-6mm}
$  {\small \begin{array}{lll}
                                                      &   \phantom{  \mathsf{L\kern -0.25pt \scalebox{0.83}{$\mathsf{ 4 }$} }  }   &  \textsfb{assume} ~   \mathsf{ tagx }    =    \mathsf{ NUM }    \ottsym{,}    \mathsf{ tagy }    =    \mathsf{ NUM }    \ottsym{,}    \mathsf{ x }    \neq   0  ~ \ottkw{else} ~    \scalebox{0.97}{$\mathsf{ div }$}  .\kern -0.5pt  \mathsf{V\kern -0.25pt \scalebox{0.87}{$\mathsf{ b }$} }  .\kern -0.5pt  \mathsf{L\kern -0.25pt \scalebox{0.83}{$\mathsf{ 1 }$} }   ~  [{  \dots   \kern 0.033em}]    \\  \,  &      &  \textsfb{return} \,   \mathsf{ y }    \ottsym{/}    \mathsf{ x }    \\  
                                                \end{array} }  $
  \\ (c) & \hspace{-6mm} (d)
\end{tabular}}
\end{table}
\caption{Case study.}\label{case}
\end{figure}

We conclude with an example.  In dynamic languages code is often dispatched
on runtime types.  If types were known, code could be specialized, resulting
in faster code with fewer checks and branches. Consider Figure \ref{case}(a)
which implements a generic binary division function that expects two values
and their type tags.
No static information is available; the arguments could be any type.
Therefore, multiple checks are needed before the division; for example the
slow branch will require even more checks on the exact value of the type
tag.  Suppose there is profiling information that indicates numbers can be
expected.  The function is specialized by speculatively pruning the
branches as shown in Figure \ref{case}(b).
In certain cases, \sourir's transformations can make it appear as though
checks have been reordered.  Consider a variation of the previous example,
that speculates on $  \mathsf{ x }  $, but not $  \mathsf{ y }  $ as shown in
Figure \ref{case}(c).
In this version, both checks on $  \mathsf{ x }  $ are performed first and then the
ones on $  \mathsf{ y }  $, whereas in the unoptimized version they are
interleaved.  By ruling out an exception early, it is possible to perform
the checks in a more efficient order.  The fully speculated on version
contains only the integer division and the required assumptions
(Figure \ref{case}(d)).
This version has no more branches and is a candidate for inlining.

\section{Speculative Compilation Formalized} \label{sec:sourir-intro-formal}

A \sourir program contains several functions, each of which can have
multiple versions.  This high-level structure is described in
\autoref{fig:syntax-prg}.  The first version is considered the currently
active version and will be executed by a call instruction.  Each version
consists of a stream of labeled instructions.  We use an indentation-based
syntax that directly reflects this structure and omit unreferenced instruction
labels. 

\begin{figure}[H]
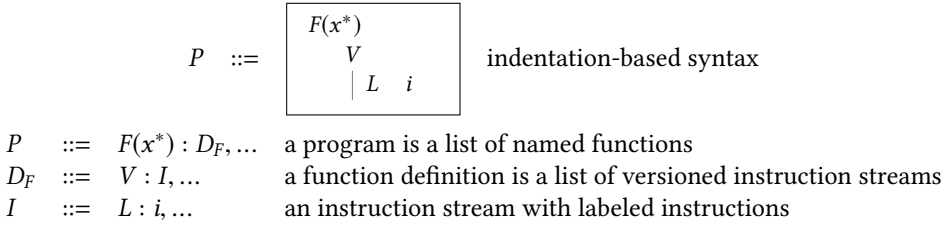

\[
\begin{array}{lrll}
  P & ::= & \boxed{ 
   {\small
                                                          \begin{array}{l}
                                                             \ottmv{F}  (  x^*  \kern 0.04em ) \\
                                                                   \begin{array}{llll}
                                                                     \kern 1.5pt   \hspace{2.5mm} \ottmv{V}  \\
                                                                             \begin{array}{l!{\,\color{gray}\vrule}lll}
                                                                               &  \ottmv{L}  &  \ottnt{i}  \\  
                                                                             \end{array} \vspace{0.25em} \\  
                                                                   \end{array} \\  
                                                          \end{array} }  } & \text{indentation-based syntax}
\end{array}
\]
\[
\begin{array}{lrll}
  \ottnt{P}        & ::=    &  \ottmv{F} (  x^*  ) :  \mathit{D_F} , ...
  & \text{a program is a list of named functions} \\
  \mathit{D_F}  & ::=    &  \ottmv{V}  :  \ottnt{I} , ...
  & \text{a function definition is a list of versioned instruction streams} \\
  \ottnt{I}        & ::=    & \ottmv{L} : \ottnt{i}, ...
  & \text{an instruction stream with labeled instructions}\\
\end{array}
\]
\caption{Program syntax.}\label{fig:syntax-prg}
\end{figure}

Besides grammatical and scoping validity, we impose the following
well-formedness requirements to ease analysis and reasoning.  The last
instruction of each version of the $  \mathsf{main}  $ function is
$ \textsfb{stop} $.  Two variable declarations for the same name cannot occur in the
same instruction stream.  This simplifies reasoning by letting us use
variable names to unambiguously track information depending on the
declaration site. Different versions have separate scopes and can have names
in common.  If a function reference $ { \ottmv{F} } $ is used,
that function $\ottmv{F}$ must exist.  Source
and target of control-flow transitions must have the same set of declared
variables. This eases determining the environment at any point.  To
jump to a label $\ottmv{L}$, all variables not in scope at $\ottmv{L}$ must be
dropped ($\textsfb{drop} \,  \mathsf{ x } $).

\subsection{Operational Semantics: Expressions}

\autoref{fig:semantics-expressions} gives the semantics of
expressions. Evaluation $\ottnt{e}$ returns a value $\ottnt{v}$, which may be a
literal $\ottnt{lit}$, a function, or an address $a$.  Arrays are
represented by addresses into heap $\ottnt{M}$. The heap is a map from addresses
to blocks of values $\ottsym{[}  \ottnt{v_{{\mathrm{1}}}}  \ottsym{,} \, .. \, \ottsym{,}  \ottnt{v_{\ottmv{n}}}  \ottsym{]}$. An environment $\ottnt{E}$ is a mapping
from variables to values. Evaluation is defined by a relation $ \ottnt{M} ~ \ottnt{E} ~ \ottnt{e}   \rightarrow   \ottnt{v} $: under  $\ottnt{M}$ and environment $\ottnt{E}$, $\ottnt{e}$ evaluates to
$\ottnt{v}$. This definition in turn relies on a relation $ \ottnt{E} ~ \ottnt{se}   \rightharpoonup   \ottnt{v} $
defining evaluation of simple expressions $\ottnt{se}$, which does not access
arrays.  The notation $ [\![ primop ]\!] $ to denote, for each primitive
operation $ primop $, a partial function on values.  Arithmetic operators
and arithmetic comparison operators are only defined when their arguments
are numbers. Equality and inequality are defined for all values.
The relation $ \ottnt{M} ~ \ottnt{E} ~ \ottnt{e}   \rightarrow   \ottnt{v} $, when seen as a function from $\ottnt{M}$,
$\ottnt{E}$, $\ottnt{e}$ to $\ottnt{v}$, is partial: it is not defined on all inputs.
For example, there is no $\ottnt{v}$ such that the relation $ \ottnt{M} ~ \ottnt{E} ~   \mathsf{ x }  [{ \ottnt{se} }]    \rightarrow   \ottnt{v} $ holds if $ \ottnt{E}  (   \mathsf{ x }   ) $ is not an address $a$, if $a$ is
not bound in $\ottnt{M}$, if $\ottnt{se}$ does not reduce to a number $n$, or if $n$
is out of bounds.

\begin{figure}[t]
\begin{small}
\begin{mathpar}
\grammartabularSTY{
  \ottv{}\ottinterrule
 }

\begin{array}{ll}
addr ::= a & \text{addresses} \\
\ottnt{M} ::= (\mathit{a} \, \rightarrow  \ottsym{[}  \ottnt{v_{{\mathrm{1}}}}  \ottsym{,} \, .. \, \ottsym{,}  \ottnt{v_{\ottmv{n}}}  \ottsym{]})^* & \text{heap} \\
\ottnt{E} ::= ( \mathsf{ x }   \rightarrow  \ottnt{v})^* & \text{environment} \\
\end{array}
\end{mathpar}

\ottusedrule{\ottdruleLiteral{}}
\hskip 2cm
\ottusedrule{\ottdruleFunref{}}
\hskip 2cm
\ottusedrule{\ottdruleLookup{}}

\vskip 5mm

\ottusedrule{\ottdruleSimpleExp{}}
\hskip 2cm
\ottusedrule{\ottdrulePrimop{}}

\vskip 5mm

\ottusedrule{\ottdruleVecLen{}}
\hskip 2cm
\ottusedrule{\ottdruleVecAccess{}}

\end{small}
  \caption{Evaluation \( \ottnt{M} ~ \ottnt{E} ~ \ottnt{e}   \rightarrow   \ottnt{v} \) of expressions and \( \ottnt{E} ~ \ottnt{se}   \rightharpoonup   \ottnt{v} \) of simple expressions.}
\label{fig:semantics-expressions}
\end{figure}

\subsection{Operational Semantics: Instructions and Programs}

We define a small-step, labeled operational semantics with a notion of
machine state, or configuration, that represents the dynamic state of a
program being executed, and a transition relation between configurations.  A
configuration is a six-component tuple $ \langle  \ottnt{P} \, \ottnt{I} \, \ottmv{L} \, K^* \; \ottnt{M} \, \ottnt{E}  \rangle $ described in
\autoref{fig:semantics-state}.
Continuations $\ottnt{K}$ are tuples of the form $ \langle  \ottnt{I} ~ \ottmv{L} ~ \mathit{x} ~ \ottnt{E}  \rangle $,
storing the information needed to correctly return to a caller function.  On
a call $ \textsfb{call} ~  \mathsf{ x }  \nobreak\hspace{0pt}=\nobreak\hspace{0pt} \ottnt{e} ( \ottnt{e_{{\mathrm{1}}}}  \ottsym{,} \, .. \, \ottsym{,}  \ottnt{e_{\ottmv{n}}} ) $, the continuation pushed on the stack
contains the current instruction stream $\ottnt{I}$ (to be restored on return),
the label $\ottmv{L}$ of the next instruction after the call (the return label),
the variable $\mathit{x}$ to name the returned result, and 
environment $\ottnt{E}$.  For the details, see the reduction rules for
$ \textsfb{call} $ and $ \textsfb{return} $ in \autoref{fig:semantics-reduction}.

\begin{figure}[h]\begin{small}\begin{mathpar}\begin{array}{ll}
\begin{array}{c}
  \ottnt{C} ::=  \langle  \ottnt{P} \, \ottnt{I} \, \ottmv{L} \, K^* \; \ottnt{M} \, \ottnt{E}  \rangle   \\ \text{configuration}
\end{array}
&
\left(
\begin{array}{lrll}
  \ottnt{P} & &                   & \text{program} \\
  \ottnt{I} & &                   & \text{instructions} \\
  \ottmv{L} & &                   & \text{next label} \\
  K^*&::=& (\ottnt{K_{{\mathrm{1}}}}  \ottsym{,} \, .. \, \ottsym{,}  \ottnt{K_{\ottmv{n}}})  & \text{call stack} \\
  \ottnt{M} & &                   & \text{heap} \\
  \ottnt{E} & &                   & \text{environment} \\
\end{array}
\right.
\\ & \\
\begin{array}{c}
  \ottnt{K} ::=  \langle  \ottnt{I} ~ \ottmv{L} ~ \mathit{x} ~ \ottnt{E}  \rangle    \\ \text{continuation}
\end{array}
&
\left(
\begin{array}{lrll}
  \ottnt{I} & &  & \text{code of calling function} \\
  \ottmv{L} & &  & \text{return label} \\
  \mathit{x} & & & \text{return variable} \\
  \ottnt{E} & &  & \text{environment at call site} \\
\end{array}
\right.
\end{array}\end{mathpar}\end{small}
\vskip -3mm
\caption{Abstract machine state.}\label{fig:semantics-state}
\end{figure}

\begin{figure}[h]
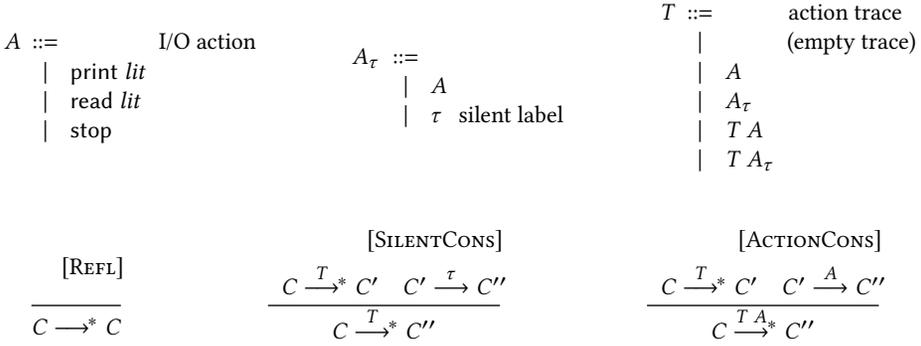

\begin{small}\begin{mathpar}
\grammartabularSTY{  \ottAct{}\ottinterrule }

\grammartabularSTY{  \ottActopt{}\ottinterrule }

\grammartabularSTY{  \ottT{}\ottinterrule }
\end{mathpar}

\ottusedrule{\ottdruleRefl{}}
\hskip 1.4cm
\ottusedrule{\ottdruleSilentCons{}}
\hskip 1.4cm
\ottusedrule{\ottdruleActionCons{}}

\end{small}
\vskip -1mm
\caption{Actions and traces.}\label{fig:semantics-actions}
\end{figure}

\noindent
The relation $ \ottnt{C}  \nto[  \mathit{A}_\tau  ]  \ottnt{C'} $ specifies that executing the next
instruction may result in the configuration $\ottnt{C'}$.  The action
$\mathit{A}_\tau$ indicates whether this reduction is observable: it is either
the silent action, written $ \tau $, an I/O action
$  \textsf{read}~ \ottnt{lit}  $ or $  \textsf{print}~ \ottnt{lit}  $, or $  \textsf{stop}  $.
We write $ \ottnt{C}  \tto[  \mathit{T}  ]  \ottnt{C'} $ when there are zero or more steps from $\ottnt{C}$
to $\ottnt{C'}$.  The trace $\mathit{T}$ is a list of non-silent actions in the order
in which they appeared. 
Actions are defined in \autoref{fig:semantics-actions}, and the full
reduction relation is given in \autoref{fig:semantics-reduction}.

\begin{figure}
\begin{Small}
\begin{mathpar}
\ottusedrule{\ottdruleDecl{}}
\quad
\ottusedrule{\ottdruleDrop{}}

\ottusedrule{\ottdruleArrayDef{}}
\quad
\ottusedrule{\ottdruleArrayDecl{}}

\ottusedrule{\ottdruleUpdate{}}
\quad
\ottusedrule{\ottdruleArrayUpdate{}}

\ottusedrule{\ottdruleRead{}}
\quad
\ottusedrule{\ottdrulePrint{}}

\ottusedrule{\ottdruleBranchT{}}
\quad
\ottusedrule{\ottdruleBranchF{}}

\ottusedrule{\ottdruleGoto{}}
\quad
\ottusedrule{\ottdruleStop{}}

\ottusedrule{\ottdruleCall{}}
\quad
\ottusedrule{\ottdruleReturn{}}

\ottusedrule{\ottdruleAssumePass{}}
\quad
\ottusedrule{\ottdruleAssumeDeopt{}}

\ottusedrule{\ottdruleDeoptimizeConf{}}

\ottusedrule{\ottdruleEvalEnv{}}
\end{mathpar}
\end{Small}
  \caption{Reduction relation \( \ottnt{C}  \nto[   \tau   ]  \ottnt{C'} \) for \sourir{} IR.}
\label{fig:semantics-reduction}
\end{figure}

Most rules get the current instruction, $ \ottnt{I} ( \ottmv{L} ) $, perform an
operation, and advance to the next label, referred to by the shorthand
\( ( \ottmv{L} \kern-1pt+\kern-2pt1) \).  The $  \textsf{read}~ \ottnt{lit}  $ and $  \textsf{print}~ \ottnt{lit}  $
actions represent observable I/O operations.  They are emitted by
\textsc{Read} and \textsc{Print} in \autoref{fig:semantics-reduction}.  The
action $  \textsf{read}~ \ottnt{lit}  $ on the $\textsfb{read} \, \mathit{x}$ transition may be any
literal value.  This is the only reduction rule that is non-deterministic.
Note that the relation $ \ottnt{C}  \tto[    ]  \ottnt{C'} $, containing only sequences of
silent reductions, is deterministic.
The $ \textsfb{stop} $ reduction emits the $  \textsf{stop}  $ transition, and also
produces a configuration with no instructions, $\emptyset$.
This is a technical device to ensure that the resulting configuration is
stuck. A program with a silent loop has a different trace from a program
that halts.
Given a program $\ottnt{P}$, let $ \mathsf{start}( \ottnt{P} ) $ be its starting
configuration, and $ \mathsf{reachable}( \ottnt{P} ) $ be the set of configurations
reachable from it; they are all the states that may be encountered during a
valid run of $\ottnt{P}$.
\begin{mathpar}
\ottusedrule{\ottdruleStartConf{}}
\quad
 \mathsf{reachable}( \ottnt{P} ) 
\  \mathrel{\stackrel{\mathsf{def} }{=} } \ %
\{ \ottnt{C} \mid \exists \mathit{T},\   \mathsf{start}( \ottnt{P} )   \tto[  \mathit{T}  ]  \ottnt{C}  \}
\end{mathpar}

\subsection{Equivalence of Configurations: Bisimulation}

We use weak bisimulation to prove equivalence between configurations.  The
idea is to define, for each program transformation, a correspondence
relation $R$ between configurations over the source and transformed
programs.  We show that related configurations have the same observable
behavior, and reducing them results in configurations that are themselves
related.
Two programs are equivalent if their starting configurations are related.

\begin{definition}[Weak Bisimulation]
  \label{def:bisimulation}
  Given programs $\ottnt{P_{{\mathrm{1}}}}$ and $\ottnt{P_{{\mathrm{2}}}}$ and relation $R$ between the
  configurations of $\ottnt{P_{{\mathrm{1}}}}$ and $\ottnt{P_{{\mathrm{2}}}}$, $R$ is a \emph{weak simulation}
  if for any related states $(\ottnt{C_{{\mathrm{1}}}}, \ottnt{C_{{\mathrm{2}}}}) \in R$ and any reduction $ \ottnt{C_{{\mathrm{1}}}}  \nto[  \mathit{A}_\tau  ]  \ottnt{C'_{{\mathrm{1}}}} $ over $\ottnt{P_{{\mathrm{1}}}}$, there exists a reduction $ \ottnt{C_{{\mathrm{2}}}}  \tto[    ~ \mathit{A}_\tau   ]  \ottnt{C'_{{\mathrm{2}}}} $ over $\ottnt{P_{{\mathrm{2}}}}$ such that $(\ottnt{C'_{{\mathrm{1}}}},\ottnt{C'_{{\mathrm{2}}}})$ are
  themselves related by $R$.  Reduction over $\ottnt{P_{{\mathrm{2}}}}$ is allowed to take
  zero or more steps, but not to change the trace.  In other words, the
  diagram on the left below can always be completed into the diagram on the
  right. 
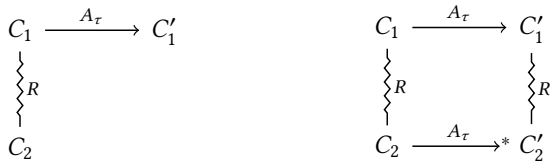
\begin{mathpar}
  \begin{tikzcd}[row sep = large, column sep = large]
    \ottnt{C_{{\mathrm{1}}}}
    \arrow[no head, squiggly, d, "R"]
    \arrow[rar, r, "A_\tau"]
    &
    \ottnt{C'_{{\mathrm{1}}}}
    \\
    \ottnt{C_{{\mathrm{2}}}}
    &
  \end{tikzcd}

  \begin{tikzcd}[row sep = large, column sep = large]
    \ottnt{C_{{\mathrm{1}}}}
    \arrow[no head, squiggly, d, "R"]
    \arrow[rar, r, "A_\tau"]
    &
    \ottnt{C'_{{\mathrm{1}}}}
    \arrow[no head, squiggly, d, "R"]
    \\
    \ottnt{C_{{\mathrm{2}}}}
    \arrow[to*, r, "A_\tau"]
    &
    \ottnt{C'_{{\mathrm{2}}}}
  \end{tikzcd}
\end{mathpar}
$R$ is a weak bisimulation if it is a weak simulation and the symmetric
relation $R^{-1}$ also is---a reduction from $\ottnt{C_{{\mathrm{2}}}}$ can be matched by
$\ottnt{C_{{\mathrm{1}}}}$.  Finally, two configurations are \emph{weakly bisimilar} if
there exists a weak bisimulation $R$ that relates them.
\end{definition}

In the remainder, the adjective weak is always implied.
The following result is standard, and essential to compose the correctness
proof of subsequent transformation passes.
\begin{lemma}[Transitivity]
  If $R_{12}$ is a weak bisimulation between $\ottnt{P_{{\mathrm{1}}}}$ and $\ottnt{P_{{\mathrm{2}}}}$, and
  $R_{23}$ is a weak bisimulation between $\ottnt{P_{{\mathrm{2}}}}$ and $\ottnt{P_{{\mathrm{3}}}}$,
  then the composed relation $R_{13}  \mathrel{\stackrel{\mathsf{def} }{=} }  (R_{12};R_{23})$ is a weak
  bisimulation between $\ottnt{P_{{\mathrm{1}}}}$ and $\ottnt{P_{{\mathrm{3}}}}$.
\end{lemma}

\begin{definition}[Version bisimilarity]\label{def:version-bisimilarity}
  Let $\ottmv{V_{{\mathrm{1}}}}$, $\ottmv{V_{{\mathrm{2}}}}$ be two versions of a function $\ottmv{F}$ in $\ottnt{P}$,
  and let $\ottnt{I_{{\mathrm{1}}}}  \mathrel{\stackrel{\mathsf{def} }{=} }   \ottnt{P} ( \ottmv{F} ,  \ottmv{V_{{\mathrm{1}}}} ) $ and $\ottnt{I_{{\mathrm{2}}}}  \mathrel{\stackrel{\mathsf{def} }{=} }   \ottnt{P} ( \ottmv{F} ,  \ottmv{V_{{\mathrm{2}}}} ) $. $\ottmv{V_{{\mathrm{1}}}}$ and $\ottmv{V_{{\mathrm{2}}}}$ are \emph{(weakly)
    bisimilar} if $ \langle  \ottnt{P} \, \ottnt{I_{{\mathrm{1}}}} \,  \mathsf{start}( \ottnt{I_{{\mathrm{1}}}} )  \, K^* \; \ottnt{M} \, \ottnt{E}  \rangle $ and $ \langle  \ottnt{P} \, \ottnt{I_{{\mathrm{2}}}} \,  \mathsf{start}( \ottnt{I_{{\mathrm{2}}}} )  \, K^* \; \ottnt{M} \, \ottnt{E}  \rangle $ are weakly bisimilar for all $K^*$,
  $\ottnt{M}$, $\ottnt{E}$.
\end{definition}

\begin{definition}[Equivalence]
 $\ottnt{P_{{\mathrm{1}}}}$, $\ottnt{P_{{\mathrm{2}}}}$ are \emph{equivalent} if $ \mathsf{start}( \ottnt{P_{{\mathrm{1}}}} ) $,
  $ \mathsf{start}( \ottnt{P_{{\mathrm{2}}}} ) $ are weakly bisimilar.
\end{definition}

\subsection{Deoptimization Invariants}\label{subsec:formal-invariants}

We can now give a formal definition of the invariants from
\autoref{subsec:invariants}: {\it Version Equivalence} holds if any pair of
versions $(\ottmv{V_{{\mathrm{1}}}}, \ottmv{V_{{\mathrm{2}}}})$ of a function $\ottmv{F}$ are bisimilar; {\it
  Assumption Transparency} holds if for any configuration $\ottnt{C}$, at an
$ \textsfb{assume} ~ e^* ~ \ottkw{else} ~ \xi ~ \tilde{\xi}^* $, $\ottnt{C}$, is bisimilar to
$ \mathsf{deoptimize}(  \ottnt{C} ,  \xi ,  \tilde{\xi}^* ) $, as defined in
\autoref{fig:semantics-reduction}, \textsc{DeoptimizeConf}.

\subsection{Creating Fresh Versions and Injecting Assumptions}

\newcommand{\displace}[3]{#1[#2\nobreak\leftarrow\nobreak#3]}

Configuration $\ottnt{C}$ is \emph{over} location $ \ottmv{F} .\kern -0.5pt \ottmv{V} .\kern -0.5pt \ottmv{L} $ if it is
  $ \langle  \ottnt{P} \,  \ottnt{P} ( \ottmv{F} ,  \ottmv{V} )  \, \ottmv{L} \, K^* \; \ottnt{M} \, \ottnt{E}  \rangle $, where $ \ottnt{P} ( \ottmv{F} ,  \ottmv{V} ) $
denotes the instructions at version~$\ottmv{V}$ of $\ottmv{F}$ in $\ottnt{P}$.  Let
$ \displace{ \ottnt{C} }{  \ottmv{F} .\kern -0.5pt \ottmv{V} .\kern -0.5pt \ottmv{L}  }{  \ottmv{F'} .\kern -0.5pt \ottmv{V'} .\kern -0.5pt \ottmv{L'}  } $ be the configuration
$ \langle  \ottnt{P} \,  \ottnt{P} ( \ottmv{F'} ,  \ottmv{V'} )  \, \ottmv{L'} \, K^* \; \ottnt{M} \, \ottnt{E}  \rangle $.  More generally,
$\displace{\ottnt{C}}{X}{Y}$ replaces various components of $\ottnt{C}$. For
example, $ \displace{ \ottnt{C} }{ \ottnt{P_{{\mathrm{1}}}} }{ \ottnt{P_{{\mathrm{2}}}} } $ updates the program in $\ottnt{C}$; if only the
versions change between two locations $ \ottmv{F} .\kern -0.5pt \ottmv{V} .\kern -0.5pt \ottmv{L} $ and $ \ottmv{F} .\kern -0.5pt \ottmv{V'} .\kern -0.5pt \ottmv{L} $, 
write $ \displace{ \ottnt{C} }{ \ottmv{V} }{ \ottmv{V'} } $ instead of repeating the locations,
etc.

\begin{theorem}
  Creating a new copy of the currently active version of a function,
  possibly adding new \assume instructions, returns an equivalent program.
\end{theorem}

\begin{proof}
  Consider $\ottnt{P_{{\mathrm{1}}}}$ with a function $\ottmv{F}$ with active version
  $\ottmv{V_{{\mathrm{1}}}}$. Adding a version yields $\ottnt{P_{{\mathrm{2}}}}$ with new active version
  $\ottmv{V_{{\mathrm{2}}}}$ of $\ottmv{F}$ such that
  \begin{itemize}
  \item any label $\ottmv{L}$ of $\ottmv{V_{{\mathrm{1}}}}$ exists in $\ottmv{V_{{\mathrm{2}}}}$L: the instruction at
    $\ottmv{L}$ in $\ottmv{V_{{\mathrm{1}}}}$ and $\ottmv{V_{{\mathrm{2}}}}$ are identical except for \assume
    instructions updated so that $ \textsfb{assume} ~ e^* ~ \ottkw{else} ~ \xi ~ \tilde{\xi}^* $ in
    $\ottmv{V_{{\mathrm{1}}}}$ has a corresponding $ \textsfb{assume} ~ e^* ~ \ottkw{else} ~   \ottmv{F} .\kern -0.5pt \ottmv{V_{{\mathrm{1}}}} .\kern -0.5pt \ottmv{L}  ~  \mathsf{Id}   $ in
    $\ottmv{V_{{\mathrm{2}}}}$ where $ \mathsf{Id} $ is the identity over the environment at $\ottmv{L}$.
  \item $\ottmv{V_{{\mathrm{2}}}}$ may contain extra empty \assume instructions: for any
    instruction $\ottnt{i}$ at $\ottmv{L}$ in $\ottmv{V_{{\mathrm{1}}}}$, $\ottmv{V_{{\mathrm{2}}}}$ may contain an
    \assume of the form $ \textsfb{assume} ~ \textsfb{true} ~ \ottkw{else} ~   \ottmv{F} .\kern -0.5pt \ottmv{V_{{\mathrm{1}}}} .\kern -0.5pt \ottmv{L}  ~  \mathsf{Id}   $, where $ \mathsf{Id} $
    is the identity mapping over the environment at $\ottmv{L}$, followed by
    $\ottnt{i}$ at a fresh label $\ottmv{L'}$.
  \end{itemize}
Let us write $\ottnt{I_{{\mathrm{1}}}}$ and $\ottnt{I_{{\mathrm{2}}}}$ for the instructions of $\ottmv{V_{{\mathrm{1}}}}$ and
$\ottmv{V_{{\mathrm{2}}}}$ respectively.  Stack $K^*_{{\mathrm{2}}}$ is a \emph{replacement} of $K^*_{{\mathrm{1}}}$
if it is obtained from $K^*_{{\mathrm{1}}}$ by replacing continuations of the form
$ \langle  \ottnt{I_{{\mathrm{1}}}} ~ \ottmv{L} ~ \mathit{x} ~ \ottnt{E}  \rangle $ by $ \langle  \ottnt{I_{{\mathrm{2}}}} ~ \ottmv{L} ~ \mathit{x} ~ \ottnt{E}  \rangle $.  Replacement is a device used in the
proof and does not correspond to any of the reduction rules.  We define a
relation $R$ as the smallest relation such that :
  \begin{enumerate}
  \item For any configuration $\ottnt{C_{{\mathrm{1}}}}$ over $\ottnt{P_{{\mathrm{1}}}}$,
    $R$ relates $\ottnt{C_{{\mathrm{1}}}}$ to $ \displace{ \ottnt{C_{{\mathrm{1}}}} }{ \ottnt{P_{{\mathrm{1}}}} }{ \ottnt{P_{{\mathrm{2}}}} } $.
  \item For any configuration $\ottnt{C_{{\mathrm{1}}}}$ over a $ \ottmv{F} .\kern -0.5pt \ottmv{V_{{\mathrm{1}}}} .\kern -0.5pt \ottmv{L} $ such that
    $\ottmv{L}$ in $\ottmv{V_{{\mathrm{2}}}}$ is not an added \assume, $R$ relates $\ottnt{C_{{\mathrm{1}}}}$ to
    $ \displace{  \displace{ \ottnt{C_{{\mathrm{1}}}} }{ \ottnt{P_{{\mathrm{1}}}} }{ \ottnt{P_{{\mathrm{2}}}} }  }{ \ottmv{V_{{\mathrm{1}}}} }{ \ottmv{V_{{\mathrm{2}}}} } $.
  \item For any configuration $\ottnt{C_{{\mathrm{1}}}}$ over a $ \ottmv{F} .\kern -0.5pt \ottmv{V_{{\mathrm{1}}}} .\kern -0.5pt \ottmv{L} $ such that at
    $\ottmv{L}$ in $\ottmv{V_{{\mathrm{2}}}}$ is a newly added \assume followed by label $\ottmv{L'}$,
    $R$ relates $\ottnt{C_{{\mathrm{1}}}}$ to both (a) $ \displace{ \ottnt{C_{{\mathrm{1}}}} }{  \ottmv{F} .\kern -0.5pt \ottmv{V_{{\mathrm{1}}}} .\kern -0.5pt \ottmv{L}  }{  \ottmv{F} .\kern -0.5pt \ottmv{V_{{\mathrm{2}}}} .\kern -0.5pt \ottmv{L}  } $
    and (b) $ \displace{ \ottnt{C_{{\mathrm{1}}}} }{  \ottmv{F} .\kern -0.5pt \ottmv{V_{{\mathrm{1}}}} .\kern -0.5pt \ottmv{L}  }{  \ottmv{F} .\kern -0.5pt \ottmv{V_{{\mathrm{2}}}} .\kern -0.5pt \ottmv{L'}  } $.
  \item For any related pair $(\ottnt{C_{{\mathrm{1}}}}, \ottnt{C_{{\mathrm{2}}}}) \in R$, where $K^*_{{\mathrm{1}}}$ is
    the call stack of $\ottnt{C_{{\mathrm{2}}}}$, for any replacement $K^*_{{\mathrm{2}}}$, the pair
    $(\ottnt{C_{{\mathrm{1}}}}, \displace{ \ottnt{C_{{\mathrm{2}}}} }{ K^*_{{\mathrm{1}}} }{ K^*_{{\mathrm{2}}} } )$ is in $R$.
  \end{enumerate}
  The proof proceeds by showing that $R$ is a bisimulation.
  If a related pair $( \ottnt{C_{{\mathrm{1}}}}, \ottnt{C_{{\mathrm{2}}}} ) \in R$ comes from the cases (1), (2)
  or (3) of the definition of $R$, we say that it is a \emph{base pair}.  A
  pair $( \ottnt{C_{{\mathrm{1}}}}, \ottnt{C_{{\mathrm{2}}}} )$ in case (4) is defined from another pair $(
  \ottnt{C_{{\mathrm{1}}}}, \ottnt{C'_{{\mathrm{2}}}} ) \in R$, such that the call stack of $\ottnt{C_{{\mathrm{2}}}}$ is a
  replacement of the stack of $\ottnt{C'_{{\mathrm{2}}}}$. If $( \ottnt{C_{{\mathrm{1}}}}, \ottnt{C'_{{\mathrm{2}}}} ) \in R$ is a
  base pair, we say that it is the base pair of $( \ottnt{C_{{\mathrm{1}}}}, \ottnt{C_{{\mathrm{2}}}}
  )$. Otherwise, we say that the base pair of $( \ottnt{C_{{\mathrm{1}}}}, \ottnt{C_{{\mathrm{2}}}} )$ is the
  base pair of $( \ottnt{C_{{\mathrm{1}}}}, \ottnt{C'_{{\mathrm{2}}}} )$.

  \paragraph{Bisimulation proof: generalities}

  To prove that $R$ is a bisimulation, consider all related pairs
  $(\ottnt{C_{{\mathrm{1}}}}, \ottnt{C_{{\mathrm{2}}}}) \in R$ and show that a reduction from
  $\ottnt{C_{{\mathrm{1}}}}$ can be matched by $\ottnt{C_{{\mathrm{2}}}}$ and conversely.
  Without loss of generality, assume that $\ottnt{C_{{\mathrm{2}}}}$ is not a
  newly added \assume instruction -- that the base pair of $(\ottnt{C_{{\mathrm{1}}}}, \ottnt{C_{{\mathrm{2}}}})$ is not in the case (3,b) of the definition of $R$. Indeed, the proof
  of the case (3,b) follows from proof of the case (3,a). In the case (3,b),
  $\ottnt{C_{{\mathrm{2}}}}$ is a newly added \assume instruction $ \textsfb{assume} ~ \textsfb{true} ~ \ottkw{else} ~  \dots  $ at $\ottmv{L}$ followed by  $\ottmv{L'}$. $\ottnt{C_{{\mathrm{2}}}}$ can only
  reduce silently into $\ottnt{C'_{{\mathrm{2}}}}  \mathrel{\stackrel{\mathsf{def} }{=} }   \displace{ \ottnt{C_{{\mathrm{2}}}} }{ \ottmv{L} }{ \ottmv{L'} } $, which
  is related to $\ottnt{C_{{\mathrm{1}}}}$ by the case (3,a). The empty reduction sequence
  from $\ottnt{C_{{\mathrm{1}}}}$ matches this reduction from $\ottnt{C_{{\mathrm{2}}}}$. Conversely, 
  assume the result in the case (3,a), then any reduction of $\ottnt{C_{{\mathrm{1}}}}$ can
  be matched from $\ottnt{C'_{{\mathrm{2}}}}$, and thus matched from $\ottnt{C_{{\mathrm{2}}}}$ by
  prepending the silent reduction $ \ottnt{C_{{\mathrm{2}}}}  \nto[   \tau   ]  \ottnt{C'_{{\mathrm{2}}}} $ to the matching
  reduction sequence. Finally, if $( \ottnt{C_{{\mathrm{1}}}}, \ottnt{C_{{\mathrm{2}}}} )$ comes from case
  (4) and has a base pair $( \ottnt{C_{{\mathrm{1}}}}, \ottnt{C'_{{\mathrm{2}}}} )$ from (3,b), and $\ottnt{C_{{\mathrm{2}}}}$ has label $\ottmv{L}$ followed by $\ottmv{L'}$, then the bisimulation
  property for $( \ottnt{C_{{\mathrm{1}}}}, \ottnt{C_{{\mathrm{2}}}} ) \in R$ comes from the one of $( \ottnt{C_{{\mathrm{1}}}},  \displace{ \ottnt{C_{{\mathrm{2}}}} }{ \ottmv{L} }{ \ottmv{L'} }  ) \in R$ by the same reasoning.

  \paragraph{Bisimulation proof: easy cases}

  The easy cases of the proof are the reductions $ \ottnt{C_{{\mathrm{1}}}}  \nto[  \mathit{A}_\tau  ]  \ottnt{C'_{{\mathrm{1}}}} $
  where neither $\ottnt{C_{{\mathrm{1}}}}$ nor $\ottnt{C'_{{\mathrm{1}}}}$ are over $\ottmv{V_{{\mathrm{1}}}}$, and the reductions
  $ \ottnt{C_{{\mathrm{2}}}}  \nto[  \mathit{A}_\tau  ]  \ottnt{C'_{{\mathrm{2}}}} $ where neither $\ottnt{C_{{\mathrm{2}}}}$ nor $\ottnt{C'_{{\mathrm{2}}}}$ are over
  $\ottmv{V_{{\mathrm{2}}}}$. For $ \ottnt{C_{{\mathrm{1}}}}  \nto[  \mathit{A}_\tau  ]  \ottnt{C'_{{\mathrm{1}}}} $, define $\ottnt{C'_{{\mathrm{2}}}}$ as
  $ \displace{ \ottnt{C'_{{\mathrm{1}}}} }{ \ottnt{P_{{\mathrm{1}}}} }{ \ottnt{P_{{\mathrm{2}}}} } $, and  both $ \ottnt{C_{{\mathrm{2}}}}  \nto[  \mathit{A}_\tau  ]  \ottnt{C'_{{\mathrm{2}}}} $
  and $(\ottnt{C'_{{\mathrm{1}}}}, \ottnt{C'_{{\mathrm{2}}}}) \in R$ hold. The $ \ottnt{C_{{\mathrm{2}}}}  \nto[  \mathit{A}_\tau  ]  \ottnt{C'_{{\mathrm{2}}}} $
  case is symmetric, defining $\ottnt{C'_{{\mathrm{1}}}}$ as $ \displace{ \ottnt{C'_{{\mathrm{2}}}} }{ \ottnt{P_{{\mathrm{2}}}} }{ \ottnt{P_{{\mathrm{1}}}} } $.

  \paragraph{Bisimulation proof: harder cases}

  The harder cases are split in two categories: version-change
  reductions (deoptimizations, functions call and returns), and
  same-version reductions within $\ottmv{V_{{\mathrm{1}}}}$ in $\ottnt{P_{{\mathrm{1}}}}$ or
  $\ottmv{V_{{\mathrm{2}}}}$ in $\ottnt{P_{{\mathrm{2}}}}$. We consider same-version reductions
  first.
  Without loss of generality, assume that the pair
  $(\ottnt{C_{{\mathrm{1}}}}, \ottnt{C_{{\mathrm{2}}}}) \in R$ is a base pair, that is a pair related
  by the cases (2) or (3) of the definition of $R$, but not (4) -- the
  case that changes the call stack of the configuration. Indeed, if 
 pair $( \ottnt{C_{{\mathrm{1}}}}, \ottnt{C'_{{\mathrm{2}}}} ) \in R$ comes from (4), the
  only difference between this pair and its base pair
  $( \ottnt{C_{{\mathrm{1}}}}, \ottnt{C_{{\mathrm{2}}}} ) \in R$ is in the call stack of $\ottnt{C_{{\mathrm{2}}}}$
  and $\ottnt{C'_{{\mathrm{2}}}}$. This means that $\ottnt{C_{{\mathrm{2}}}}$ and $\ottnt{C'_{{\mathrm{2}}}}$ have the
  exact same reduction behavior for non-version-change reductions. As
  long as the proof that the related configurations $\ottnt{C_{{\mathrm{1}}}}$ and
  $\ottnt{C_{{\mathrm{2}}}}$ match each other does not use version-change reductions
  (a property that holds for the proofs of the non-version-change
  cases below), it also applies to $\ottnt{C_{{\mathrm{1}}}}$ and $\ottnt{C'_{{\mathrm{2}}}}$.
  For a reduction $ \ottnt{C_{{\mathrm{2}}}}  \nto[  \mathit{A}_\tau  ]  \ottnt{C'_{{\mathrm{2}}}} $ that is not a version-change
  reduction (deoptimization, call or return), prove that it can be matched
  from $\ottnt{C_{{\mathrm{1}}}}$ by reasoning on whether $\ottnt{C_{{\mathrm{2}}}}$ or $\ottnt{C'_{{\mathrm{2}}}}$ are \assume
  instructions, coming from $\ottmv{V_{{\mathrm{1}}}}$ or newly added.
  \begin{itemize}
  \item If none of them are \assume instructions, then they are both
    in the case (2) of the definition of $R$, they are equal to
    $ \displace{ \ottnt{C_{{\mathrm{1}}}} }{ \ottmv{V_{{\mathrm{1}}}} }{ \ottmv{V_{{\mathrm{2}}}} } $ and $ \displace{ \ottnt{C'_{{\mathrm{1}}}} }{ \ottmv{V_{{\mathrm{1}}}} }{ \ottmv{V_{{\mathrm{2}}}} } $
    respectively, so  $ \ottnt{C_{{\mathrm{1}}}}  \nto[  \mathit{A}_\tau  ]  \ottnt{C'_{{\mathrm{1}}}} $ and
    $(\ottnt{C'_{{\mathrm{1}}}}, \ottnt{C'_{{\mathrm{2}}}}) \in R$ hold.
  \item If $\ottnt{C_{{\mathrm{2}}}}$ or $\ottnt{C'_{{\mathrm{2}}}}$ are \assume instructions coming from
    $\ottmv{V_{{\mathrm{1}}}}$, the same reasoning holds -- the problematic case where the
    \assume is $\ottnt{C_{{\mathrm{2}}}}$ and the guards do not pass is not considered here as
    the reduction is not a deoptimization.
  \item If $\ottnt{C'_{{\mathrm{2}}}}$ is a newly added \assume in $\ottmv{V_{{\mathrm{2}}}}$ at $\ottmv{L}$
    followed by $\ottmv{L'}$, $\ottnt{C_{{\mathrm{2}}}}$ is an instruction of
    $\ottmv{V_{{\mathrm{2}}}}$ copied from $\ottmv{V_{{\mathrm{1}}}}$, so $(\ottnt{C_{{\mathrm{1}}}},\ottnt{C_{{\mathrm{2}}}})$ are in the case (2)
    of the definition of $R$ and $\ottnt{C_{{\mathrm{1}}}}$ is $ \displace{ \ottnt{C_{{\mathrm{1}}}} }{ \ottmv{V_{{\mathrm{2}}}} }{ \ottmv{V_{{\mathrm{1}}}} } $. The
    reduction from $\ottnt{C_{{\mathrm{2}}}}$ corresponds to a reduction $ \ottnt{C_{{\mathrm{1}}}}  \nto[  \mathit{A}_\tau  ]  \ottnt{C'_{{\mathrm{1}}}} $ in $\ottnt{P_{{\mathrm{1}}}}$ with $\ottnt{C'_{{\mathrm{1}}}}  \mathrel{\stackrel{\mathsf{def} }{=} }   \displace{ \ottnt{C'_{{\mathrm{2}}}} }{ \ottmv{V_{{\mathrm{2}}}} }{ \ottmv{V_{{\mathrm{1}}}} } $, and  $(\ottnt{C'_{{\mathrm{1}}}}, \ottnt{C'_{{\mathrm{2}}}})\in R$ by the case
    (3,a) of the definition of $R$.
  \end{itemize}
  The reasoning for transitions $ \ottnt{C_{{\mathrm{1}}}}  \nto[  \mathit{A}_\tau  ]  \ottnt{C'_{{\mathrm{1}}}} $ that have to be
  matched from $\ottnt{C_{{\mathrm{2}}}}$ and are not version-change transitions
  (deoptimization, function calls or return) is similar. $\ottnt{C_{{\mathrm{2}}}}$ cannot
  be a new \assume, so we have $ \ottnt{C_{{\mathrm{2}}}}  \nto[  \mathit{A}_\tau  ]  \ottnt{C'_{{\mathrm{2}}}} $, and either $\ottnt{C'_{{\mathrm{2}}}}$ is not a new \assume and matches $\ottnt{C_{{\mathrm{1}}}}$ by case (2) of the
  definition of R, or it is a new \assume and it matches it by the case
  (3,a).

  \paragraph{Bisimulation proof: final cases}

  The cases that remain are the hard cases of version-change
  reductions: function call, return and deoptimization.
  If $ \ottnt{C_{{\mathrm{1}}}}  \nto[  \mathit{A}_\tau  ]  \ottnt{C'_{{\mathrm{1}}}} $ is a deoptimization reduction, then
  $\ottnt{C_{{\mathrm{1}}}}$ is over a location $ \ottmv{F} .\kern -0.5pt \ottmv{V_{{\mathrm{1}}}} .\kern -0.5pt \ottmv{L} $ in $\ottnt{P_{{\mathrm{1}}}}$, and its
  instruction is $ \textsfb{assume} ~ e^* ~ \ottkw{else} ~ \xi ~ \tilde{\xi}^* $, and $\ottnt{C'_{{\mathrm{1}}}}$ is
  $ \mathsf{deoptimize}(  \ottnt{C_{{\mathrm{1}}}} ,  \xi ,  \tilde{\xi}^* ) $.  $\ottnt{C_{{\mathrm{2}}}}$ is over the copied
  instruction $ \textsfb{assume} ~ e^* ~ \ottkw{else} ~   \ottmv{F} .\kern -0.5pt \ottmv{V_{{\mathrm{1}}}} .\kern -0.5pt \ottmv{L}  ~  \mathsf{Id}   $ and $ \mathsf{Id} $ is the
  identity. $\ottnt{C_{{\mathrm{2}}}}$ also deoptimizes, given that the tests give the same
  results in the same environment, so we have $ \ottnt{C_{{\mathrm{2}}}}  \nto[   \tau   ]  \ottnt{C'_{{\mathrm{2}}}} $ for
  $\ottnt{C'_{{\mathrm{2}}}}  \mathrel{\stackrel{\mathsf{def} }{=} }   \mathsf{deoptimize}(  \ottnt{C_{{\mathrm{2}}}} ,    \ottmv{F} .\kern -0.5pt \ottmv{V} .\kern -0.5pt \ottmv{L_{{\mathrm{1}}}}  ~  \mathsf{Id}   ,   \emptyset  ) $.
  $\ottnt{C'_{{\mathrm{2}}}}$ is over $ \ottmv{F} .\kern -0.5pt \ottmv{V_{{\mathrm{1}}}} .\kern -0.5pt \ottmv{L} $, that is the same \assume
  instruction as $\ottnt{C_{{\mathrm{1}}}}$, so it also deoptimizes, to $\ottnt{C''_{{\mathrm{2}}}}  \mathrel{\stackrel{\mathsf{def} }{=} } 
   \mathsf{deoptimize}(  \ottnt{C'_{{\mathrm{2}}}} ,  \xi ,  \tilde{\xi}^* ) $. We show that $\ottnt{C'_{{\mathrm{1}}}}$ and
  $\ottnt{C''_{{\mathrm{2}}}}$ are related by $R$:
  \begin{itemize}
  \item If $(\ottnt{C_{{\mathrm{1}}}}, \ottnt{C_{{\mathrm{2}}}}) \in R$ is a base pair, then $\ottnt{C_{{\mathrm{1}}}}$ is
    $ \displace{ \ottnt{C_{{\mathrm{2}}}} }{ \ottmv{V_{{\mathrm{2}}}} }{ \ottmv{V_{{\mathrm{1}}}} } $. In particular, the two configurations have
    the same environment, and $\ottnt{C'_{{\mathrm{2}}}}$ is identical to $\ottnt{C_{{\mathrm{2}}}}$ except it is
    over $ \ottmv{F} .\kern -0.5pt \ottmv{V} .\kern -0.5pt \ottmv{L_{{\mathrm{1}}}} $. It is thus equal to $\ottnt{C_{{\mathrm{1}}}}$. As a consequence,
    $\ottnt{C'_{{\mathrm{1}}}}$ and $\ottnt{C''_{{\mathrm{2}}}}$, which are obtained from $\ottnt{C_{{\mathrm{1}}}}$ and $\ottnt{C'_{{\mathrm{2}}}}$
    by the same deoptimization reduction, are the same configurations, and
    related in $R$.
  \item If $\ottnt{C_{{\mathrm{1}}}}$ and $\ottnt{C_{{\mathrm{2}}}}$ are related by the case (4) of the
    definition of $R$, the stack of $\ottnt{C_{{\mathrm{2}}}}$ is a replacement of the
    stack of $\ottnt{C_{{\mathrm{1}}}}$. The same reasoning as in the previous case
    shows that configurations $\ottnt{C'_{{\mathrm{1}}}}$ and $\ottnt{C''_{{\mathrm{2}}}}$ are identical,
    except that the stack of $\ottnt{C''_{{\mathrm{2}}}}$ is a replacement of the stack
    of $\ottnt{C'_{{\mathrm{1}}}}$: they are related by the case (4) of the definition
    of $R$.
  \end{itemize}
  Conversely, if $ \ottnt{C_{{\mathrm{2}}}}  \nto[  \mathit{A}_\tau  ]  \ottnt{C'_{{\mathrm{2}}}} $ is a deoptimization
  instruction then, by the same reasoning as in the proof of matching
  a deoptimization of $\ottnt{C_{{\mathrm{1}}}}$, $\ottnt{C'_{{\mathrm{2}}}}$ is identical to $\ottnt{C_{{\mathrm{1}}}}$
  (modulo replaced stacks). This means that the empty reduction
  sequence from $\ottnt{C_{{\mathrm{1}}}}$ matches the reduction of $\ottnt{C_{{\mathrm{2}}}}$.

  If $ \ottnt{C_{{\mathrm{1}}}}  \nto[  \mathit{A}_\tau  ]  \ottnt{C'_{{\mathrm{1}}}} $ is a function call transition,
  \[
      \langle  \ottnt{P_{{\mathrm{1}}}} \, \ottnt{I_{{\mathrm{1}}}} \, \ottmv{L} \, K^*_{{\mathrm{1}}} \; \ottnt{M} \, \ottnt{E}  \rangle   \nto[   \tau   ]   \langle  \ottnt{P_{{\mathrm{1}}}} \, \ottnt{I'_{{\mathrm{1}}}} \, \ottmv{L'} \,  (K^*,   \langle  \ottnt{I_{{\mathrm{1}}}} ~  ( \ottmv{L} \kern-1pt+\kern-2pt1)  ~  \mathsf{ x }  ~ \ottnt{E}  \rangle  )  \; \ottnt{M} \, \ottnt{E'}  \rangle  
  \]
  $\ottnt{C_{{\mathrm{2}}}}$ is on the same call with the same arguments,
  so it takes a transition $ \ottnt{C_{{\mathrm{2}}}}  \nto[   \tau   ]  \ottnt{C'_{{\mathrm{2}}}} $ of the form
  \[
      \langle  \ottnt{P_{{\mathrm{2}}}} \, \ottnt{I_{{\mathrm{2}}}} \, \ottmv{L} \, K^*_{{\mathrm{2}}} \; \ottnt{M} \, \ottnt{E}  \rangle   \nto[   \tau   ]   \langle  \ottnt{P_{{\mathrm{2}}}} \, \ottnt{I'_{{\mathrm{2}}}} \, \ottmv{L'} \,  (K^*,   \langle  \ottnt{I_{{\mathrm{2}}}} ~  ( \ottmv{L} \kern-1pt+\kern-2pt1)  ~  \mathsf{ x }  ~ \ottnt{E}  \rangle  )  \; \ottnt{M} \, \ottnt{E'}  \rangle  
  \]
  The stack of $\ottnt{C'_{{\mathrm{2}}}}$ is a replacement of the stack of $\ottnt{C'_{{\mathrm{1}}}}$:
  assuming that $K^*_{{\mathrm{2}}}$ is a replacement of $K^*_{{\mathrm{1}}}$, the
  difference in the new continuation is precisely the definition of
  stack replacement -- note that it is precisely this reasoning step
  that required the addition of case (4) in the definition of
  $R$. Also, the new instruction streams $\ottnt{I'_{{\mathrm{1}}}}$ and $\ottnt{I'_{{\mathrm{2}}}}$ are either
  identical (if the function is not $\ottmv{F}$ itself) or equal to $\ottnt{I_{{\mathrm{1}}}}$ and
  $\ottnt{I_{{\mathrm{2}}}}$ respectively, so we do have $(\ottnt{C'_{{\mathrm{1}}}}, \ottnt{C'_{{\mathrm{2}}}}) \in R$ as
  expected. The proof of the symmetric case, matching a function call
  from $\ottnt{C_{{\mathrm{2}}}}$, is identical.

  If $ \ottnt{C_{{\mathrm{1}}}}  \nto[  \mathit{A}_\tau  ]  \ottnt{C'_{{\mathrm{1}}}} $ is a function return transition
  \[
      \langle  \ottnt{P_{{\mathrm{1}}}} \, \ottnt{I_{{\mathrm{1}}}} \, L \,  (K^*,   \langle  \ottnt{I'_{{\mathrm{1}}}} ~ L' ~ \mathit{x} ~ \ottnt{E'}  \rangle  )  \; \ottnt{M} \, \ottnt{E}  \rangle   \nto[   \tau   ]   \langle  \ottnt{P_{{\mathrm{1}}}} \, \ottnt{I'_{{\mathrm{1}}}} \, L' \, K^*_{{\mathrm{1}}} \; \ottnt{M} \,  \ottnt{E'}  [  \mathit{x}  \leftarrow  \ottnt{v}  ]   \rangle  
  \]
  then $ \ottnt{C_{{\mathrm{2}}}}  \nto[  \mathit{A}_\tau  ]  \ottnt{C'_{{\mathrm{2}}}} $ is also a function return transition
  \[
      \langle  \ottnt{P_{{\mathrm{2}}}} \, \ottnt{I_{{\mathrm{2}}}} \, L \,  (K^*,   \langle  \ottnt{I'_{{\mathrm{2}}}} ~ L' ~ \mathit{x} ~ \ottnt{E'}  \rangle  )  \; \ottnt{M} \, \ottnt{E}  \rangle   \nto[   \tau   ]   \langle  \ottnt{P_{{\mathrm{2}}}} \, \ottnt{I'_{{\mathrm{2}}}} \, L' \, K^*_{{\mathrm{2}}} \; \ottnt{M} \,  \ottnt{E'}  [  \mathit{x}  \leftarrow  \ottnt{v}  ]   \rangle  
  \]
  We have to show that $\ottnt{C'_{{\mathrm{1}}}}$ and $\ottnt{C'_{{\mathrm{2}}}}$ are related by $R$.
  The environments and heaps of the two configurations are identical.
  We know that the stack of $\ottnt{C_{{\mathrm{2}}}}$ is a replacement of the stack
  of $\ottnt{C_{{\mathrm{1}}}}$, which means that $K^*_{{\mathrm{2}}}$ a replacement of 
  $K^*_{{\mathrm{1}}}$, and that either $\ottnt{I'_{{\mathrm{1}}}}$ and $\ottnt{I'_{{\mathrm{2}}}}$ are
  identical or they are respectively equal to $\ottnt{I_{{\mathrm{1}}}}$ and $\ottnt{I_{{\mathrm{2}}}}$.
  In either case, $\ottnt{C'_{{\mathrm{1}}}}$ and $\ottnt{C'_{{\mathrm{2}}}}$ are related
  by $R$. The proof of the symmetric case, matching a function return
  from $\ottnt{C_{{\mathrm{2}}}}$, is identical.
  We have established that $R$ is a bisimulation.

  Finally, remark that our choice of $R$ also proves that the new
  version respects the assumption transparency invariant. A new
  \assume at $\ottmv{L}$ in $\ottmv{V_{{\mathrm{2}}}}$ is of the form
  $ \textsfb{assume} ~ \textsfb{true} ~ \ottkw{else} ~   \ottmv{F} .\kern -0.5pt \ottmv{V_{{\mathrm{1}}}} .\kern -0.5pt \ottmv{L}  ~  \mathsf{Id}   $, with $ \mathsf{Id} $ the identity
  environment.  Any configuration $\ottnt{C}$ over $ \ottmv{F} .\kern -0.5pt \ottmv{V_{{\mathrm{2}}}} .\kern -0.5pt \ottmv{L} $ is
  related by $R^{-1}$ to $ \displace{ \ottnt{C} }{  \ottmv{F} .\kern -0.5pt \ottmv{V_{{\mathrm{2}}}} .\kern -0.5pt \ottmv{L}  }{  \ottmv{F} .\kern -0.5pt \ottmv{V_{{\mathrm{1}}}} .\kern -0.5pt \ottmv{L}  } $,
  which is equal to $ \mathsf{deoptimize}(  \ottnt{C} ,    \ottmv{F} .\kern -0.5pt \ottmv{V_{{\mathrm{1}}}} .\kern -0.5pt \ottmv{L}  ~  \mathsf{Id}   ,   \emptyset  ) $.
  These two configurations are related by the bisimulation $R^{-1}$,
  so they are bisimilar.
\end{proof}

\begin{lemma}
  Adding a new predicate $\ottnt{e'}$ to an existing \assume instruction $  \textsfb{assume} ~ e^* ~ \ottkw{else} ~ \xi ~ \tilde{\xi}^*  $ of $\ottnt{P_{{\mathrm{1}}}}$ returns an equivalent program $\ottnt{P_{{\mathrm{2}}}}$.
\end{lemma}

\begin{proof}
  This is a consequence of the invariant of assumption transparency.
  Let $R_{\ottnt{P_{{\mathrm{1}}}}}$ be the bisimilarity relation for configurations
  over $\ottnt{P_{{\mathrm{1}}}}$, and $ \ottmv{F} .\kern -0.5pt \ottmv{V} .\kern -0.5pt \ottmv{L} $ be the location of the modified
  \assume. Let us define the relation $R$ between $\ottnt{P_{{\mathrm{1}}}}$ and $\ottnt{P_{{\mathrm{2}}}}$ by
  \[
    (\ottnt{C_{{\mathrm{1}}}}, \ottnt{C_{{\mathrm{2}}}}) \in R \quad\iff\quad (\ottnt{C_{{\mathrm{1}}}},  \displace{ \ottnt{C_{{\mathrm{2}}}} }{ \ottnt{P_{{\mathrm{2}}}} }{ \ottnt{P_{{\mathrm{1}}}} } ) \in R_{\ottnt{P_{{\mathrm{1}}}}}
  \]
  We show that $R$ is a bisimulation.
  Consider $(\ottnt{C_{{\mathrm{1}}}}, \ottnt{C_{{\mathrm{2}}}}) \in R$. If $\ottnt{C_{{\mathrm{2}}}}$ is not over
  $ \ottmv{F} .\kern -0.5pt \ottmv{V} .\kern -0.5pt \ottmv{L} $, the reductions of $\ottnt{C_{{\mathrm{2}}}}$ (in $\ottnt{P_{{\mathrm{2}}}}$) and
  $ \displace{ \ottnt{C_{{\mathrm{2}}}} }{ \ottnt{P_{{\mathrm{2}}}} }{ \ottnt{P_{{\mathrm{1}}}} } $ (in $\ottnt{P_{{\mathrm{1}}}}$) are identical, and the
  latter configuration is, by assumption, bisimilar to $\ottnt{C_{{\mathrm{1}}}}$, so
  it is immediate that any reduction from $\ottnt{C_{{\mathrm{1}}}}$ can be matched by
  $\ottnt{C_{{\mathrm{2}}}}$ and conversely.
  If $\ottnt{C_{{\mathrm{2}}}}$ is over $ \ottmv{F} .\kern -0.5pt \ottmv{V} .\kern -0.5pt \ottmv{L} $, we can compare its
  reduction behavior (in $\ottnt{P_{{\mathrm{2}}}}$) with the one of
  $ \displace{ \ottnt{C_{{\mathrm{2}}}} }{ \ottnt{P_{{\mathrm{2}}}} }{ \ottnt{P_{{\mathrm{1}}}} } $ (in $\ottnt{P_{{\mathrm{1}}}}$). The first
  configuration deoptimizes when one of the $e^*  \ottsym{,}  \ottnt{e'}$ is not
  true in the environment of $\ottnt{C_{{\mathrm{2}}}}$, while the second deoptimizes
  when one of the $e^*$ is not true -- in the same
  environment. If $\ottnt{C_{{\mathrm{2}}}}$ gives the same boolean value to both
  series of test, then the two configurations have the same reduction
  behavior, and $(\ottnt{C_{{\mathrm{1}}}}, \ottnt{C_{{\mathrm{2}}}})$ match each other by the same
  reasoning as in the previous paragraph. The only interesting case is
  the configurations $\ottnt{C_{{\mathrm{2}}}}$ that pass all the tests in
  $e^*$, but fail $\ottnt{e'}$. Let us show that, even in that
  case, the reductions of $\ottnt{C_{{\mathrm{1}}}}$ and $\ottnt{C_{{\mathrm{2}}}}$ match each other.
  The following diagram will be useful to follow the proof below:
  \begin{mathpar}
    \begin{tikzcd}[row sep = large, column sep = large]
      \ottnt{C_{{\mathrm{1}}}}
      \arrow[no head, squiggly, d, "R"]
      \arrow[rar, r, "A_\tau"]
      \arrow[no head, squiggly, dr, "R"]
      &
      \ottnt{C'_{{\mathrm{1}}}}
      \arrow[no head, squiggly, dr, "R"]
      \\
      \ottnt{C_{{\mathrm{2}}}}
      \arrow[to, r, "\tau"]
      &
       \mathsf{deoptimize}(  \ottnt{C_{{\mathrm{2}}}} ,  \xi ,  \tilde{\xi}^* ) 
      \arrow[to, r, "A_\tau"]
      &
      \ottnt{C''_{{\mathrm{1}}}}
    \end{tikzcd}
  \end{mathpar}
  Let us first show that the reductions of $\ottnt{C_{{\mathrm{2}}}}$ can be matched
  by $\ottnt{C_{{\mathrm{1}}}}$. The only possible reduction from $\ottnt{C_{{\mathrm{2}}}}$, given
  our assumptions, is
  $ \ottnt{C_{{\mathrm{2}}}}  \nto[   \tau   ]   \mathsf{deoptimize}(  \ottnt{C_{{\mathrm{2}}}} ,  \xi ,  \tilde{\xi}^* )  $. We claim
  that the empty reduction sequence from $\ottnt{C_{{\mathrm{1}}}}$ matches it, that
  is, that
  $( \ottnt{C_{{\mathrm{1}}}},  \mathsf{deoptimize}(  \ottnt{C_{{\mathrm{2}}}} ,  \xi ,  \tilde{\xi}^* ) ) \in R$. By
  definition of $R$, this goal means that
  $\ottnt{C_{{\mathrm{1}}}}$ and $ \displace{  \mathsf{deoptimize}(  \ottnt{C_{{\mathrm{2}}}} ,  \xi ,  \tilde{\xi}^* )  }{ \ottnt{P_{{\mathrm{2}}}} }{ \ottnt{P_{{\mathrm{1}}}} } $
  are bisimilar in $\ottnt{P_{{\mathrm{1}}}}$. But the latter configuration is the same as
  $ \mathsf{deoptimize}(   \displace{ \ottnt{C_{{\mathrm{2}}}} }{ \ottnt{P_{{\mathrm{2}}}} }{ \ottnt{P_{{\mathrm{1}}}} }  ,  \xi ,  \tilde{\xi}^* ) $, which is bisimilar
  to $\ottnt{C_{{\mathrm{2}}}}$ by the invariant of assumption transparency, and thus to $\ottnt{C_{{\mathrm{1}}}}$
  by transitivity.
  Conversely, we show that the reductions of $\ottnt{C_{{\mathrm{1}}}}$ can be matched by
  $\ottnt{C_{{\mathrm{2}}}}$. Suppose a reduction $ \ottnt{C_{{\mathrm{1}}}}  \nto[  \mathit{A}_\tau  ]  \ottnt{C'_{{\mathrm{1}}}} $.
  The configuration
  $ \displace{  \mathsf{deoptimize}(  \ottnt{C_{{\mathrm{2}}}} ,  \xi ,  \tilde{\xi}^* )  }{ \ottnt{P_{{\mathrm{2}}}} }{ \ottnt{P_{{\mathrm{1}}}} } $ is bisimilar to
  $\ottnt{C_{{\mathrm{1}}}}$ (same reasoning as in the previous paragraph), so there is a
  matching state $\ottnt{C''_{{\mathrm{1}}}}$ such that $  \displace{  \mathsf{deoptimize}(  \ottnt{C_{{\mathrm{2}}}} ,  \xi ,  \tilde{\xi}^* )  }{ \ottnt{P_{{\mathrm{2}}}} }{ \ottnt{P_{{\mathrm{1}}}} }   \nto[  \mathit{A}_\tau  ]  \ottnt{C''_{{\mathrm{1}}}} $ in $\ottnt{P_{{\mathrm{1}}}}$ with
  $(\ottnt{C'_{{\mathrm{1}}}}, \ottnt{C''_{{\mathrm{1}}}}) \in R_{\ottnt{P_{{\mathrm{1}}}}}$. We can transpose this reduction in
  $\ottnt{P_{{\mathrm{2}}}}$:  $  \mathsf{deoptimize}(  \ottnt{C_{{\mathrm{2}}}} ,  \xi ,  \tilde{\xi}^* )   \nto[  \mathit{A}_\tau  ]   \displace{ \ottnt{C''_{{\mathrm{1}}}} }{ \ottnt{P_{{\mathrm{1}}}} }{ \ottnt{P_{{\mathrm{2}}}} }  $ in $\ottnt{P_{{\mathrm{2}}}}$, and thus $ \ottnt{C_{{\mathrm{2}}}}  \tto[    ~ \mathit{A}_\tau   ]   \displace{ \ottnt{C''_{{\mathrm{1}}}} }{ \ottnt{P_{{\mathrm{1}}}} }{ \ottnt{P_{{\mathrm{2}}}} }  $. This matches the reduction of $\ottnt{C_{{\mathrm{1}}}}$,
  given that our assumption $( \ottnt{C'_{{\mathrm{1}}}}, \ottnt{C''_{{\mathrm{1}}}} ) \in R_{\ottnt{P_{{\mathrm{1}}}}}$ exactly
  means that $( \ottnt{C'_{{\mathrm{1}}}},  \displace{ \ottnt{C''_{{\mathrm{1}}}} }{ \ottnt{P_{{\mathrm{1}}}} }{ \ottnt{P_{{\mathrm{2}}}} }  ) \in R$.
\end{proof}

\clearpage

\section{Optimization Correctness} \label{sec:optimizations-correctness}

The proofs of the optimizations from \autoref{sec:optimizations}
are easier than the proofs for deoptimization invariants in the previous
section (although, as program transformations, they seem more
elaborate). This comes from the fact that the classical optimizations
rewrite an existing version and interact little with deoptimization.

\subsection{Constant Propagation}

We say that given a version $\ottmv{V}$, a \emph{static environment} $SE$ for
label $\ottmv{L}$ maps a subset of the variables in scope at $\ottmv{L}$ to values.
A static environment is \emph{valid}, written $SE \vDash \ottmv{L}$, if for any
configuration $\ottnt{C}$ over $\ottmv{L}$ reachable from the start of
$\ottmv{V}$ we have that $SE$ is a subset of the lexical environment $\ottnt{E}$.
Constant propagation can use a classic work-queue data-flow algorithm to
compute a valid static environment $SE$ at each label $\ottmv{L}$. It then
replaces, in the instruction at $\ottmv{L}$, each expression or simple
expression that can be evaluated in $SE$ by its value.  This is speculative
since assumption predicates of the form $ \mathit{x}   =   \ottnt{lit} $ populate the
static environment with the binding $\mathit{x}  \rightarrow  \ottnt{lit}$. 

\begin{lemma}\label{lem:constfold}
  For any version $\ottmv{V_{{\mathrm{1}}}}$, let $\ottmv{V_{{\mathrm{2}}}}$ be the result of constant
  propagation. $\ottmv{V_{{\mathrm{1}}}}$ and $\ottmv{V_{{\mathrm{2}}}}$ are bisimilar.
\end{lemma}

\begin{proof}
  The relation $R$ to use here for bisimulation is the one that relates each
  reachable $\ottnt{C_{{\mathrm{1}}}}$ in $ \mathsf{reachable}( \ottnt{P_{{\mathrm{1}}}} ) $ to the corresponding
  state $\ottnt{C_{{\mathrm{2}}}}  \mathrel{\stackrel{\mathsf{def} }{=} }   \displace{ \ottnt{C_{{\mathrm{1}}}} }{ \ottmv{V_{{\mathrm{1}}}} }{ \ottmv{V_{{\mathrm{2}}}} } $ in $ \mathsf{reachable}( \ottnt{P_{{\mathrm{2}}}} ) $.
  Consider two related $\ottnt{C_{{\mathrm{1}}}}, \ottnt{C_{{\mathrm{2}}}}$ over  $\ottmv{L}$, and $SE$ be
  the valid static environment at $\ottmv{L}$ inferred by our constant
  propagation algorithm. Reducing the next instruction of $\ottnt{C_{{\mathrm{1}}}}$ and
  $\ottnt{C_{{\mathrm{2}}}}$ will produce the same result, given that they only differ by
  substitutions of subexpressions by values that are valid under the static
  environment $SE$, and thus under $\ottnt{E}$.  If
  $ \ottnt{C_{{\mathrm{1}}}}  \nto[  \mathit{A}_\tau  ]  \ottnt{C'_{{\mathrm{1}}}} $ then $ \ottnt{C_{{\mathrm{2}}}}  \nto[  \mathit{A}_\tau  ]  \ottnt{C'_{{\mathrm{2}}}} $, and conversely.
\end{proof}

The restriction of our bisimulation $R$ to reachable configurations
introduced is crucial for the proof to work.  Indeed, a configuration that
is not reachable may \emph{not} respect the static environment $SE$.
Consider the following example, with $ \mathsf{V\kern -0.5pt \scalebox{0.83}{$\mathsf{ 1 }$} } $ on the left and $ \mathsf{V\kern -0.5pt \scalebox{0.83}{$\mathsf{ 2 }$} } $
on the right.
  \begin{mathpar}
     {\small \begin{array}{lll}
                                                      &   \mathsf{L\kern -0.25pt \scalebox{0.83}{$\mathsf{ 1 }$} }   &   \textsfb{var} ~  \mathsf{ x }  \nobreak\hspace{0pt}=\nobreak\hspace{0pt} 1   \\  \,  &      &  \textsfb{print} \,   \mathsf{ x }    \ottsym{+}    \mathsf{ x }    \\  \,  &      &  \textsfb{return} \, 3  \\  
                                                \end{array} } 

     {\small \begin{array}{lll}
                                                      &   \mathsf{L\kern -0.25pt \scalebox{0.83}{$\mathsf{ 1 }$} }   &   \textsfb{var} ~  \mathsf{ x }  \nobreak\hspace{0pt}=\nobreak\hspace{0pt} 1   \\  \,  &      &  \textsfb{print} \, 2  \\  \,  &      &  \textsfb{return} \, 3  \\  
                                                \end{array} } 
  \end{mathpar}
Now consider a pair of configurations at $ \mathsf{L\kern -0.25pt \scalebox{0.83}{$\mathsf{ 1 }$} } $ with the binding $ \mathsf{ x }   \rightarrow  0$ in the environment.
  \begin{mathpar}
    \ottnt{C_{{\mathrm{1}}}}  \mathrel{\stackrel{\mathsf{def} }{=} }   \langle  \ottnt{P} \,  \ottnt{P} ( \ottmv{F} ,  \ottmv{V_{{\mathrm{1}}}} )  \, \ottmv{L_{{\mathrm{1}}}} \, K^* \; \ottnt{M} \,  [   \mathsf{ x }   \rightarrow  0  ]   \rangle 

    \ottnt{C_{{\mathrm{2}}}}  \mathrel{\stackrel{\mathsf{def} }{=} }   \langle  \ottnt{P} \,  \ottnt{P} ( \ottmv{F} ,  \ottmv{V_{{\mathrm{2}}}} )  \, \ottmv{L_{{\mathrm{1}}}} \, K^* \; \ottnt{M} \,  [   \mathsf{ x }   \rightarrow  0  ]   \rangle 
  \end{mathpar}
They would be related by the relation $R$ used by the proof, yet they are not
bisimilar: we have $ \ottnt{C_{{\mathrm{1}}}}  \nto[   \textsf{print}~ 0   ]  \ottnt{C'_{{\mathrm{1}}}} $ as the only transition of
$\ottnt{C_{{\mathrm{1}}}}$ in $\ottmv{V_{{\mathrm{1}}}}$, and $ \ottnt{C_{{\mathrm{2}}}}  \nto[   \textsf{print}~ 2   ]  \ottnt{C'_{{\mathrm{2}}}} $ as the only
transition of $\ottnt{C_{{\mathrm{2}}}}$ in $\ottmv{V_{{\mathrm{2}}}}$.

\subsection{Unreachable Code Elimination}

The following two lemmas are trivial: the simple version-change mapping
between configurations on the two version is clearly a bisimulation.  In the
first case, this comes from the case that $ \textsfb{branch} ~ \textsfb{true} ~ \ottmv{L_{{\mathrm{1}}}} ~ \ottmv{L_{{\mathrm{2}}}} $ and
$\textsfb{goto} \, \ottmv{L_{{\mathrm{1}}}}$ reduce in the example same way.  In the second case,
unreachable configurations are not even considered by the proof.

\begin{lemma}
 Replacing $ \textsfb{branch} ~ \textsfb{true} ~ \ottmv{L_{{\mathrm{1}}}} ~ \ottmv{L_{{\mathrm{2}}}} $ by $\textsfb{goto} \, \ottmv{L_{{\mathrm{1}}}}$ or $ \textsfb{branch} ~ \textsfb{false} ~ \ottmv{L_{{\mathrm{1}}}} ~ \ottmv{L_{{\mathrm{2}}}} $ by $ \textsfb{goto} \, \ottmv{L_{{\mathrm{2}}}}$ results in an equivalent program.
\end{lemma}

\begin{lemma}
  Removing an unreachable label results in an equivalent program.
\end{lemma}

\subsection{Function Inlining}

Assume that the function $\ottmv{F}$ has active version $  \mathsf{V\kern -0.25pt \scalebox{0.87}{$\mathsf{ callee }$} }  $.
If the new version contains a call to $\ottmv{F}$, $ \textsfb{call} ~  \mathsf{ res }  \nobreak\hspace{0pt}=\nobreak\hspace{0pt}  { \ottmv{F} }  ( \ottnt{e_{{\mathrm{1}}}}  \ottsym{,} \, .. \, \ottsym{,}  \ottnt{e_{\ottmv{n}}} ) $ with return label $  \mathsf{L\kern -0.25pt \scalebox{0.87}{$\mathsf{ ret }$} }  $ (the label after the call), 
inlining removes the call and instead:
\begin{itemize}
\item declares a fresh mutable return variable $ \textsfb{var} ~  \mathsf{ res }  \nobreak\hspace{0pt}=\nobreak\hspace{0pt} \textsfb{nil} $;
\item for the formal variables $\mathit{x}, ..$ of $\ottmv{F}$,
  defines the argument variables
  $ \textsfb{var} ~ \mathit{x}_{{\mathrm{1}}} \nobreak\hspace{0pt}=\nobreak\hspace{0pt} \ottnt{se_{{\mathrm{1}}}} , ..,  \textsfb{var} ~ \mathit{x}_{\ottmv{n}} \nobreak\hspace{0pt}=\nobreak\hspace{0pt} \ottnt{se_{\ottmv{n}}} $;
\item inserts the instructions from $  \mathsf{V\kern -0.25pt \scalebox{0.87}{$\mathsf{ callee }$} }  $,
  replacing each instruction $\textsfb{return} \, \ottnt{e}$ by the sequence:
    \(   \mathsf{ res }  \nobreak\hspace{0pt} \leftarrow \nobreak\hspace{0pt} \ottnt{e}  ;~ \textsfb{drop} \, \mathit{x}_{{\mathrm{1}}} ;~...~;~ \textsfb{drop} \, \mathit{x}_{\ottmv{n}} ;~ \textsfb{goto} \,  \mathsf{L\kern -0.25pt \scalebox{0.87}{$\mathsf{ ret }$} }  \)
\end{itemize}

\begin{theorem}
  The inlining transformation presented returns a version equivalent to the
  caller version.
\end{theorem}

\begin{proof}
  \newcommand{\Ecaller}{E_{\mathsf{caller}}}
  \newcommand{\Ecallee}{E_{\mathsf{callee}}} The key idea of the proof is
  that any environment $E$ in the inlined instruction stream can be
  split into two disjoint parts: an environment corresponding to the
  caller function, $\Ecaller$, and an environment corresponding to
  the callee, $\Ecallee$.
  To build the bisimulation, we relate the inlined version, on one hand,
  with the callee on the other hand, \emph{when} the callee was called by
  the called at the inlined call point. This takes two forms:
  \begin{itemize}
  \item If a configuration is currently executing in the callee, and has the
    caller on the top of the call stack with the expected return address, we
    relate it to a configuration in the inlined version (at the same
    position in the callee).  The environment of the inlined version
    is exactly the union of the callee environment (the environment of the
    configuration) and the caller environment (found on the call stack).
  \item If the stack contains a caller frame above a callee frame, we
    relate this to a single frame in the inlined version; again, there is a
    bidirectional correspondence between inlined environment and a pair of a
    caller and callee environment.
  \end{itemize}

To check that this relation is a bisimulation, there are three interesting
cases:
  \begin{itemize}
  \item If a transition is purely within the callee's code on one side, and
    within the inlined version of the callee on the other, it suffices to
    check that the environment decomposition is preserved.  During the
    execution of inlinee, $\Ecaller$ never changes, given that the
    instruction coming from the callee do not have the caller's variable in
    scope---and thus cannot mutate them.
  \item If the transition is a call of the callee from the caller on one side,
    and the entry into the declaration of the return variable $ \textsfb{var} ~  \mathsf{ res }  \nobreak\hspace{0pt}=\nobreak\hspace{0pt} \textsfb{nil} $
    on the other, we step through the silent transitions that bind the call
    parameters $ \textsfb{var} ~ \mathit{x}_{{\mathrm{1}}} \nobreak\hspace{0pt}=\nobreak\hspace{0pt} \ottnt{e_{{\mathrm{1}}}} , ..,  \textsfb{var} ~ \mathit{x}_{\ottmv{n}} \nobreak\hspace{0pt}=\nobreak\hspace{0pt} \ottnt{e_{\ottmv{n}}} $ and get to
    a state in the inlined function corresponding to the start of the callee.
  \item If the transition is a $\textsfb{return} \, \ottnt{e}$ of the callee to the caller
    on one side, and the entry into the result assignment $  \mathsf{ res }  \nobreak\hspace{0pt} \leftarrow \nobreak\hspace{0pt} \ottnt{e} $
    on the other, we similarly step through the $\textsfb{drop} \, \mathit{x}$ for each
    $ \mathit{x} $ in the callee's environment, and get to related state on
    the label $ret$ following the function call.
  \end{itemize}
\end{proof}

\subsection{Unrestricted Deoptimization}

Consider $\ottnt{P_{{\mathrm{1}}}}$ containing an \assume at $\ottmv{L_{{\mathrm{1}}}}$, followed by $\ottnt{i_{\ottmv{m}}}$ at
$\ottmv{L_{{\mathrm{2}}}}  \mathrel{\stackrel{\mathsf{def} }{=} }   ( \ottmv{L_{{\mathrm{1}}}} \kern-1pt+\kern-2pt1) $.  Let $\ottnt{i_{\ottmv{m}}}$ be such it has a unique
successor, is the unique predecessor of $\ottmv{L_{{\mathrm{1}}}}$, and is not a function
call, has no side-effect, does not modify the heap (array write or
creation), and does not modify the variables mentioned in the \assume.
Under these conditions, we can move the \assume immediately after the
successor of $\ottnt{i_{\ottmv{m}}}$.  Let us name $\ottnt{P_{{\mathrm{2}}}}$ the program modified in this
way.

\begin{lemma}
  Given a program $\ottnt{P_{{\mathrm{1}}}}$, and $\ottnt{P_{{\mathrm{2}}}}$ obtained by
  permuting an \assume instruction $\ottmv{L_{{\mathrm{1}}}}$ after $\ottnt{i_{\ottmv{m}}}$ at $\ottmv{L_{{\mathrm{2}}}}$ under the conditions
  above, $\ottnt{P_{{\mathrm{1}}}}$ and $\ottnt{P_{{\mathrm{2}}}}$ are bisimilar.
\end{lemma}

\begin{proof}
  The applicability restrictions are specific enough that we can reason
  precisely about the structure of reductions around the permuted
  instructions.
  Consider a configuration $\ottnt{C_{{\mathrm{1}}}}$ over the \assume at $\ottmv{L_{{\mathrm{1}}}}$ in
  $\ottnt{P_{{\mathrm{1}}}}$, and the corresponding configuration $\ottnt{C_{{\mathrm{2}}}}  \mathrel{\stackrel{\mathsf{def} }{=} } 
   \displace{  \displace{ \ottnt{C_{{\mathrm{1}}}} }{ \ottnt{P_{{\mathrm{1}}}} }{ \ottnt{P_{{\mathrm{2}}}} }  }{ \ottmv{L_{{\mathrm{1}}}} }{ \ottmv{L_{{\mathrm{2}}}} } $ over $\ottmv{L_{{\mathrm{2}}}}$ in $\ottnt{P_{{\mathrm{2}}}}$.
  Instruction $\ottnt{i_{\ottmv{m}}}$ has a single successor, so there is only one possible
  reduction rule.  Since $\ottnt{i_{\ottmv{m}}}$ is not an I/O instruction, it must be a
  silent action.  Hence there is a unique $\ottnt{C'_{{\mathrm{2}}}}$ such that $ \ottnt{C_{{\mathrm{2}}}}  \nto[  \mathit{A}_\tau  ]  \ottnt{C'_{{\mathrm{2}}}} $ holds, and furthermore $\mathit{A}_\tau$ is $ \tau $.
  Configurations $\ottnt{C_{{\mathrm{1}}}}$ and $\ottnt{C'_{{\mathrm{2}}}}$ are over the same \assume.  Let
  $\ottnt{E_{{\mathrm{1}}}}$ and $\ottnt{E_{{\mathrm{2}}}}$ be environments of $\ottnt{C_{{\mathrm{1}}}}$ and $\ottnt{C'_{{\mathrm{2}}}}$
  respectively, and $\ottnt{E'}$ be their common sub-environment that contain
  only the variables mentioned in the \assume ($\ottnt{i_{\ottmv{m}}}$ does not modify its variables).
  If all tests in the \assume instruction are true under $\ottnt{E'}$, then
  $\ottnt{C_{{\mathrm{1}}}}$ and $\ottnt{C'_{{\mathrm{2}}}}$ silently reduce to $\ottnt{C'_{{\mathrm{1}}}}$ and
  $\ottnt{C''_{{\mathrm{2}}}}$. $\ottnt{C'_{{\mathrm{1}}}}$ is over $\ottnt{i_{\ottmv{m}}}$ at $\ottmv{L_{{\mathrm{2}}}}$, so it reduces $ \ottnt{C'_{{\mathrm{1}}}}  \nto[   \tau   ]  \ottnt{C''_{{\mathrm{1}}}} $; notice that $\ottnt{C''_{{\mathrm{1}}}}$ and $\ottnt{C''_{{\mathrm{2}}}}$ are over the
  labels $ ( \ottmv{L_{{\mathrm{2}}}} \kern-1pt+\kern-2pt1) $ in $\ottnt{P_{{\mathrm{1}}}}$ and $ ( \ottmv{L_{{\mathrm{1}}}} \kern-1pt+\kern-2pt1) $ in $\ottnt{P_{{\mathrm{2}}}}$,
  which are equal.
  If not all tests of the \assume are true under $\ottnt{E'}$, then both $\ottnt{C_{{\mathrm{1}}}}$ and $\ottnt{C'_{{\mathrm{2}}}}$ deoptimize.
  The deoptimized configurations are the same
  \begin{itemize}
  \item their function, version and label are the same: the \assume's deoptimization target;
  \item they have the same call stack: it only depends on the call stack of $\ottnt{C_{{\mathrm{1}}}}$ and the
    interpretation of the \assume's extra frames under $\ottnt{E'}$;
  \item they have the same heap, as we assumed that $\ottnt{i_{\ottmv{m}}}$ does not modify
    the heap;
  \item they have the same deoptimized environment: it only depends on $\ottnt{E'}$.
  \end{itemize}
  Let us call $\ottnt{C_{{\mathrm{0}}}}$ the configuration resulting from either deoptimization transitions.

  We establish bisimilarity using definition a relation $R$ and proving it is a bisimulation.
  The following diagrams are useful to follow the definition of $R$ and the proofs.
  
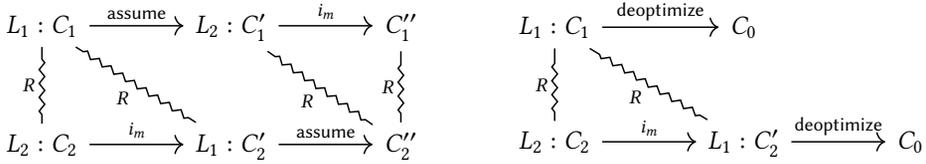
\begin{mathpar}
  \begin{tikzcd}[row sep = large, column sep = large]
    \ottmv{L_{{\mathrm{1}}}} : \ottnt{C_{{\mathrm{1}}}}
    \arrow[to, r, "\assume"]
    &
    \ottmv{L_{{\mathrm{2}}}} : \ottnt{C'_{{\mathrm{1}}}}
    \arrow[to, r, "\ottnt{i_{\ottmv{m}}}"]
    &
    \ottnt{C''_{{\mathrm{1}}}}
    \\
    \ottmv{L_{{\mathrm{2}}}} : \ottnt{C_{{\mathrm{2}}}}
    \arrow[no head, squiggly, u, "R"]
    \arrow[to, r, "\ottnt{i_{\ottmv{m}}}"]
    &
    \ottmv{L_{{\mathrm{1}}}} : \ottnt{C'_{{\mathrm{2}}}}
    \arrow[to, r, "\assume"]
    \arrow[no head, squiggly, ul, "R"]
    &
    \ottnt{C''_{{\mathrm{2}}}}
    \arrow[no head, squiggly, ul, "R"]
    \arrow[no head, squiggly, u, "R"]
  \end{tikzcd}

  \begin{tikzcd}[row sep = large, column sep = large]
    \ottmv{L_{{\mathrm{1}}}} : \ottnt{C_{{\mathrm{1}}}}
    \arrow[to, r, "\mathsf{deoptimize}"]
    &
    \ottnt{C_{{\mathrm{0}}}}
    &
    \\
    \ottmv{L_{{\mathrm{2}}}} : \ottnt{C_{{\mathrm{2}}}}
    \arrow[no head, squiggly, u, "R"]
    \arrow[to, r, "\ottnt{i_{\ottmv{m}}}"]
    &
    \ottmv{L_{{\mathrm{1}}}} : \ottnt{C'_{{\mathrm{2}}}}
    \arrow[no head, squiggly, ul, "R"]
    \arrow[to, r, "\mathsf{deoptimize}"]
    &
    \ottnt{C_{{\mathrm{0}}}}
  \end{tikzcd}
  \end{mathpar}

  We define $R$ as the smallest relation such that:
  \begin{enumerate}
    \item For any $\ottnt{C_{{\mathrm{1}}}}$ and $\ottnt{C_{{\mathrm{2}}}}$ as above,
      $\ottnt{C_{{\mathrm{1}}}}$ and $\ottnt{C'_{{\mathrm{1}}}}$ are related to $\ottnt{C_{{\mathrm{2}}}}$.
    \item For any $\ottnt{C_{{\mathrm{1}}}}$ and $\ottnt{C_{{\mathrm{2}}}}$ as above
      such that $\ottnt{C_{{\mathrm{1}}}}$ passes the \assume tests (does not deoptimize),
      both $\ottnt{C'_{{\mathrm{2}}}}$ and $\ottnt{C''_{{\mathrm{2}}}}$ are related to $\ottnt{C''_{{\mathrm{1}}}}$.
    \item For any $\ottnt{C}$ over $\ottnt{P_{{\mathrm{1}}}}$ that is over neither $\ottmv{L_{{\mathrm{1}}}}$ nor $\ottmv{L_{{\mathrm{2}}}}$,
      $\ottnt{C}$ and $ \displace{ \ottnt{C} }{ \ottnt{P_{{\mathrm{1}}}} }{ \ottnt{P_{{\mathrm{2}}}} } $ are related.
  \end{enumerate}
  We now prove that $R$ is a bisimulation.
  Any pair of configurations that are not over either $\ottmv{L_{{\mathrm{1}}}}$ or
  $\ottmv{L_{{\mathrm{2}}}}$ come from the case (3), so they are identical and it is
  immediate that they match each other. The interesting cases are for
  matching pairs of configurations over $\ottmv{L_{{\mathrm{1}}}}$ or $\ottmv{L_{{\mathrm{2}}}}$.

  In the case where no deoptimization happens, the reductions in $\ottnt{P_{{\mathrm{2}}}}$ are
  either $ \ottnt{C_{{\mathrm{2}}}}  \nto[   \tau   ]  \ottnt{C'_{{\mathrm{2}}}} $, where both configurations are related to $\ottnt{C_{{\mathrm{1}}}}$,
  or $ \ottnt{C'_{{\mathrm{2}}}}  \nto[   \tau   ]  \ottnt{C''_{{\mathrm{2}}}} $ which is matched by $ \ottnt{C_{{\mathrm{1}}}}  \nto[   \tau   ]  \ottnt{C'_{{\mathrm{1}}}} $.
  The reductions in $\ottnt{P_{{\mathrm{1}}}}$ are
  either $ \ottnt{C_{{\mathrm{1}}}}  \nto[   \tau   ]  \ottnt{C'_{{\mathrm{1}}}} $, which is matched by $ \ottnt{C_{{\mathrm{2}}}}  \nto[   \tau   ]   \ottnt{C'_{{\mathrm{2}}}}  \nto[   \tau   ]  \ottnt{C''_{{\mathrm{2}}}}  $ and $ \ottnt{C'_{{\mathrm{2}}}}  \nto[   \tau   ]  \ottnt{C''_{{\mathrm{2}}}} $,
  or $ \ottnt{C'_{{\mathrm{1}}}}  \nto[   \tau   ]  \ottnt{C''_{{\mathrm{1}}}} $, which are both related to $\ottnt{C''_{{\mathrm{2}}}}$.

  In the case where a deoptimization happens,
  the only reduction in $\ottnt{P_{{\mathrm{1}}}}$ is $ \ottnt{C_{{\mathrm{1}}}}  \nto[   \tau   ]  \ottnt{C_{{\mathrm{0}}}} $,
  which is matched by $ \ottnt{C_{{\mathrm{2}}}}  \nto[   \tau   ]   \ottnt{C'_{{\mathrm{2}}}}  \nto[   \tau   ]  \ottnt{C_{{\mathrm{0}}}}  $ and $ \ottnt{C'_{{\mathrm{2}}}}  \nto[   \tau   ]  \ottnt{C_{{\mathrm{0}}}} $.
  The reductions in $\ottnt{P_{{\mathrm{2}}}}$ are $ \ottnt{C_{{\mathrm{2}}}}  \nto[   \tau   ]  \ottnt{C'_{{\mathrm{2}}}} $,
  which are matched by the empty reduction on $\ottnt{C_{{\mathrm{1}}}}$
  and $ \ottnt{C'_{{\mathrm{2}}}}  \nto[   \tau   ]  \ottnt{C_{{\mathrm{0}}}} $ are matched by $ \ottnt{C_{{\mathrm{1}}}}  \nto[   \tau   ]  \ottnt{C_{{\mathrm{0}}}} $.

  Finally, we show preservation of the assumption transparency invariant.
  We have to establish the invariant for $\ottnt{P_{{\mathrm{2}}}}$, assuming the invariant
  for $\ottnt{P_{{\mathrm{2}}}}$.  We have to show that $\ottnt{C_{{\mathrm{0}}}}$ and $\ottnt{C'_{{\mathrm{2}}}}$ are bisimilar.
  $\ottnt{C_{{\mathrm{0}}}}$ is bisimilar to $\ottnt{C_{{\mathrm{1}}}}$ (this is the transparency invariant on
  $\ottnt{P_{{\mathrm{1}}}}$), and $\ottnt{C_{{\mathrm{1}}}}$ and $\ottnt{C'_{{\mathrm{2}}}}$ are bisimilar because they are
  related by the bisimulation $R$.
\end{proof}

\subsection{Predicate Hoisting}

Hoisting predicates takes a version $\ottmv{V_{{\mathrm{1}}}}$, an expression $\ottnt{e}$, and two
labels $\ottmv{L_{{\mathrm{1}}}}, \ottmv{L_{{\mathrm{2}}}}$, such that the instruction at $\ottmv{L_{{\mathrm{1}}}}, \ottmv{L_{{\mathrm{2}}}}$ are
both \assume instructions and $\ottnt{e}$ is a part of the predicate list at $\ottmv{L_{{\mathrm{1}}}}$.
The pass
copies $\ottnt{e}$ from $\ottmv{L_{{\mathrm{1}}}}$ to $\ottmv{L_{{\mathrm{2}}}}$, if all variables mentioned in
$\ottnt{e}$ are in scope at $\ottmv{L_{{\mathrm{2}}}}$.  If, after this step the $\ottnt{e}$ can be
constant folded to $ \textsfb{true} $ at $\ottmv{L_{{\mathrm{1}}}}$ by the optimization from
\autoref{sec:specl-constantprop}, then it is removed from $\ottmv{L_{{\mathrm{1}}}}$,
otherwise the whole version stays unchanged.

\begin{lemma}
  Let $\ottmv{V_{{\mathrm{2}}}}$ be the result of hoisting  $\ottnt{e}$ from $\ottmv{L_{{\mathrm{1}}}}$
  to $\ottmv{L_{{\mathrm{2}}}}$ in $\ottmv{V_{{\mathrm{1}}}}$. $\ottmv{V_{{\mathrm{1}}}}$ and $\ottmv{V_{{\mathrm{2}}}}$ are bisimilar.
\end{lemma}

\begin{proof}
Copying is bisimilar due to the assumption transparency invariant and to the fact that
the constant-folded version is bisimilar due to Lemma \ref{lem:constfold}.
\end{proof}

\subsection{Assume Composition}

Let $\ottmv{V_{{\mathrm{1}}}},\ottmv{V_{{\mathrm{2}}}},\ottmv{V_{{\mathrm{3}}}}$ be three versions of a function $\ottmv{F}$ with
instruction streams $\ottnt{I_{{\mathrm{1}}}},\ottnt{I_{{\mathrm{2}}}},\ottnt{I_{{\mathrm{3}}}}$, and $\ottmv{L_{{\mathrm{1}}}},\ottmv{L_{{\mathrm{2}}}},\ottmv{L_{{\mathrm{3}}}}$
labels, such that $ \ottnt{I_{{\mathrm{1}}}} ( \ottmv{L_{{\mathrm{1}}}} )  =  \textsfb{assume} ~ \ottnt{e_{{\mathrm{1}}}} ~ \ottkw{else} ~   \ottmv{F} .\kern -0.5pt \ottmv{V_{{\mathrm{2}}}} .\kern -0.5pt \ottmv{L_{{\mathrm{2}}}}  ~ \mathit{VA}_{{\mathrm{1}}}  $ and $ \ottnt{I_{{\mathrm{2}}}} ( \ottmv{L_{{\mathrm{2}}}} )  =  \textsfb{assume} ~ \ottnt{e_{{\mathrm{2}}}} ~ \ottkw{else} ~   \ottmv{F} .\kern -0.5pt \ottmv{V_{{\mathrm{3}}}} .\kern -0.5pt \ottmv{L_{{\mathrm{3}}}}  ~ \mathit{VA}_{{\mathrm{2}}}  $.
The composition pass creates a new program $\ottnt{P_{{\mathrm{2}}}}$ from $\ottnt{P_{{\mathrm{1}}}}$
identical but the \assume $ \ottnt{P_{{\mathrm{2}}}} ( \ottmv{F} . \ottmv{V_{{\mathrm{1}}}} . \ottmv{L_{{\mathrm{1}}}} ) $
is replaced by $ \textsfb{assume} ~ \ottnt{e_{{\mathrm{1}}}}  \ottsym{,}  \ottnt{e_{{\mathrm{2}}}} ~ \ottkw{else} ~   \ottmv{F} .\kern -0.5pt \ottmv{V_{{\mathrm{3}}}} .\kern -0.5pt \ottmv{L_{{\mathrm{3}}}}  ~  \mathit{VA}_{{\mathrm{2}}}  \circ  \mathit{VA}_{{\mathrm{1}}}   $ where
$(  [{ \mathit{x}_{{\mathrm{1}}}  \ottsym{=}  \ottnt{e_{{\mathrm{1}}}}  \ottsym{,} \, .. \, \ottsym{,}  \mathit{x}_{\ottmv{n}}  \ottsym{=}  \ottnt{e_{\ottmv{n}}}  \kern 0.033em}]   \circ  \mathit{VA} )$ is defined as
$[x_1 = e_1\{\frac{\mathit{VA}(y)}{y} \forall y \in
  \mathit{VA}\}, .., x_n = e_n\{\frac{\mathit{VA}(y)}{y} \forall
  y \in \mathit{VA}\}]$.

\begin{lemma}
  Let $\ottnt{P_{{\mathrm{2}}}}$ be the result of composing \assume instructions at $\ottmv{L_{{\mathrm{1}}}}$
  and $\ottmv{L_{{\mathrm{2}}}}$. $\ottnt{P_{{\mathrm{1}}}}$ and $\ottnt{P_{{\mathrm{2}}}}$ are bisimilar.
\end{lemma}

\begin{proof}
  For $ \ottnt{C_{{\mathrm{1}}}}  \nto[   \tau   ]  \ottnt{C'_{{\mathrm{1}}}} $, $ \ottnt{C_{{\mathrm{2}}}}  \nto[   \tau   ]  \ottnt{C'_{{\mathrm{2}}}} $ over $\ottmv{L_{{\mathrm{1}}}}$ in $\ottnt{P_{{\mathrm{1}}}}$, $\ottnt{P_{{\mathrm{2}}}}$, we distinguish four cases:
  \begin{enumerate}
    \item If $\ottnt{e_{{\mathrm{1}}}}$ and $\ottnt{e_{{\mathrm{2}}}}$ both hold, the \assume does not
      deoptimize in $\ottnt{P_{{\mathrm{1}}}}$ and $\ottnt{P_{{\mathrm{2}}}}$ and they behave identically.
    \item If $\ottnt{e_{{\mathrm{1}}}}$ and $\ottnt{e_{{\mathrm{2}}}}$ both fail, the original program
      deoptimizes twice; the modified $\ottnt{P_{{\mathrm{2}}}}$ only once. Assuming
      deoptimizing under the combined varmap $\ottnt{M} \, \ottnt{E} \,  \mathit{VA}_{{\mathrm{2}}}  \circ  \mathit{VA}_{{\mathrm{1}}}   \rightsquigarrow  \ottnt{E''}$
      produces an environment equivalent to $\ottnt{M} \, \ottnt{E} \, \mathit{VA}_{{\mathrm{1}}}  \rightsquigarrow  \ottnt{E'}$ and $\ottnt{M} \, \ottnt{E'} \, \mathit{VA}_{{\mathrm{2}}}  \rightsquigarrow  \ottnt{E''}$ the final
      configuration is identical.  Since the extra intermediate step is
      silent, both programs are bisimilar.
    \item If $\ottnt{e_{{\mathrm{1}}}}$ fails and $\ottnt{e_{{\mathrm{2}}}}$ holds, we deoptimize
      to $\ottmv{V_{{\mathrm{3}}}}$ in $\ottnt{P_{{\mathrm{2}}}}$, but to $\ottmv{V_{{\mathrm{2}}}}$ in $\ottnt{P_{{\mathrm{1}}}}$.
      As shown in case (2) the deoptimized configuration $\ottnt{C'_{{\mathrm{2}}}}$ over $\ottmv{L_{{\mathrm{3}}}}$ is equivalent to a post-deoptimization configuration of $\ottnt{C'_{{\mathrm{1}}}}$, which, due to assumption transparency is bisimilar to $\ottnt{C'_{{\mathrm{1}}}}$ itself.
    \item If $\ottnt{e_{{\mathrm{1}}}}$ holds and $\ottnt{e_{{\mathrm{2}}}}$ fails, deoptimize
      to $\ottmv{V_{{\mathrm{3}}}}$ in $\ottnt{P_{{\mathrm{2}}}}$ but not in $\ottnt{P_{{\mathrm{1}}}}$.
      Again $\ottnt{C'_{{\mathrm{2}}}}$ is equivalent to a post-deoptimization state, which is, transitively, bisimilar to $\ottnt{C'_{{\mathrm{1}}}}$. 
  \end{enumerate}
Since a well-formed \assume has only unique names in the
deoptimization metadata, it is simple to show the assumption in (2) with a
substitution lemma.
\end{proof}

\section{Discussion}\label{sec:discussion}
Our formalization raises new questions and makes apparent certain design
choices.  In this section, we present insights into the design space for JIT
implementations.

\begin{wrapfigure}{r}{6.8cm}
  \vskip 1mm
$  {\small \begin{array}{lll}
                                                      &   \mathsf{L\kern -0.25pt \scalebox{0.87}{$\mathsf{ loop }$} }   &   \textsfb{branch} ~   \mathsf{ z }    \neq   0  ~  \mathsf{L\kern -0.25pt \scalebox{0.87}{$\mathsf{ body }$} }  ~  \mathsf{L\kern -0.25pt \scalebox{0.87}{$\mathsf{ done }$} }    \\  \,  &   \mathsf{L\kern -0.25pt \scalebox{0.87}{$\mathsf{ body }$} }   &   \textsfb{call} ~  \mathsf{ x }  \nobreak\hspace{0pt}=\nobreak\hspace{0pt}  {  \scalebox{0.97}{$\mathsf{ dostuff }$}  }  ( \, )   \\  \,  &      &   \textsfb{var} ~  \mathsf{ y }  \nobreak\hspace{0pt}=\nobreak\hspace{0pt}   \mathsf{ x }    \ottsym{+}   13    \\  \,  &      &  \textsfb{assume} ~ \ottnt{e} ~ \ottkw{else} ~   \ottmv{F} .\kern -0.5pt \ottmv{V} .\kern -0.5pt \ottmv{L}  ~  [{  \mathsf{ x }   \ottsym{=}   \mathsf{ x }   \ottsym{,}   \mathsf{ y }   \ottsym{=}    \mathsf{ x }    \ottsym{+}   13   \kern 0.033em}]    \\  \,  &      &  \textsfb{drop} \,  \mathsf{ y }   \\  \,  &      &  \textsfb{goto} \,  \mathsf{L\kern -0.25pt \scalebox{0.87}{$\mathsf{ loop }$} }   \\  \,  &   \mathsf{L\kern -0.25pt \scalebox{0.87}{$\mathsf{ done }$} }   & \dots \\  
                                                \end{array} }  $
\caption{Deoptimization keeps variables alive.}\label{fig:hoist-ex} 
\vspace{-4mm}\end{wrapfigure}

\paragraph{The Cost of Assuming}
Assumptions restrict optimizations. Variables needed for deoptimization must
be kept alive.  Consider \autoref{fig:hoist-ex}, where an \assume is at the
end of a loop.  As $  \mathsf{ y }  $ is not modified, it can be removed.  There
is enough information to reconstruct it if needed.  On the other hand,
$  \mathsf{ x }  $ cannot be synthesized out of thin air because it is computed
in another function.  Additionally, \assume restricts code motion in two cases.  First,
side-effecting code cannot be moved over an \assume.  Second, \assume
instructions cannot be hoisted over instructions that interfere with variables
mentioned in metadata.  It is possible to move \assume forward, since data
dependencies can be resolved by taking a snapshot of the environment at the
original location.  For the reverse effect, we support hoisting the
predicate from one \assume to another (see \autoref{sec:hoist-assumpt}).
Moving \assume instructions up is tricky and also unnecessary, since in
combination those two primitives allow moving checks to any position.  In
the above example, if $e$ is invariant in the loop body and there is an
\assume before $  \mathsf{L\kern -0.25pt \scalebox{0.87}{$\mathsf{ loop }$} }  $, the predicate can be hoisted out of the
loop.  If the \assume is only relevant for a subset of the instructions
after the current location, it can be moved down as a whole.

\paragraph{Lazy Deoptimization}
The runtime cost of an \assume is the cost of monitoring the predicates.
Suppose we speculate that the contents of an array remain unchanged
throughout a loop.  An implementation would have to check every single
element of the array.  An eager strategy where predicates  are checked at every
iteration is wasteful.  It is more efficient to associate checks to
operations that may invalidate the predicates, such as array writes, to
invalidate the assumption, a strategy sometimes known as \emph{lazy
  deoptimization}.  We could implement dependencies by separating
\begin{wrapfigure}[31]{r}{6.5cm}
$  {\small
                                                          \begin{array}{l}
                                                              \scalebox{0.97}{$\mathsf{ stuck }$}   (  \,  \kern 0.04em ) \\
                                                                   \begin{array}{llll}
                                                                     \kern 1.5pt   \hspace{2.5mm}  \mathsf{V\kern -0.25pt \scalebox{0.87}{$\mathsf{ base }$} }   \\
                                                                             \begin{array}{l!{\,\color{gray}\vrule}lll}
                                                                               &      &   \textsfb{call} ~  \mathsf{ debug }  \nobreak\hspace{0pt}=\nobreak\hspace{0pt}  {  \scalebox{0.97}{$\mathsf{ debug }$}  }  ( \, )   \\  \,  &   \mathsf{L\kern -0.25pt \scalebox{0.87}{$\mathsf{ h }$} }   &   \textsfb{branch} ~   \mathsf{ x }    <   1000000  ~  \mathsf{L\kern -0.25pt \scalebox{0.87}{$\mathsf{ o }$} }  ~  \mathsf{L\kern -0.25pt \scalebox{0.87}{$\mathsf{ rt }$} }    \\  \,  &   \mathsf{L\kern -0.25pt \scalebox{0.87}{$\mathsf{ o }$} }   &   \textsfb{branch} ~  \mathsf{ debug }  ~  \mathsf{L\kern -0.25pt \scalebox{0.87}{$\mathsf{ slow }$} }  ~  \mathsf{L\kern -0.25pt \scalebox{0.87}{$\mathsf{ fast }$} }    \\  \,  &   \mathsf{L\kern -0.25pt \scalebox{0.87}{$\mathsf{ slow }$} }   & \dots \\  \,  &   \mathsf{L\kern -0.25pt \scalebox{0.87}{$\mathsf{ fast }$} }   & \dots \\  \,  &      &  \textsfb{goto} \,  \mathsf{L\kern -0.25pt \scalebox{0.87}{$\mathsf{ h }$} }   \\  \,  &   \mathsf{L\kern -0.25pt \scalebox{0.87}{$\mathsf{ rt }$} }   & \dots \\  
                                                                             \end{array} \vspace{0.25em} \\  
                                                                   \end{array} \\  
                                                          \end{array} }  $

  \caption{Long running execution.}\label{fig:osr-in-ex}
  \vskip 8mm
$  {\small
                                                          \begin{array}{l}
                                                              \scalebox{0.97}{$\mathsf{ cont }$}   (   \mathsf{ x }   \kern 0.04em ) \\
                                                                   \begin{array}{llll}
                                                                     \kern 1.5pt   \hspace{2.5mm}  \mathsf{V\kern -0.25pt \scalebox{0.87}{$\mathsf{ opt }$} }   \\
                                                                             \begin{array}{l!{\,\color{gray}\vrule}lll}
                                                                               &   \mathsf{L\kern -0.25pt \scalebox{0.87}{$\mathsf{ h }$} }   &   \textsfb{branch} ~   \mathsf{ x }    <   1000000  ~  \mathsf{L\kern -0.25pt \scalebox{0.87}{$\mathsf{ fast }$} }  ~  \mathsf{L\kern -0.25pt \scalebox{0.87}{$\mathsf{ rt }$} }    \\  \,  &   \mathsf{L\kern -0.25pt \scalebox{0.87}{$\mathsf{ fast }$} }   & \dots \\  \,  &      &  \textsfb{goto} \,  \mathsf{L\kern -0.25pt \scalebox{0.87}{$\mathsf{ h }$} }   \\  \,  &   \mathsf{L\kern -0.25pt \scalebox{0.87}{$\mathsf{ rt }$} }   & \dots \\  
                                                                             \end{array} \vspace{0.25em} \\  
                                                                   \end{array} \\  
                                                          \end{array} }  $

  \caption{Switching to optimized code.}\label{fig:osr-in-ex2}

\vskip 8mm
    $  {\small
                                                          \begin{array}{l}
                                                              \scalebox{0.97}{$\mathsf{ undo }$}   (  \,  \kern 0.04em ) \\
                                                                   \begin{array}{llll}
                                                                     \kern 1.5pt   \hspace{2.5mm}  \mathsf{V\kern -0.25pt \scalebox{0.87}{$\mathsf{ s123 }$} }   \\
                                                                             \begin{array}{l!{\,\color{gray}\vrule}lll}
                                                                               &   \mathsf{L\kern -0.25pt \scalebox{0.83}{$\mathsf{ 0 }$} }   &  \textsfb{assume} ~ \ottnt{e_{{\mathrm{1}}}}  \ottsym{,}  \ottnt{e_{{\mathrm{2}}}}  \ottsym{,}  \ottnt{e_{{\mathrm{3}}}} ~ \ottkw{else} ~    \scalebox{0.97}{$\mathsf{ undo }$}  .\kern -0.5pt  \mathsf{V\kern -0.25pt \scalebox{0.87}{$\mathsf{ s12 }$} }  .\kern -0.5pt  \mathsf{L\kern -0.25pt \scalebox{0.83}{$\mathsf{ 0 }$} }   ~  [{  \dots   \kern 0.033em}]    \\  
                                                                             \end{array} \vspace{0.25em} \\  \,  \hspace{2.5mm}  \mathsf{V\kern -0.25pt \scalebox{0.87}{$\mathsf{ s12 }$} }   \\
                                                                             \begin{array}{l!{\,\color{gray}\vrule}lll}
                                                                               &   \mathsf{L\kern -0.25pt \scalebox{0.83}{$\mathsf{ 0 }$} }   &  \textsfb{assume} ~ \ottnt{e_{{\mathrm{1}}}}  \ottsym{,}  \ottnt{e_{{\mathrm{2}}}} ~ \ottkw{else} ~    \scalebox{0.97}{$\mathsf{ undo }$}  .\kern -0.5pt  \mathsf{V\kern -0.25pt \scalebox{0.87}{$\mathsf{ s1 }$} }  .\kern -0.5pt  \mathsf{L\kern -0.25pt \scalebox{0.83}{$\mathsf{ 0 }$} }   ~  [{  \dots   \kern 0.033em}]    \\  
                                                                             \end{array} \vspace{0.25em} \\  \,  \hspace{2.5mm}  \mathsf{V\kern -0.25pt \scalebox{0.87}{$\mathsf{ s1 }$} }   \\
                                                                             \begin{array}{l!{\,\color{gray}\vrule}lll}
                                                                               &   \mathsf{L\kern -0.25pt \scalebox{0.83}{$\mathsf{ 0 }$} }   &  \textsfb{assume} ~ \ottnt{e_{{\mathrm{1}}}} ~ \ottkw{else} ~    \scalebox{0.97}{$\mathsf{ undo }$}  .\kern -0.5pt  \mathsf{V\kern -0.25pt \scalebox{0.87}{$\mathsf{ base }$} }  .\kern -0.5pt  \mathsf{L\kern -0.25pt \scalebox{0.83}{$\mathsf{ 0 }$} }   ~  [{  \dots   \kern 0.033em}]    \\  
                                                                             \end{array} \vspace{0.25em} \\  
                                                                   \end{array} \\  
                                                          \end{array} } $
\caption{Undoing an isolated predicate.}\label{fig:fine-grained-deopt}
\end{wrapfigure} 
assumptions from runtime checks.  Specifically, let $  \mathsf{ GUARDS }  [{ 13 }]  =
 \textsfb{true} $ be the runtime check, where the global array $  \mathsf{ GUARDS }  $ is
a collection of all remote assumptions that can be invalidated by an
operation, such as an array assignment.  In terms of correctness, both eager
and lazy deoptimization are similar; however, we would need to prove
correctness of the dependency mechanism that modifies the global array.

\paragraph{Jumping Into Optimized Code}
We have shown how to transfer control out of optimized code.  The inverse
transition, jumping into optimized code, is interesting as well.  Consider
executing the long running loop of \autoref{fig:osr-in-ex}.  The value of
$  \mathsf{ debug }  $ is constant in the loop, yet execution is stuck in the
long running function and must branch on each iteration.  A JIT can compile
an optimized version that speculates on $  \mathsf{ debug }  $, but it may only
use it on the next invocation.  Ideally, the JIT would jump into the newly
optimized code from the slow loop; this is known as \emph{hot loop
  transfer}.  Specifically, the next time $  \mathsf{L\kern -0.25pt \scalebox{0.87}{$\mathsf{ o }$} }  $ is reached,
control is transferred to an equivalent location in the optimized version.
To do so, continuation-passing style can be used to compile a staged
continuation function from the beginning of the loop where $  \mathsf{ debug }  $
is known to be false.  The optimized continuation might look like
$  \scalebox{0.97}{$\mathsf{ cont }$}  $ in \autoref{fig:osr-in-ex2}.  In some sense, this is
easier than deoptimization because it strengthens assumptions rather than
weakening them and all the values needed to construct the state at the
target version are readily available.  

\paragraph{Fine-Grained Deoptimization}
Instead of blindly removing all assumptions on deoptimization, it is
possible to undo only failing assumptions while preserving the rest.  As
shown in \autoref{fig:fine-grained-deopt}, if $\ottnt{e_{{\mathrm{2}}}}$ fails in
version~$  \mathsf{V\kern -0.25pt \scalebox{0.87}{$\mathsf{ s123 }$} }  $, one can jump to the last version that did not
rely on this predicate.
By deoptimizing to
version~$  \mathsf{V\kern -0.25pt \scalebox{0.87}{$\mathsf{ s1 }$} }  $, assumption~$\ottnt{e_{{\mathrm{3}}}}$ must be discarded.
However,
$\ottnt{e_{{\mathrm{1}}}}  \ottsym{,}  \ottnt{e_{{\mathrm{3}}}}$ still hold, so we would like to preserve optimizations based on
those assumptions.  Using the technique mentioned above, 
execution can be transferred to a version~$  \mathsf{L\kern -0.25pt \scalebox{0.87}{$\mathsf{ s13 }$} }  $ that reintroduces $\ottnt{e_{{\mathrm{3}}}}$.  The
overall effect is that we remove only the invalidated assumption and its
optimizations.  We are not aware of an existing implementation that explores
such a strategy.

\begin{wrapfigure}{r}{6cm}
\vskip -5mm
$ {\small \begin{array}{lll}
                                                      &      & \dots \\  \,  &   \mathsf{L\kern -0.25pt \scalebox{0.87}{$\mathsf{ loop }$} }   &   \textsfb{branch} ~ \ottnt{e} ~  \mathsf{L\kern -0.25pt \scalebox{0.87}{$\mathsf{ body }$} }  ~  \mathsf{L\kern -0.25pt \scalebox{0.87}{$\mathsf{ done }$} }    \\  \,  &   \mathsf{L\kern -0.25pt \scalebox{0.87}{$\mathsf{ body }$} }   &    \mathsf{ x }  \nobreak\hspace{0pt} \leftarrow \nobreak\hspace{0pt} 0   \\  \,  &      & \dots \\  \,  &      &  \textsfb{goto} \,  \mathsf{L\kern -0.25pt \scalebox{0.87}{$\mathsf{ loop }$} }   \\  \,  &   \mathsf{L\kern -0.25pt \scalebox{0.87}{$\mathsf{ done }$} }   & \dots \\  
                                                \end{array} } $
\vskip -2mm
\caption{Loop with a dead store.}\label{fig:trace-jit-ex} 
\end{wrapfigure} 

\paragraph{Simulating a Tracing JIT}\label{sec:tracing}
A tracing JIT~\citep{dynamo00,trace-franz09} records instructions that are
executed in a trace. Branches and redundant checks can be discarded from the
trace. Typically, a trace corresponds to a path through a hot loop.  On
subsequent runs the trace is executed directly. The JIT ensures that
execution follows the same path, otherwise it deoptimizes back to the
original program.  In this context \citet{guo11} develop a framework for
reasoning about optimizations applied to traces.  One of their
results is that dead store elimination is unsound, because the trace is only
a partial view of the entire program.  For example, a variable $  \mathsf{ x }  $
might be assigned to within a trace, but never used.  However, it is unsound
to remove the assignment, because $x$ might be used outside the trace.  We
can simulate their tracing formalism in \sourir.  Consider a variant of
their running example shown in \autoref{fig:trace-jit-ex}, a trace of the
loop $\textbf{while}\; e\; (x \leftarrow 0;\; \ldots)$ embedded in a larger
context.
Instead of a JIT that records instructions, assume only branch targets are
recorded.  For this example, suppose the two targets $  \mathsf{L\kern -0.25pt \scalebox{0.87}{$\mathsf{ body }$} }  $ and
$  \mathsf{L\kern -0.25pt \scalebox{0.87}{$\mathsf{ done }$} }  $ are recorded, which means the loop body executed once and
then exited.  In other words, the loop condition~$\ottnt{e}$ was $ \textsfb{true} $ the
first time and $ \textsfb{false} $ the second time.  The compiler could unroll the
loop twice and assert $\ottnt{e}$ for the first iteration and $  \neg    \ottnt{e} $ for the
second iteration (left).  Then  unreachable code elimination yields the
right hand side, resembling a trace.
\begin{table}[H]
  \vskip -4mm
  \begin{tabular}{@{}c@{}c}
$ {\small \begin{array}{lll}
                                                      &      & \dots \\  \,  &      &  \textsfb{assume} ~ \ottnt{e} ~ \ottkw{else} ~    \mathsf{F}  .\kern -0.5pt  \mathsf{V\kern -0.25pt \scalebox{0.87}{$\mathsf{ base }$} }  .\kern -0.5pt  \mathsf{L\kern -0.25pt \scalebox{0.87}{$\mathsf{ loop }$} }   ~  [{  \mathsf{ x }   \ottsym{=}   \mathsf{ x }   \ottsym{,}   \dots   \kern 0.033em}]    \\  \,  &      &   \textsfb{branch} ~ \ottnt{e} ~  \mathsf{L\kern -0.25pt \scalebox{0.87}{$\mathsf{ body_0 }$} }  ~  \mathsf{L\kern -0.25pt \scalebox{0.87}{$\mathsf{ done }$} }    \\  \,  &   \mathsf{L\kern -0.25pt \scalebox{0.87}{$\mathsf{ body_0 }$} }   &    \mathsf{ x }  \nobreak\hspace{0pt} \leftarrow \nobreak\hspace{0pt} 0   \\  \,  &      & \dots \\  \,  &      &  \textsfb{assume} ~   \neg    \ottnt{e}  ~ \ottkw{else} ~    \mathsf{F}  .\kern -0.5pt  \mathsf{V\kern -0.25pt \scalebox{0.87}{$\mathsf{ base }$} }  .\kern -0.5pt  \mathsf{L\kern -0.25pt \scalebox{0.87}{$\mathsf{ loop }$} }   ~  [{  \mathsf{ x }   \ottsym{=}   \mathsf{ x }   \ottsym{,}   \dots   \kern 0.033em}]    \\  \,  &      &   \textsfb{branch} ~ \ottnt{e} ~  \mathsf{L\kern -0.25pt \scalebox{0.87}{$\mathsf{ body_1 }$} }  ~  \mathsf{L\kern -0.25pt \scalebox{0.87}{$\mathsf{ done }$} }    \\  \,  &   \mathsf{L\kern -0.25pt \scalebox{0.87}{$\mathsf{ body_1 }$} }   &    \mathsf{ x }  \nobreak\hspace{0pt} \leftarrow \nobreak\hspace{0pt} 0   \\  \,  &      & \dots \\  \,  &      &  \textsfb{goto} \,  \mathsf{L\kern -0.25pt \scalebox{0.87}{$\mathsf{ loop }$} }   \\  \,  &   \mathsf{L\kern -0.25pt \scalebox{0.87}{$\mathsf{ done }$} }   & \dots \\  
                                                \end{array} } $
    &
$ {\small \begin{array}{lll}
                                                      &      & \dots \\  \,  &      &  \textsfb{assume} ~ \ottnt{e} ~ \ottkw{else} ~    \mathsf{F}  .\kern -0.5pt  \mathsf{V\kern -0.25pt \scalebox{0.87}{$\mathsf{ base }$} }  .\kern -0.5pt  \mathsf{L\kern -0.25pt \scalebox{0.87}{$\mathsf{ loop }$} }   ~  [{  \mathsf{ x }   \ottsym{=}   \mathsf{ x }   \ottsym{,}   \dots   \kern 0.033em}]    \\  \,  &      &    \mathsf{ x }  \nobreak\hspace{0pt} \leftarrow \nobreak\hspace{0pt} 0   \\  \,  &      & \dots \\  \,  &      &  \textsfb{assume} ~   \neg    \ottnt{e}  ~ \ottkw{else} ~    \mathsf{F}  .\kern -0.5pt  \mathsf{V\kern -0.25pt \scalebox{0.87}{$\mathsf{ base }$} }  .\kern -0.5pt  \mathsf{L\kern -0.25pt \scalebox{0.87}{$\mathsf{ loop }$} }   ~  [{  \mathsf{ x }   \ottsym{=}   \mathsf{ x }   \ottsym{,}   \dots   \kern 0.033em}]    \\  \,  &      & \dots \\  
                                                \end{array} } $
  \end{tabular}
  \vspace{-3mm}
\end{table} %
\noindent Say $  \mathsf{ x }  $ is not accessed after the store in this
optimized version.  In \sourir, it is obvious why dead store elimination of
$  \mathsf{ x }  $ would be unsound: the deoptimization metadata indicates that
$  \mathsf{ x }  $ is needed for deoptimization and the store operation can only
be removed it can be replayed.  In this specific example, a constant
propagation pass could update the metadata to materialize the write of $0$,
only when deoptimizing at the second \assume.  But, before the code can be
reduced, loop unrolling might result in intermediate versions that are much
larger than the original program.  In contrast, tracing JITs can handle this
case without the drastic expansion in code size~\citep{trace-franz09}, but
lose more information about instructions outside of the trace.

\section{Conclusions}
  \label{sec:conclusion}
Speculative optimizations are key to just-in-time optimization of dynamic
languages. As these optimizations depend on predicates about the program
state, the language implementation must monitor the validity of predicates
and be ready to deoptimize the program if a predicate is invalidated. While,
many modern compiler rely on this approach, the interplay between
optimization and deoptimization often remains opaque.

Our contribution is to show that when the predicates and the deoptimization
metadata are reified in the program representation, it becomes quite easy to
define correct program transformations that are deoptimization aware. In
this work we extend the intermediate representation with one
new instruction, \assume, which plays the double role of checking for the
validity of predicates and specifying the actions required to deoptimize the
program. Program transformations can inspect both the predicates that are
being monitored and the deoptimization metadata and transform them when
needed.  The formalization presented here is for one particular intermediate
language that we hope to be representative of a typical dynamic language.
We present a bisimulation proof between multiple versions of the same function, optimized under different assumptions.
We formalize deoptimization invariants between
versions and show that they enable very simple proofs
for standard compiler optimizations, constant folding, unreachable code
elimination, and function inlining. We also prove correct three
optimizations that are specifically dealing with deoptimizations, namely
unrestricted deoptimization, predicate hoisting, and assume composition.

There are multiple avenues of future investigation. The optimizations presented
here rely on intraprocedural analysis and the granularity of deoptimization
is a whole function.  If we were to extend this work to interprocedural
analysis, it would become much trickier to determine what functions are to
be invalidated as a speculation in one function may allow optimizations in
many other functions.  The current representation forces to check predicates
before each use, but some predicates are cheaper to check by monitoring
operations that could invalidate them. To do this would require changes to
our model as the \assume instruction would need to be split between a
monitor and a deoptimization point. Lastly, the expressive power of
predicates is an interesting question as there is a clear trade-off ---
richer predicates may allow more optimizations but are likely to be costlier
to monitor.

\begin{acks}
Jean-Marie Madiot provided guidance on the use and limitations of various
notions of bisimulation. In particular, he suggested adding a non-silent
\code{stop} transition to recover equi-termination from weak
bisimilarity. Francesco Zappa Nardelli helped with the motivation and
presentation of our work.  We thank \citet{sew07} for writing and maintaining
Ott.  Our work supported by the National Science Foundation under Grants
CCF--1544542, CCF--1318227, CCF--1618732, ONR award 503353, and the European
Research Council (ERC) under the European Union's Horizon 2020 research and
innovation program (grant agreement 695412). Any opinions, findings, and
conclusions expressed in this material may be those of the authors and
likely do not reflect the views of our funding agencies.
\end{acks}

\clearpage

\bibliography{main}

\end{document}

%% file: sourir.tex
\newcommand{\ottdrule}[4][]{{\displaystyle\frac{\begin{array}{l}#2\end{array}}{#3}\quad\ottdrulename{#4}}}
\newcommand{\ottusedrule}[1]{\[#1\]}
\newcommand{\ottpremise}[1]{ #1 \\}
\newenvironment{ottdefnblock}[3][]{ \framebox{\mbox{#2}} \quad #3 \\[0pt]}{}

\newcommand{\ottnt}[1]{\mathit{#1}}
\newcommand{\ottmv}[1]{\mathit{#1}}
\newcommand{\ottkw}[1]{\mathbf{#1}}
\newcommand{\ottsym}[1]{#1}
\newcommand{\ottcom}[1]{\text{#1}}
\newcommand{\ottdrulename}[1]{\textsc{#1}}

\newcommand{\ottrulehead}[3]{$#1$ & & $#2$ & & & \multicolumn{2}{l}{#3}}
\newcommand{\ottprodline}[6]{& & $#1$ & $#2$ & $#3 #4$ & $#5$ & $#6$}
\newcommand{\ottfirstprodline}[6]{\ottprodline{#1}{#2}{#3}{#4}{#5}{#6}}

\newcommand{\ottprodnewline}{\\}
\newcommand{\ottinterrule}{\\[5.0mm]}

\newcommand{\otti}{
\ottrulehead{\ottnt{i}}{::=}{\ottcom{instructions}}\ottprodnewline
\ottfirstprodline{|}{ \textsfb{var} ~ \mathit{x} \nobreak\hspace{0pt}=\nobreak\hspace{0pt} \ottnt{e} }{}{}{}{\ottcom{variable declaration}}\ottprodnewline
\ottprodline{|}{\textsfb{drop} \, \mathit{x}}{}{}{}{\ottcom{drop a variable from scope}}\ottprodnewline
\ottprodline{|}{ \mathit{x} \nobreak\hspace{0pt} \leftarrow \nobreak\hspace{0pt} \ottnt{e} }{}{}{}{\ottcom{assignment}}\ottprodnewline
\ottprodline{|}{ \textsfb{array} ~ \mathit{x} [{ \ottnt{e}  \kern 0.06em}] }{}{}{}{\ottcom{array allocation}}\ottprodnewline
\ottprodline{|}{ \textsfb{array} ~ \mathit{x} \nobreak\hspace{0pt}=\nobreak\hspace{0pt}[{ e^*  \kern 0.08em}] }{}{}{}{\ottcom{array creation}}\ottprodnewline
\ottprodline{|}{ \mathit{x} [{ \ottnt{e_{{\mathrm{1}}}}  \kern 0.05em}]\nobreak\hspace{0pt} \leftarrow \nobreak\hspace{0pt} \ottnt{e_{{\mathrm{2}}}} }{}{}{}{\ottcom{array assignment}}\ottprodnewline
\ottprodline{|}{ \textsfb{branch} ~ \ottnt{e} ~ L_{{\mathrm{1}}} ~ L_{{\mathrm{2}}} }{}{}{}{\ottcom{conditional branch}}\ottprodnewline
\ottprodline{|}{\textsfb{goto} \, L}{}{}{}{\ottcom{unconditional branch}}\ottprodnewline
\ottprodline{|}{\textsfb{print} \, \ottnt{e}}{}{}{}{\ottcom{print}}\ottprodnewline
\ottprodline{|}{\textsfb{read} \, \mathit{x}}{}{}{}{\ottcom{read}}\ottprodnewline
\ottprodline{|}{ \textsfb{call} ~ \mathit{x} \nobreak\hspace{0pt}=\nobreak\hspace{0pt} \ottnt{e} ( e^* ) }{}{}{}{\ottcom{function call}}\ottprodnewline
\ottprodline{|}{\textsfb{return} \, \ottnt{e}}{}{}{}{\ottcom{return}}\ottprodnewline
\ottprodline{|}{ \textsfb{assume} ~ e^* ~ \ottkw{else} ~ \xi ~ \tilde{\xi}^* }{}{}{}{\ottcom{assume instruction}}\ottprodnewline
\ottprodline{|}{\textsfb{stop}}{}{}{}{\ottcom{terminate execution}}}

\newcommand{\otte}{
\ottrulehead{\ottnt{e}}{::=}{\ottcom{expression}}\ottprodnewline
\ottfirstprodline{|}{\ottnt{se}}{}{}{}{\ottcom{simple expression}}\ottprodnewline
\ottprodline{|}{ \mathit{x} [{ \ottnt{e} }] }{}{}{}{\ottcom{array access}}\ottprodnewline
\ottprodline{|}{ \textsf{length} (  \ottnt{se}  \kern 0.04em ) }{}{}{}{\ottcom{array length}}\ottprodnewline
\ottprodline{|}{ \texttildelow   \ottnt{e} }{}{}{}{\ottcom{unary operation}}\ottprodnewline
\ottprodline{|}{ \ottnt{se_{{\mathrm{1}}}}   \star   \ottnt{se_{{\mathrm{2}}}} }{}{}{}{\ottcom{binary operation}}\ottprodnewline
\ottprodline{|}{ primop  (  se^*  ) }{}{}{}{\ottcom{arbitrary prmitive}}\ottprodnewline
\ottprodline{|}{\ottkw{box} \, \ottsym{(}  se^*  \ottsym{)}}{}{}{}{}}

\newcommand{\otteprint}{
\ottrulehead{\mathit{e}}{::=}{\ottcom{expression}}\ottprodnewline
\ottfirstprodline{|}{\ottnt{se}}{}{}{}{\ottcom{simple expression}}\ottprodnewline
\ottprodline{|}{ \mathit{x} [  \ottnt{se}  \kern 0.075em ] }{}{}{}{\ottcom{array access}}\ottprodnewline
\ottprodline{|}{ \textsf{length}  (  \ottnt{se}  ) }{}{}{}{\ottcom{array length}}\ottprodnewline
\ottprodline{|}{primop \, \ottsym{(}  se^*  \ottsym{)}}{}{}{}{\ottcom{primitive operation}}}

\newcommand{\ottse}{
\ottrulehead{\ottnt{se}}{::=}{\ottcom{simple expressions}}\ottprodnewline
\ottfirstprodline{|}{\ottnt{lit}}{}{}{}{\ottcom{literals}}\ottprodnewline
\ottprodline{|}{ { F } }{}{}{}{\ottcom{function reference}}\ottprodnewline
\ottprodline{|}{\mathit{x}}{}{}{}{\ottcom{variables}}}

\newcommand{\otttargetXXtriple}{
\ottrulehead{\ottnt{target\_triple}}{::=}{}\ottprodnewline
\ottfirstprodline{|}{ F .\kern -0.5pt V .\kern -0.5pt L }{}{}{}{}}

\newcommand{\otttopXXframeXXmap}{
\ottrulehead{\mathit{VA}}{::=}{}\ottprodnewline
\ottfirstprodline{|}{ [{ \ottnt{top\_frame\_map\_cont}  \kern 0.033em}] }{}{}{}{}\ottprodnewline
\ottprodline{|}{ \mathit{VA}'  \circ  \mathit{VA} }{}{}{}{}\ottprodnewline
\ottprodline{|}{ \mathsf{Id} }{}{}{}{}}

\newcommand{\ottlit}{
\ottrulehead{\ottnt{lit}}{::=}{\ottcom{literals}}\ottprodnewline
\ottfirstprodline{|}{\dots , -1 , 0 , 1 , \dots}{}{}{}{\ottcom{numbers}}\ottprodnewline
\ottprodline{|}{\textsf{\textbf{nil} }~|~\textsf{\textbf{true} }~|~\textsf{\textbf{false} }}{}{}{}{\ottcom{others}}}

\newcommand{\ottv}{
\ottrulehead{\ottnt{v}}{::=}{\ottcom{values}}\ottprodnewline
\ottfirstprodline{|}{\ottnt{lit}}{}{}{}{}\ottprodnewline
\ottprodline{|}{F}{}{}{}{}\ottprodnewline
\ottprodline{|}{ \mathit{a} }{}{}{}{}}

\newcommand{\ottfunref}{
\ottrulehead{F}{::=}{}\ottprodnewline
\ottfirstprodline{|}{ \mathsf{F} }{}{}{}{}\ottprodnewline
\ottprodline{|}{ \mathsf{F\kern -0.33pt \scalebox{0.83}{$\mathsf{ \ottmv{num} }$} } }{}{}{}{}\ottprodnewline
\ottprodline{|}{\ottmv{F}}{}{}{}{}\ottprodnewline
\ottprodline{|}{ \scalebox{0.97}{$\mathsf{ \ottmv{alphanum} }$} }{}{}{}{}\ottprodnewline
\ottprodline{|}{ fun(  \ottnt{v}  ) }{}{}{}{}\ottprodnewline
\ottprodline{|}{ \mathsf{main} }{}{}{}{}}

\newcommand{\ottversref}{
\ottrulehead{V}{::=}{}\ottprodnewline
\ottfirstprodline{|}{ \mathsf{V} }{}{}{}{}\ottprodnewline
\ottprodline{|}{ \mathsf{V\kern -0.5pt \scalebox{0.83}{$\mathsf{ \ottmv{num} }$} } }{}{}{}{}\ottprodnewline
\ottprodline{|}{ \mathsf{V\kern -0.25pt \scalebox{0.87}{$\mathsf{ \ottmv{alphanum} }$} } }{}{}{}{}\ottprodnewline
\ottprodline{|}{\ottmv{V}}{}{}{}{}\ottprodnewline
\ottprodline{|}{ \mathsf{V\kern -0.25pt \scalebox{0.87}{$\mathsf{ \ottmv{alphanum} }$} } }{}{}{}{}\ottprodnewline
\ottprodline{|}{ \mathsf{active} }{}{}{}{}}

\newcommand{\ottlabel}{
\ottrulehead{L}{::=}{}\ottprodnewline
\ottfirstprodline{|}{  }{}{}{}{}\ottprodnewline
\ottprodline{|}{ \phantom{ L } }{}{}{}{}\ottprodnewline
\ottprodline{|}{ \mathsf{L} }{}{}{}{}\ottprodnewline
\ottprodline{|}{ \mathsf{L\kern -0.25pt \scalebox{0.83}{$\mathsf{ \ottmv{num} }$} } }{}{}{}{}\ottprodnewline
\ottprodline{|}{ \mathsf{L\kern -0.25pt \scalebox{0.87}{$\mathsf{ \ottmv{alphanum} }$} } }{}{}{}{}\ottprodnewline
\ottprodline{|}{ L_{\dots} }{}{}{}{}\ottprodnewline
\ottprodline{|}{\ottmv{L}}{}{}{}{}\ottprodnewline
\ottprodline{|}{ \mathsf{L\kern -0.25pt \scalebox{0.87}{$\mathsf{ \ottmv{alphanum} }$} } }{}{}{}{}\ottprodnewline
\ottprodline{|}{ ( L \kern-1pt+\kern-2pt1) }{}{}{}{}\ottprodnewline
\ottprodline{|}{ \mathsf{start}( \ottnt{I} ) }{}{}{}{}}

\newcommand{\ottAct}{
\ottrulehead{\mathit{A}}{::=}{\ottcom{I/O action}}\ottprodnewline
\ottfirstprodline{|}{ \textsf{print}~ \ottnt{lit} }{}{}{}{}\ottprodnewline
\ottprodline{|}{ \textsf{read}~ \ottnt{lit} }{}{}{}{}\ottprodnewline
\ottprodline{|}{ \textsf{stop} }{}{}{}{}}

\newcommand{\ottActopt}{
\ottrulehead{\mathit{A}_\tau}{::=}{}\ottprodnewline
\ottfirstprodline{|}{\mathit{A}}{}{}{}{}\ottprodnewline
\ottprodline{|}{ \tau }{}{}{}{\ottcom{silent label}}}

\newcommand{\ottT}{
\ottrulehead{\mathit{T}}{::=}{\ottcom{action trace}}\ottprodnewline
\ottfirstprodline{|}{}{}{}{}{\ottcom{(empty trace)}}\ottprodnewline
\ottprodline{|}{ \mathit{A} }{}{}{}{}\ottprodnewline
\ottprodline{|}{ \mathit{A}_\tau }{}{}{}{}\ottprodnewline
\ottprodline{|}{ \mathit{T} ~ \mathit{A} }{}{}{}{}\ottprodnewline
\ottprodline{|}{ \mathit{T} ~ \mathit{A}_\tau }{}{}{}{}}

\newcommand{\ottalignedXXfundecl}{
\ottrulehead{\ottnt{aligned\_fundecl}}{::=}{}\ottprodnewline
\ottfirstprodline{|}{ F  (  x^*  \kern 0.04em ) \\
                                                                   \begin{array}{llll}
                                                                     \kern 1.5pt  \ottnt{aligned\_F} 
                                                                   \end{array} \\ }{}{}{}{}\ottprodnewline
\ottprodline{|}{ F  \dots }{}{}{}{}}

\newcommand{\ottdruleLiteral}[1]{\ottdrule[#1]{%
}{
 \ottnt{E} ~ \ottnt{lit}   \rightharpoonup   \ottnt{lit} }{%
{\ottdrulename{Literal}}{}%
}}

\newcommand{\ottdruleFunref}[1]{\ottdrule[#1]{%
}{
 \ottnt{E} ~  { F }    \rightharpoonup   F }{%
{\ottdrulename{Funref}}{}%
}}

\newcommand{\ottdruleLookup}[1]{\ottdrule[#1]{%
}{
 \ottnt{E} ~ \mathit{x}   \rightharpoonup    \ottnt{E}  (  \mathit{x}  )  }{%
{\ottdrulename{Lookup}}{}%
}}

\newcommand{\ottdruleSimpleExp}[1]{\ottdrule[#1]{%
\ottpremise{ \ottnt{E} ~ \ottnt{se}   \rightharpoonup   \ottnt{v} }%
}{
 \ottnt{M} ~ \ottnt{E} ~ \ottnt{se}   \rightarrow   \ottnt{v} }{%
{\ottdrulename{SimpleExp}}{}%
}}

\newcommand{\ottdruleVecLen}[1]{\ottdrule[#1]{%
\ottpremise{ \ottnt{E} ~ \ottnt{se}   \rightharpoonup    \mathit{a}   \quad  \ottnt{M}  (  \mathit{a}  )   \ottsym{=}  \ottsym{[}  \ottnt{v_{{\mathrm{1}}}}  \ottsym{,} \, .. \, \ottsym{,}  \ottnt{v_{\ottmv{n}}}  \ottsym{]}}%
}{
 \ottnt{M} ~ \ottnt{E} ~  \textsf{length} (  \ottnt{se}  \kern 0.04em )    \rightarrow    \ottmv{n}  }{%
{\ottdrulename{VecLen}}{}%
}}

\newcommand{\ottdruleVecAccess}[1]{\ottdrule[#1]{%
\ottpremise{ \mathit{a}   \mathrel{\stackrel{\mathsf{def} }{=} }   \ottnt{E}  (  \mathit{x}  )  \quad  \ottnt{M}  (  \mathit{a}  )   \ottsym{=}  \ottsym{[}  \ottnt{v_{{\mathrm{0}}}}  \ottsym{,} \, .. \, \ottsym{,}  \ottnt{v_{\ottmv{m}}}  \ottsym{]}}%
\ottpremise{ \ottnt{E} ~ \ottnt{se}   \rightharpoonup    \ottmv{n}   \quad  0 \leq  \ottmv{n}  \leq  \ottmv{m} }%
}{
 \ottnt{M} ~ \ottnt{E} ~  \mathit{x} [{ \ottnt{se} }]    \rightarrow   \ottnt{v_{\ottmv{n}}} }{%
{\ottdrulename{VecAccess}}{}%
}}

\newcommand{\ottdrulePrimop}[1]{\ottdrule[#1]{%
\ottpremise{ \ottnt{E} ~ \ottnt{se_{{\mathrm{1}}}}   \rightharpoonup   \ottnt{v_{{\mathrm{1}}}}  \quad .. \quad  \ottnt{E} ~ \ottnt{se_{\ottmv{n}}}   \rightharpoonup   \ottnt{v_{\ottmv{n}}} }%
}{
 \ottnt{M} ~ \ottnt{E} ~  primop  (  \ottnt{se_{{\mathrm{1}}}}  \ottsym{,} \, .. \, \ottsym{,}  \ottnt{se_{\ottmv{n}}}  )    \rightarrow   [\![ primop ]\!]  \ottsym{(}  \ottnt{v_{{\mathrm{1}}}}  \ottsym{,} \, .. \, \ottsym{,}  \ottnt{v_{\ottmv{n}}}  \ottsym{)} }{%
{\ottdrulename{Primop}}{}%
}}

\newcommand{\ottdruleDecl}[1]{\ottdrule[#1]{%
\ottpremise{  \ottnt{I} ( L )  =~   \textsfb{var} ~ \mathit{x} \nobreak\hspace{0pt}=\nobreak\hspace{0pt} \ottnt{e}    \quad  \ottnt{M} ~ \ottnt{E} ~ \ottnt{e}   \rightarrow   \ottnt{v} }%
}{
  \langle  \ottnt{P} \, \ottnt{I} \, L \, K^* \; \ottnt{M} \, \ottnt{E}  \rangle   \nto[   \tau   ]   \langle  \ottnt{P} \, \ottnt{I} \,  ( L \kern-1pt+\kern-2pt1)  \, K^* \; \ottnt{M} \,  \ottnt{E}  [  \mathit{x}  \leftarrow  \ottnt{v}  ]   \rangle  }{%
{\ottdrulename{Decl}}{}%
}}

\newcommand{\ottdruleUpdate}[1]{\ottdrule[#1]{%
\ottpremise{  \ottnt{I} ( L )  =~   \mathit{x} \nobreak\hspace{0pt} \leftarrow \nobreak\hspace{0pt} \ottnt{e}    \quad   \mathsf{ x }   \in dom( \ottnt{E} )  \quad  \ottnt{M} ~ \ottnt{E} ~ \ottnt{e}   \rightarrow   \ottnt{v} }%
}{
  \langle  \ottnt{P} \, \ottnt{I} \, L \, K^* \; \ottnt{M} \, \ottnt{E}  \rangle   \nto[   \tau   ]   \langle  \ottnt{P} \, \ottnt{I} \,  ( L \kern-1pt+\kern-2pt1)  \, K^* \; \ottnt{M} \,  \ottnt{E}  [  \mathit{x}  \leftarrow  \ottnt{v}  ]   \rangle  }{%
{\ottdrulename{Update}}{}%
}}

\newcommand{\ottdruleDrop}[1]{\ottdrule[#1]{%
\ottpremise{  \ottnt{I} ( L )  =~  \textsfb{drop} \, \mathit{x}  }%
}{
  \langle  \ottnt{P} \, \ottnt{I} \, L \, K^* \; \ottnt{M} \, \ottnt{E}  \rangle   \nto[   \tau   ]   \langle  \ottnt{P} \, \ottnt{I} \,  ( L \kern-1pt+\kern-2pt1)  \, K^* \; \ottnt{M} \,  \ottnt{E} \mathord{\setminus}\{  \mathit{x}  \}   \rangle  }{%
{\ottdrulename{Drop}}{}%
}}

\newcommand{\ottdruleArrayDef}[1]{\ottdrule[#1]{%
\ottpremise{  \ottnt{I} ( L )  =~   \textsfb{array} ~ \mathit{x} \nobreak\hspace{0pt}=\nobreak\hspace{0pt}[{ \ottnt{e_{{\mathrm{1}}}}  \ottsym{,} \, .. \, \ottsym{,}  \ottnt{e_{\ottmv{n}}}  \kern 0.08em}]    \quad  \ottnt{M} ~ \ottnt{E} ~ \ottnt{e_{{\mathrm{1}}}}   \rightarrow   \ottnt{v_{{\mathrm{1}}}}  \quad .. \quad  \ottnt{M} ~ \ottnt{E} ~ \ottnt{e_{\ottmv{n}}}   \rightarrow   \ottnt{v_{\ottmv{n}}} }%
\ottpremise{ \mathit{a} ~\text{fresh}  \quad \ottnt{M'}  \mathrel{\stackrel{\mathsf{def} }{=} }   \ottnt{M}  [  \mathit{a}  \leftarrow  \ottsym{[}  \ottnt{v_{{\mathrm{1}}}}  \ottsym{,} \, .. \, \ottsym{,}  \ottnt{v_{\ottmv{n}}}  \ottsym{]}  ] }%
}{
  \langle  \ottnt{P} \, \ottnt{I} \, L \, K^* \; \ottnt{M} \, \ottnt{E}  \rangle   \nto[   \tau   ]   \langle  \ottnt{P} \, \ottnt{I} \,  ( L \kern-1pt+\kern-2pt1)  \, K^* \; \ottnt{M'} \,  \ottnt{E}  [  \mathit{x}  \leftarrow   \scalebox{0.97}{$\mathsf{ a }$}   ]   \rangle  }{%
{\ottdrulename{ArrayDef}}{}%
}}

\newcommand{\ottdruleArrayDecl}[1]{\ottdrule[#1]{%
\ottpremise{  \ottnt{I} ( L )  =~   \textsfb{array} ~ \mathit{x} [{ \ottnt{e}  \kern 0.06em}]    \quad  \ottnt{M} ~ \ottnt{E} ~ \ottnt{e}   \rightarrow    n  }%
\ottpremise{ \mathit{a} ~\text{fresh}  \quad \ottnt{M'}  \mathrel{\stackrel{\mathsf{def} }{=} }   \ottnt{M}  [  \mathit{a}  \leftarrow   [  \textsfb{nil} _1, ..,  \textsfb{nil} _n ]   ] }%
}{
  \langle  \ottnt{P} \, \ottnt{I} \, L \, K^* \; \ottnt{M} \, \ottnt{E}  \rangle   \nto[   \tau   ]   \langle  \ottnt{P} \, \ottnt{I} \,  ( L \kern-1pt+\kern-2pt1)  \, K^* \; \ottnt{M'} \,  \ottnt{E}  [  \mathit{x}  \leftarrow   \scalebox{0.97}{$\mathsf{ a }$}   ]   \rangle  }{%
{\ottdrulename{ArrayDecl}}{}%
}}

\newcommand{\ottdruleArrayUpdate}[1]{\ottdrule[#1]{%
\ottpremise{  \ottnt{I} ( L )  =~   \mathit{x} [{ \ottnt{e'}  \kern 0.05em}]\nobreak\hspace{0pt} \leftarrow \nobreak\hspace{0pt} \ottnt{e}    \quad  \scalebox{0.97}{$\mathsf{ a }$}   \mathrel{\stackrel{\mathsf{def} }{=} }   \ottnt{E}  (  \mathit{x}  )  \quad  \ottnt{M} ~ \ottnt{E} ~ \ottnt{e'}   \rightarrow    \ottmv{n}   \quad  \ottnt{M} ~ \ottnt{E} ~ \ottnt{e}   \rightarrow   \ottnt{v} }%
\ottpremise{ \ottnt{M}  (  \mathit{a}  )   \ottsym{=}  \ottsym{[}  \ottnt{v_{{\mathrm{0}}}}  \ottsym{,} \, .. \, \ottsym{,}  \ottnt{v_{\ottmv{m}}}  \ottsym{]} \quad  0 \leq  \ottmv{n}  \leq  \ottmv{m} }%
\ottpremise{\ottnt{M'}  \mathrel{\stackrel{\mathsf{def} }{=} }   \ottnt{M}  [  \mathit{a}  \leftarrow  \ottsym{[}  \ottnt{v_{{\mathrm{0}}}}  \ottsym{,} \, .. \, \ottsym{,}  \ottnt{v_{\ottmv{m}}}  \ottsym{]}  \ottsym{\{}  \ottnt{v_{\ottmv{n}}}  \ottsym{/}  \ottnt{v}  \ottsym{\}}  ] }%
}{
  \langle  \ottnt{P} \, \ottnt{I} \, L \, K^* \; \ottnt{M} \, \ottnt{E}  \rangle   \nto[   \tau   ]   \langle  \ottnt{P} \, \ottnt{I} \,  ( L \kern-1pt+\kern-2pt1)  \, K^* \; \ottnt{M'} \, \ottnt{E}  \rangle  }{%
{\ottdrulename{ArrayUpdate}}{}%
}}

\newcommand{\ottdrulePrint}[1]{\ottdrule[#1]{%
\ottpremise{  \ottnt{I} ( L )  =~  \textsfb{print} \, \ottnt{e}   \quad  \ottnt{M} ~ \ottnt{E} ~ \ottnt{e}   \rightarrow   \ottnt{lit} }%
}{
  \langle  \ottnt{P} \, \ottnt{I} \, L \, K^* \; \ottnt{M} \, \ottnt{E}  \rangle   \nto[   \textsf{print}~ \ottnt{lit}   ]   \langle  \ottnt{P} \, \ottnt{I} \,  ( L \kern-1pt+\kern-2pt1)  \, K^* \; \ottnt{M} \, \ottnt{E}  \rangle  }{%
{\ottdrulename{Print}}{}%
}}

\newcommand{\ottdruleRead}[1]{\ottdrule[#1]{%
\ottpremise{  \ottnt{I} ( L )  =~  \textsfb{read} \, \mathit{x}  }%
}{
  \langle  \ottnt{P} \, \ottnt{I} \, L \, K^* \; \ottnt{M} \, \ottnt{E}  \rangle   \nto[   \textsf{read}~ \ottnt{lit}   ]   \langle  \ottnt{P} \, \ottnt{I} \,  ( L \kern-1pt+\kern-2pt1)  \, K^* \; \ottnt{M} \,  \ottnt{E}  [  \mathit{x}  \leftarrow  \ottnt{lit}  ]   \rangle  }{%
{\ottdrulename{Read}}{}%
}}

\newcommand{\ottdruleGoto}[1]{\ottdrule[#1]{%
\ottpremise{  \ottnt{I} ( L )  =~  \textsfb{goto} \, L'  }%
}{
  \langle  \ottnt{P} \, \ottnt{I} \, L \, K^* \; \ottnt{M} \, \ottnt{E}  \rangle   \nto[   \tau   ]   \langle  \ottnt{P} \, \ottnt{I} \, L' \, K^* \; \ottnt{M} \, \ottnt{E}  \rangle  }{%
{\ottdrulename{Goto}}{}%
}}

\newcommand{\ottdruleBranchT}[1]{\ottdrule[#1]{%
\ottpremise{  \ottnt{I} ( L )  =~   \textsfb{branch} ~ \ottnt{e} ~ L_{{\mathrm{1}}} ~ L_{{\mathrm{2}}}    \quad  \ottnt{M} ~ \ottnt{E} ~ \ottnt{e}   \rightarrow   \textsfb{true} }%
}{
  \langle  \ottnt{P} \, \ottnt{I} \, L \, K^* \; \ottnt{M} \, \ottnt{E}  \rangle   \nto[   \tau   ]   \langle  \ottnt{P} \, \ottnt{I} \, L_{{\mathrm{1}}} \, K^* \; \ottnt{M} \, \ottnt{E}  \rangle  }{%
{\ottdrulename{BranchT}}{}%
}}

\newcommand{\ottdruleBranchF}[1]{\ottdrule[#1]{%
\ottpremise{  \ottnt{I} ( L )  =~   \textsfb{branch} ~ \ottnt{e} ~ L_{{\mathrm{1}}} ~ L_{{\mathrm{2}}}    \quad  \ottnt{M} ~ \ottnt{E} ~ \ottnt{e}   \rightarrow   \textsfb{false} }%
}{
  \langle  \ottnt{P} \, \ottnt{I} \, L \, K^* \; \ottnt{M} \, \ottnt{E}  \rangle   \nto[   \tau   ]   \langle  \ottnt{P} \, \ottnt{I} \, L_{{\mathrm{2}}} \, K^* \; \ottnt{M} \, \ottnt{E}  \rangle  }{%
{\ottdrulename{BranchF}}{}%
}}

\newcommand{\ottdruleCall}[1]{\ottdrule[#1]{%
\ottpremise{  \ottnt{I} ( L )  =~   \textsfb{call} ~ \mathit{x} \nobreak\hspace{0pt}=\nobreak\hspace{0pt} \ottnt{e} ( \ottnt{e_{{\mathrm{1}}}}  \ottsym{,} \, .. \, \ottsym{,}  \ottnt{e_{\ottmv{n}}} )   }%
\ottpremise{ \ottnt{M} ~ \ottnt{E} ~ \ottnt{e}   \rightarrow   F }%
\ottpremise{ \ottnt{P} ( F )   \ottsym{=}   F (  \mathit{x}_{{\mathrm{1}}}  \ottsym{,} \, .. \, \ottsym{,}  \mathit{x}_{\ottmv{n}}  ) :  \mathit{D_F}  \quad \ottnt{I'}  \mathrel{\stackrel{\mathsf{def} }{=} }   \ottnt{P} ( F ,   \mathsf{active}  ) }%
\ottpremise{L'  \mathrel{\stackrel{\mathsf{def} }{=} }   \mathsf{start}( \ottnt{I'} )  \quad \ottnt{M} \, \ottnt{E} \,  [{ \mathit{x}_{{\mathrm{1}}}  \ottsym{=}  \ottnt{e_{{\mathrm{1}}}}  \ottsym{,} \, .. \, \ottsym{,}  \mathit{x}_{\ottmv{n}}  \ottsym{=}  \ottnt{e_{\ottmv{n}}}  \kern 0.033em}]   \rightsquigarrow  \ottnt{E'}}%
}{
  \langle  \ottnt{P} \, \ottnt{I} \, L \, K^* \; \ottnt{M} \, \ottnt{E}  \rangle   \nto[   \tau   ]   \langle  \ottnt{P} \, \ottnt{I'} \, L' \,  (K^*,   \langle  \ottnt{I} ~  ( L \kern-1pt+\kern-2pt1)  ~ \mathit{x} ~ \ottnt{E}  \rangle  )  \; \ottnt{M} \, \ottnt{E'}  \rangle  }{%
{\ottdrulename{Call}}{}%
}}

\newcommand{\ottdruleReturn}[1]{\ottdrule[#1]{%
\ottpremise{  \ottnt{I} ( L )  =~  \textsfb{return} \, \ottnt{e}   \quad  \ottnt{M} ~ \ottnt{E} ~ \ottnt{e}   \rightarrow   \ottnt{v} }%
}{
  \langle  \ottnt{P} \, \ottnt{I} \, L \,  (K^*,   \langle  \ottnt{I'} ~ L' ~ \mathit{x} ~ \ottnt{E'}  \rangle  )  \; \ottnt{M} \, \ottnt{E}  \rangle   \nto[   \tau   ]   \langle  \ottnt{P} \, \ottnt{I'} \, L' \, K^* \; \ottnt{M} \,  \ottnt{E'}  [  \mathit{x}  \leftarrow  \ottnt{v}  ]   \rangle  }{%
{\ottdrulename{Return}}{}%
}}

\newcommand{\ottdruleStop}[1]{\ottdrule[#1]{%
\ottpremise{  \ottnt{I} ( L )  =~  \textsfb{stop}  }%
}{
  \langle  \ottnt{P} \, \ottnt{I} \, L \, K^* \; \ottnt{M} \, \ottnt{E}  \rangle   \nto[   \textsf{stop}   ]   \langle  \ottnt{P} \,  \emptyset  \, L \, K^* \; \ottnt{M} \, \ottnt{E}  \rangle  }{%
{\ottdrulename{Stop}}{}%
}}

\newcommand{\ottdruleAssumePass}[1]{\ottdrule[#1]{%
\ottpremise{  \ottnt{I} ( L )  =~   \textsfb{assume} ~ e^* ~ \ottkw{else} ~ \xi ~ \tilde{\xi}^*    \quad  \forall  \ottmv{m}  ,   \ottnt{M} ~ \ottnt{E} ~ \ottnt{e_{\ottmv{m}}}   \rightarrow   \textsfb{true}  }%
}{
  \langle  \ottnt{P} \, \ottnt{I} \, L \, K^* \; \ottnt{M} \, \ottnt{E}  \rangle   \nto[   \tau   ]   \langle  \ottnt{P} \, \ottnt{I} \,  ( L \kern-1pt+\kern-2pt1)  \, K^* \; \ottnt{M} \, \ottnt{E}  \rangle  }{%
{\ottdrulename{AssumePass}}{}%
}}

\newcommand{\ottdruleAssumeDeopt}[1]{\ottdrule[#1]{%
\ottpremise{  \ottnt{I} ( L )  =~   \textsfb{assume} ~ e^* ~ \ottkw{else} ~ \xi ~ \tilde{\xi}^*    \quad  \neg (\forall  \ottmv{m}  ,   \ottnt{M} ~ \ottnt{E} ~ \ottnt{e_{\ottmv{m}}}   \rightarrow   \textsfb{true}  ) }%
}{
  \langle  \ottnt{P} \, \ottnt{I} \, L \, K^* \; \ottnt{M} \, \ottnt{E}  \rangle   \nto[   \tau   ]   \mathsf{deoptimize}(   \langle  \ottnt{P} \, \ottnt{I} \, L \, K^* \; \ottnt{M} \, \ottnt{E}  \rangle  ,  \xi ,  \tilde{\xi}^* )  }{%
{\ottdrulename{AssumeDeopt}}{}%
}}

\newcommand{\ottdruleAssumeDeoptSimple}[1]{\ottdrule[#1]{%
\ottpremise{  \ottnt{I} ( L )  =~   \textsfb{assume} ~ \ottnt{e_{{\mathrm{1}}}}  \ottsym{,} \, .. \, \ottsym{,}  \ottnt{e_{\ottmv{n}}} ~ \ottkw{else} ~   F' .\kern -0.5pt V' .\kern -0.5pt L'  ~ \mathit{VA}  ~ \,   }%
\ottpremise{ \exists  \ottmv{m}  ,   \ottnt{M} ~ \ottnt{E} ~ \ottnt{e_{\ottmv{m}}}   \rightarrow   \textsfb{false}  }%
\ottpremise{\ottnt{M} \, \ottnt{E} \, \mathit{VA}  \rightsquigarrow  \ottnt{E'}}%
}{
  \langle  \ottnt{P} \, \ottnt{I} \, L \, K^* \; \ottnt{M} \, \ottnt{E}  \rangle   \nto[   \tau   ]   \langle  \ottnt{P} \,  \ottnt{P} ( F' ,  V' )  \, L' \, K^* \; \ottnt{M} \, \ottnt{E'}  \rangle  }{%
{\ottdrulename{AssumeDeoptSimple}}{}%
}}

\newcommand{\ottdruleRefl}[1]{\ottdrule[#1]{%
}{
 \ottnt{C}  \tto[    ]  \ottnt{C} }{%
{\ottdrulename{Refl}}{}%
}}

\newcommand{\ottdruleSilentCons}[1]{\ottdrule[#1]{%
\ottpremise{ \ottnt{C}  \tto[  \mathit{T}  ]  \ottnt{C'}  \quad  \ottnt{C'}  \nto[   \tau   ]  \ottnt{C''} }%
}{
 \ottnt{C}  \tto[  \mathit{T}  ]  \ottnt{C''} }{%
{\ottdrulename{SilentCons}}{}%
}}

\newcommand{\ottdruleActionCons}[1]{\ottdrule[#1]{%
\ottpremise{ \ottnt{C}  \tto[  \mathit{T}  ]  \ottnt{C'}  \quad  \ottnt{C'}  \nto[  \mathit{A}  ]  \ottnt{C''} }%
}{
 \ottnt{C}  \tto[   \mathit{T} ~ \mathit{A}   ]  \ottnt{C''} }{%
{\ottdrulename{ActionCons}}{}%
}}

\newcommand{\ottdruleStartConf}[1]{\ottdrule[#1]{%
\ottpremise{\ottnt{I}  \mathrel{\stackrel{\mathsf{def} }{=} }   \ottnt{P} (  \mathsf{main}  ,   \mathsf{active}  )  \quad L  \mathrel{\stackrel{\mathsf{def} }{=} }   \mathsf{start}( \ottnt{I} ) }%
}{
 \mathsf{start}(  \ottnt{P}  )  \mathrel{\stackrel{\mathsf{def} }{=} }    \langle  \ottnt{P} \, \ottnt{I} \, L \,  \emptyset  \;  \emptyset  \,  \emptyset   \rangle  }{%
{\ottdrulename{StartConf}}{}%
}}

\newcommand{\ottdruleDeoptimizeConf}[1]{\ottdrule[#1]{%
\ottpremise{\ottnt{M} \, \ottnt{E} \, \mathit{VA}  \rightsquigarrow  \ottnt{E'} \quad \ottnt{I'}  \mathrel{\stackrel{\mathsf{def} }{=} }   \ottnt{P} ( F' ,  V' ) }%
\ottpremise{ \begin{matrix*}[l] \forall  \ottmv{q}  & \hspace{-3mm} \in  1 ,.., \ottmv{r}  , \\ &   \tilde{\xi}_{\ottmv{q}}  \ottsym{=}    F_{\ottmv{q}} .\kern -0.5pt V_{\ottmv{q}} .\kern -0.5pt L_{\ottmv{q}}  ~ \mathit{x}_{\ottmv{q}} ~ \mathit{VA}_{\ottmv{q}}   \\ &    \ottnt{M} \, \ottnt{E} \, \mathit{VA}_{\ottmv{q}}  \rightsquigarrow  \ottnt{E_{\ottmv{q}}}  \quad   \ottnt{I_{\ottmv{q}}}  \mathrel{\stackrel{\mathsf{def} }{=} }   \ottnt{P} ( F_{\ottmv{q}} ,  V_{\ottmv{q}} )   \quad  \ottnt{K_{\ottmv{q}}}  \mathrel{\stackrel{\mathsf{def} }{=} }   \langle  \ottnt{I_{\ottmv{q}}} ~ L_{\ottmv{q}} ~ \mathit{x}_{\ottmv{q}} ~ \ottnt{E_{\ottmv{q}}}  \rangle       \end{matrix*} }%
}{
 \mathsf{deoptimize}(   \langle  \ottnt{P} \, \ottnt{I} \, L \, K^* \; \ottnt{M} \, \ottnt{E}  \rangle  ,    F' .\kern -0.5pt V' .\kern -0.5pt L'  ~ \mathit{VA}  ,  \tilde{\xi}_{{\mathrm{1}}}  \ottsym{,} \, .. \, \ottsym{,}  \tilde{\xi}_{\ottmv{r}} )  \mathrel{\stackrel{\mathsf{def} }{=} }    \langle  \ottnt{P} \, \ottnt{I'} \, L' \,  (K^*,  \ottnt{K_{{\mathrm{1}}}}  \ottsym{,} \, .. \, \ottsym{,}  \ottnt{K_{\ottmv{r}}} )  \; \ottnt{M} \, \ottnt{E'}  \rangle  }{%
{\ottdrulename{DeoptimizeConf}}{}%
}}

\newcommand{\ottdruleEvalEnv}[1]{\ottdrule[#1]{%
\ottpremise{ \ottnt{M} ~ \ottnt{E} ~ \ottnt{e_{{\mathrm{1}}}}   \rightarrow   \ottnt{v_{{\mathrm{1}}}}  \quad .. \quad  \ottnt{M} ~ \ottnt{E} ~ \ottnt{e_{\ottmv{n}}}   \rightarrow   \ottnt{v_{\ottmv{n}}} }%
}{
\ottnt{M} \, \ottnt{E} \,  [{ \mathit{x}_{{\mathrm{1}}}  \ottsym{=}  \ottnt{e_{{\mathrm{1}}}}  \ottsym{,} \, .. \, \ottsym{,}  \mathit{x}_{\ottmv{n}}  \ottsym{=}  \ottnt{e_{\ottmv{n}}}  \kern 0.033em}]   \rightsquigarrow   [  \mathit{x}_{{\mathrm{1}}}  \rightarrow  \ottnt{v_{{\mathrm{1}}}}  \ottsym{,} \, .. \, \ottsym{,}  \mathit{x}_{\ottmv{n}}  \rightarrow  \ottnt{v_{\ottmv{n}}}  ] }{%
{\ottdrulename{EvalEnv}}{}%
}}